\documentclass[journal,onecolumn]{IEEEtran}
\usepackage{amsmath,amsthm,amssymb,mathrsfs,mathtools,bbm}

\usepackage[shortlabels]{enumitem}
\usepackage{color,xcolor,framed,tcolorbox,graphicx}

\theoremstyle{plain}
    \newtheorem{theorem}{Theorem}
    \newtheorem{proposition}{Proposition} 
    \newtheorem{corollary}[theorem]{Corollary}
    \newtheorem{claim}{Claim}
    \newtheorem{lemma}[theorem]{Lemma}

    \newtheorem{fact}{Fact}
\theoremstyle{definition}
    \newtheorem{definition}{Definition}
    \newtheorem{example}{\bf Example}
    \newtheorem{remark}{Remark}

\newcommand{\Conv}{\operatorname{Conv}}
\newcommand{\D}{\mathsf{D}}
\renewcommand{\vec}[1]{\boldsymbol{#1}}
\newcommand{\Bern}{\operatorname{Bern}}

\newcommand{\spec}{\operatorname{spec}}

\newcommand{\HPhi}{H_{\Phi}}

\newcommand{\DPhi}{D_{\Phi}}

\newcommand{\IPhi}{I_{\Phi}}

\newcommand{\supp}{\operatorname{supp}}
\newcommand{\dom}{\operatorname{dom}}
\newcommand{\Unif}{\operatorname{Unif}}

\newcommand{\E}{\mathbb{E}}
\newcommand{\F}{\mathbb{F}}

\newcommand{\R}{\mathbb{R}}

\newcommand{\fR}{\mathfrak{R}}

\newcommand\independent{\protect\mathpalette{\protect\independenT}{\perp}}
\def\independenT#1#2{\mathrel{\rlap{$#1#2$}\mkern2mu{#1#2}}}
\newcommand{\defeq}{\triangleq}

\newcommand{\Var}{\operatorname{Var}}

\makeatletter
\def\rv{\@ifnextchar\bgroup{\rv@i}{\rv@ii}}
\def\rv@i#1{\mathbf{#1}} 
\def\rv@ii#1{\mathbf{#1}} 
\makeatother

\newcommand{\DSBS}{\operatorname{DSBS}}

\newcommand{\diag}{\operatorname{diag}}

\newcommand{\sa}{\mathbb{M}_{\operatorname{sa}}}
\newcommand{\PSD}{\mathbb{M}_{+}}

\makeatletter
\def\mt{\@ifnextchar\bgroup{\mt@i}{\mt@ii}}
\def\mt@i#1{\boldsymbol{#1}}
\def\mt@ii#1{\boldsymbol{#1}}
\makeatother

\newcommand{\tr}{\operatorname{tr}}
\newcommand{\nortr}{\overline{\operatorname{tr}}}

\newcommand{\cD}{\mathcal{D}}
\newcommand{\cE}{\mathcal{E}}

\newcommand{\cH}{\mathcal{H}}
\newcommand{\cI}{\mathcal{I}}

\newcommand{\cM}{\mathcal{M}}

\newcommand{\cP}{\mathcal{P}}

\newcommand{\cU}{\mathcal{U}}
\newcommand{\cV}{\mathcal{V}}

\newcommand{\cX}{\mathcal{X}}
\newcommand{\cY}{\mathcal{Y}}

\newcommand{\wOmega}{\widetilde{\Omega}}
\newcommand{\wPi}{\widetilde{\Pi}}

\newcommand{\wA}{\widetilde{A}}
\newcommand{\wB}{\widetilde{B}}

\newcommand{\wX}{\widetilde{X}}
\newcommand{\wY}{\widetilde{Y}}

\newcommand{\wT}{\widetilde{T}}
\newcommand{\wS}{\widetilde{S}}

\newcommand{\An}{A_{[n]}}
\newcommand{\Xn}{X_{[n]}}
\newcommand{\Pin}{\Pi_{[n]}}
\newcommand{\Bn}{B_{[n]}}
\newcommand{\Yn}{Y_{[n]}}
\newcommand{\Omegan}{\Omega_{[n]}}

\newcommand{\wAn}{\wA_{[n]}}
\newcommand{\wXn}{\wX_{[n]}}
\newcommand{\wPin}{\wPi_{[n]}}
\newcommand{\wBn}{\wB_{[n]}}
\newcommand{\wYn}{\wY_{[n]}}
\newcommand{\wOmegan}{\wOmega_{[n]}}

\newcommand{\wYi}{\wY_{[i]}}
\newcommand{\wOmegai}{\wOmega_{[i]}}

\newcommand{\Tie}{T_i^e}
\newcommand{\Sie}{S_i^e}

\newcommand{\wTie}{\widetilde{T}_i^e}
\newcommand{\wSie}{\widetilde{S}_i^e}

\usepackage[margin=1in]{geometry}
\usepackage{float}
\usepackage{hyperref}
\usepackage{booktabs}
\usepackage{tabularx}
\usepackage{array}
\usepackage{centernot}
\usepackage{xcolor}
\usepackage{thmtools,thm-restate}

\definecolor{DarkRed}{rgb}{0.8, 0.0, 0.0} 
\definecolor{DarkGreen}{rgb}{0.0, 0.7, 0.0}
\hypersetup{
    colorlinks=true,
    linkcolor=DarkRed,
    citecolor=DarkGreen
}

\newcommand{\rank}{\operatorname{rank}}

\newcommand{\interior}{\operatorname{int}}

\newif\ifproofreading
\proofreadingfalse

\ifproofreading
    \newenvironment{unread}
        {\color{red}}
        {}
    \newenvironment{finalized}
        {\color{green!50!black}}
        {}
\else

\fi

\begin{document}
\title{Matrix $\Phi$-Entropic Measures of Dependence
\thanks{
This work was supported by the Hong Kong Research Grants Council (RGC) under grant number 14310025, and by the 1+1+1 CUHK–CUHK(SZ)–GDST Joint Collaboration Fund.}
}

\author{Chenyu Wang, Amin Gohari\\
\emph{\small Department of Information Engineering}\\
\emph{\small The Chinese University of Hong Kong}\\
{\small Email: \{wc025, agohari\}@ie.cuhk.edu.hk}
}

\maketitle

\begin{abstract}
We introduce matrix $\Phi$-ribbons, extending the scalar $\Phi$-ribbons to the matrix case. For $\Phi$ in the subadditivity class, we establish data processing and tensorization properties of matrix $\Phi$-ribbons. We establish scalar and matrix $\Phi$-entropy inequalities and use them to show that both $\Phi$-ribbons characterize the corresponding $\Phi$-strong data processing inequality (SDPI) constants. We obtain exact evaluations of the SDPI constants for several sources. We provide examples showing that the matrix $\Phi$-SDPI constant can exceed the scalar counterpart. Finally, we derive inner bounds for $\Phi$-ribbons from prescribed independence structures and prove that these bounds are optimal when $\Phi(t)=t\log t$.\footnote{This work was presented in part at the IEEE Symposium on Information Theory 2025 \cite{wangphient} and 2026 \cite{wangribbons}.} 

\end{abstract}

\section{Introduction}

A dependence measure is said to have the \emph{tensorization} property if its value remains unchanged under i.i.d. repetition of the underlying joint distribution. Such measures have applications in noninteractive distribution simulation \cite{kamath2012non}, distributed source and channel coding \cite{kang2010new}, and the simulation of nonlocal correlations via wirings \cite{beigi2015monotone}. In this paper, we extend the dependence measures with the tensorization property to the matrix case. 

Previously known dependence measures with the tensorization property include the hypercontractivity (HC) ribbon and its generalization, the $\Phi$-ribbon. 
The HC ribbon was first defined by Ahlswede and G\'acs \cite{ahlswede1976spreading}. The following equivalent form in terms of mutual information is proven by Nair \cite{nair2014equivalent}: the HC ribbon of $(X_1,\dots,X_n)$, denoted as $\fR(X_1;\dots;X_n)$, is defined as the set of all $(\lambda_1,\dots,\lambda_n)\in [0,1]^n$, such that for any $p_{U|X^n}$, the following inequality holds: 
\begin{align*}
    \sum_{i=1}^n \lambda_i I(U;X_i) \leq I(U;X^n).
\end{align*} 
Beigi and the second author generalize the HC ribbon to the $\Phi$-ribbon \cite{beigi2018phi} by replacing the mutual-information formulation
with the scalar $\Phi$-entropy 
\[
    H_\Phi(f)
    \defeq \E[\Phi(f)]-\Phi(\E[f]).
\]
For $\Phi(t)=t\log t$, this construction recovers the
hypercontractivity ribbon; for $\Phi(t)=t^2$, it gives the
maximal-correlation ribbon \cite{hirschfeld1935connection,gebelein1941statistische,renyi1959new,renyi1959measures}. 

To generalize the $\Phi$-entropy, Chen and Tropp
\cite{tropp2014subadditivity} introduced the matrix $\Phi$-entropy
\[
    H_\Phi(\mt F)
    \defeq
    \frac{1}{d}\tr\!\left(
        \E[\Phi(\mt F)]-\Phi(\E[\mt F])
    \right),
\]
where $\mt F$ takes values in the $d\times d$ positive semidefinite matrices and its spectra lie in the domain of $\Phi$. They established subadditivity for a class of functions
that includes $t\log t$ and $t^p$, $1<p\leq2$.
Cheng and Hsieh \cite{cheng2016characterizations} subsequently
gave equivalent characterizations of the matrix entropy subadditivity class. 
Matrix $\Phi$-entropy also has an information-theoretic
interpretation. Cheng and Hsieh remark that, when $\Phi(t)=t\log t$ and each $\mt F(x)$ is a
density matrix describing a quantum system $\mt Q$ conditioned on
a classical variable $X$, we have 
\begin{align*}
    H_\Phi(\mt F)&=
    \frac{1}{d} \big(H(\mt Q)-H(\mt Q|X)\big)\\
    &\defeq\frac{1}{d}I(\mt Q;X)
\end{align*}
where by abuse of notation, $H(\cdot)$ denotes the von Neumann entropy, and $I(\mt Q;X)$ is the Holevo quantity of the classical-quantum system \cite{cheng2019matrix}. Thus, entropy inequalities for
matrix-valued functions include inequalities involving quantum 
auxiliary systems. This connection also motivates studying the exponential decay in the quantum Markov semigroup \cite{cheng2017exponential}. 

In this paper, we define the 
\emph{$d$-dimensional matrix $\Phi$-ribbon},  denoted by $\fR_\Phi^{M,d}(X_1;\ldots;X_n)$, as the set of $(\lambda_1,\ldots,\lambda_n)\in[0,1]^n$ such that the inequality 
\[
    \sum_{i=1}^n
    \lambda_i H_\Phi\!\left(\E[\mt F|X_i]\right)
    \leq H_\Phi(\mt F)
\]
holds for every $d$-dimensional positive semidefinite matrix-valued function $\mt F=\mt F(X^n)$. This definition recovers the scalar $\Phi$-ribbon in dimension one. 
The matrix $\Phi$-ribbon is defined as $\fR_\Phi^M(X_1;\ldots;X_n)=\bigcap_{d\geq 1}\fR_\Phi^{M,d}(X_1;\ldots;X_n)$. 
Throughout the paper, the random variables
take values in finite sets. 

Our contributions are as follows.

\begin{enumerate}[\rm (i)]

    \item We prove the tensorization and data processing properties of matrix $\Phi$-ribbons. We show that the scalar and matrix $\Phi$-ribbons characterize the corresponding $\Phi$-strong data processing inequality (SDPI) constants. Previously, the characterization was proved for the scalar case with the additional assumption that $\Phi$ is defined on a compact interval \cite{beigi2018phi}. We remove this domain assumption and extend it to the matrix $\Phi$-entropy case. The proof is established by proving scalar and matrix $\Phi$-entropy inequalities. The information inequalities we introduce in Corrolaries~\ref{corr6} and \ref{cor:delta-DPhi-bd} relate to the copy lemma and might be of independent interest. As an application, we apply the matrix $\Phi$-ribbon to the wirings of the no-signaling boxes and prove its monotonicity. 

    \item 
    We determine a ribbon region that follows solely from
    specified mutually independent subsets of
    $X_1,\ldots,X_n$.
    Representing these subsets by a downward-closed hypergraph,
    we show that the convex hull of their incidence vectors is
    contained in the matrix $\Phi$-ribbon.
    For $\Phi(t)=t\log t$, this guarantee is sharp in every
    matrix dimension: any point outside the convex hull is
    excluded by some finite joint distribution satisfying the
    prescribed independence relations. This generalizes our recent result in the scalar case \cite{wangscalar}. 

\end{enumerate}

\subsection{Organization}

Section~\ref{sec:scalar} reviews scalar $\Phi$-ribbons and $\Phi$-SDPI
constants, establishes the new scalar $\Phi$-entropy inequalities,
and proves the $\Phi$-ribbon characterization of the $\Phi$-SDPI constants. The scalar proofs introduce the techniques used later in the matrix setting; we therefore recommend reading them before the corresponding matrix proofs. 
Section~\ref{sec:matrix-Phi} develops matrix $\Phi$-ribbons and proves their structural properties. 
Section~\ref{sec:SDPI} studies matrix $\Phi$-SDPI constants, proves their characterization by matrix $\Phi$-ribbons, obtains exact evaluations and provides an example on its dimension dependence. 
Section~\ref{sec:wirings} establishes monotonicity under wirings of no-signaling boxes as an application of the matrix $\Phi$-ribbon.  
Section~\ref{sec:characterization-ribbon} derives the inner bound for the matrix $\Phi$-ribbon guarantee from prescribed
independence relations and proves the optimality when $\Phi(t)=t\log t$.

\subsection{Notations}

We write $\R_+\defeq[0,\infty)$ and $\R_{++}\defeq(0,\infty)$.
For a nonnegative integer $n$, let $[n]\defeq\{1,\ldots,n\}$. For a finite set $S$, $|S|$ denotes its cardinality, and $\cP(S)$ (also written $2^S$) denotes its power set. We use $\Conv(S)$ for the convex hull of a subset $S$ of a real vector space and $\F_q$ for the finite field with $q$ elements, where
$q$ is a prime power. The notation $\rank(L)$ denotes the rank of a matrix or linear map $L$. For $t\in\R$, let $t_+\defeq\max\{t,0\}$. 

All random variables are assumed to take values in finite sets. Capital letters denote random variables, lowercase letters their realizations, and calligraphic letters their alphabets. For $S\subseteq[n]$, set
$X_S\defeq(X_i:i\in S)$, and $\cX_S\defeq\prod_{i\in S}\cX_i$; $x_S$ denotes a realization of $X_S$. In particular, $X_{[n]}=X^n=(X_1,\ldots,X_n)$. 
We use the same indexing convention for other families,
such as $A_{[n]}$, $B_{[n]}$, $\Pi_{[n]}$, and $\Omega_{[n]}$. 
Vectors of parameters are written in bold, e.g., $\vec\lambda=(\lambda_1,\ldots,\lambda_n)$. We write $\vec e_i$ 
for the $i$-th standard basis vector and
$\vec e_S\defeq\sum_{i\in S}\vec e_i$. 
The notation $X\sim p_X$ means that $X$ has law $p_X$, and we denote 
\[
\supp(X)=\supp(p_X)\defeq\{x\in\cX:p_X(x)>0\}. 
\]
We write $p\ll q$ if $\supp(p)\subseteq\supp(q)$, and
$p_X\otimes p_Y$ for the product law with mass function $p(x)p(y)$.
The relation $X\independent Y$ denotes independence,
$X\independent Y|Z$ denotes conditional independence given $Z$, 
and $U\to X\to Y$ denotes a Markov chain, equivalently
$U\independent Y|X$. We use $\Bern(\theta)$  for the distribution on $\{0,1\}$ assigning 
probability $\theta$ to $1$, and $\Unif(S)$ for the uniform distribution on a nonempty finite set $S$. 

Let $X$ be a random variable taking values in the finite set $\cX$.
Then a function $f:\cX\to\R$ can be thought of as a random variable. 
The expectation and variance of $f$ are denoted by 
$\E[f]$ and $\Var(f)$, respectively. 
We sometimes denote the expectation by $\E_X[f]$ to emphasize that it is  computed with respect to $X$.

Let $f=f_{XY}=f(X,Y)$ be a function of two random variables
$(X,Y)$ with the joint distribution $p_{XY}$.
Then the conditional expectation of $f$ over $Y$ given a fixed value for $X$, denoted as $\E[f|X]$, is a function of $X$. To be precise, 
\[
\E[f| X](x)
=
\E[f| X=x]
=
\sum_y p(y| x)f(x,y).
\]
For simplicity, we sometimes denote $f_X=\E[f|X]$. 

We use $\log$ for logarithm of base 2 and $\ln$ for natural logarithm. 

Unless a different domain is specified, $\Phi:(0,\infty)\to\R$
is convex and is extended to zero by
\[
\Phi(0)\defeq\lim_{t\downarrow0}\Phi(t)
\in\R\cup\{+\infty\}.
\]
Convexity ensures that this limit exists. This convention gives a
lower semicontinuous convex extension to $[0,\infty)$, which is
continuous at zero when $\Phi(0)<\infty$. We write $\dom(\Phi)$
for the domain of $\Phi$.

\section{Scalar $\Phi$-ribbons, $\Phi$-SDPI constants and a new $\Phi$-entropic inequality}

\label{sec:scalar}

In this section, we first review the definitions of $\Phi$-ribbons and $\Phi$-SDPI constants. Beigi and the second author showed that, when $\Phi$ is defined on a compact interval, the $\Phi$-SDPI constant equals the boundary slope of the $\Phi$-ribbon. We establish $\Phi$-entropy inequalities that allow us to remove this compactness assumption. Extending the argument to matrix $\Phi$-entropy yields a corresponding matrix $\Phi$-entropy inequality, from which we derive the equality between the matrix $\Phi$-SDPI constant and the boundary slope of the matrix $\Phi$-ribbon. These $\Phi$-entropy inequalities also imply mutual information inequalities that, to our knowledge, are new.

Given a convex function $\Phi:\R_+\to\R$, a function $f:\cX\to \R_+$, and a random variable $X$ on $\cX$, the $\Phi$-entropy of $f$ is defined to be (see \cite{Boucheron2013concentration,chafai2004entropies,chafai2006binomial})
\begin{align*}
    \HPhi(f) &\defeq \E[\Phi(f)] - \Phi(\E[f]).
\end{align*}

Given a function $f$ of two random variables $(X,Y)$, we define the \emph{conditional $\Phi$-entropy} as 
\begin{align}
    H_\Phi(f|Y)& \defeq\E[\Phi(f)]-\mathbb{E}_Y[\Phi(\E[f|Y])]\nonumber\\
    &= \sum_{y}p(y)\big(\E[\Phi(f)|Y=y] - \Phi(\E [f|Y=y])\big)\nonumber\\
    &=\HPhi(f)-\HPhi(\E[f|Y]) \label{Phi-ent-CR}.
\end{align}
Note that \eqref{Phi-ent-CR} is the chain rule for the $\Phi$-entropy:
\begin{align*}
    \HPhi(f)
    &=H_\Phi(f|Y)
    +\HPhi(\E[f|Y]).
\end{align*}
This definition of the conditional $\Phi$-entropy has the chain rule. Since $\IPhi$ can be represented in $\Phi$-entropy, we can analogously define the conditional $\Phi$-mutual information in order to preserve the chain rule.

The $\Phi$-divergence (also known as $f$-divergence in the literature) was originally proposed by Ali and Silvey \cite{ali1966general} as well as Csiszár \cite{csiszar2011information,csiszar1967information} and defined as 
\begin{align*}
    \DPhi(q_X\|p_X)
    &\defeq \sum_{x} p(x) \Phi\left(\frac{q(x)}{p(x)}\right)-\Phi(1)
\end{align*}
when $q_X\ll p_X$. 
Define $f(x)=q(x)/p(x)$. Then the $\Phi$-divergence and the $\Phi$-entropy are related as 
\begin{align*}
    \DPhi(q_X\|p_X)=\HPhi(f). 
\end{align*}

In \cite{ziv1973functionals} and \cite[Eq. 3.10.1]{cohen1998comparisons}, the {$\Phi$-mutual information} is defined as  
\begin{align*}
    I_{\Phi}(A;B)\defeq D_{\Phi}(p_{AB}\|p_A\otimes p_B).
\end{align*}
Then we also have 
\begin{align*}
    I_{\Phi}(U;X)\defeq \E_{U} \left[
        D_{\Phi}(p_{X|U}\|p_X)
    \right].
\end{align*}

Given some arbitrary $p_{XY}$, define $f_x(y)\defeq \frac{p(x,y)}{p(x)p(y)}$. Note that $f_x(y)\geq 0$ for any $x,y$ and $\E_Y[f_x]=1$ for any $x$. Then we have 
\begin{align*}
    \IPhi(X;Y)
    &=\sum_{x,y} p(x) p(y) \Phi \left(\frac{p(x,y)}{p(x)p(y)}\right) -\Phi(1) \nonumber\\
    &=\sum_{x} p(x) \HPhi(f_x),
\end{align*}
where $\HPhi(f_x)$ is calculated using $p_Y$ for every $x$. Following \cite{wangscalar}, we define the conditional $\Phi$-mutual information as 
\begin{align*}
    I_\Phi(X;Y| Z)
\defeq
I_\Phi(X;YZ)-I_\Phi(X;Z).
\end{align*}

\subsection{Scalar $\Phi$-ribbons and $\Phi$-SDPI constants}

\begin{definition}[Exercise 14.2 
\cite{Boucheron2013concentration}] \label{def:class-F}
    We define $\mathscr{F}$ to be the class of all smooth convex functions $\Phi:\cI\to \R$ where $\cI\subseteq \R_+$, that are not affine, and satisfy one of the following equivalent conditions: 
\begin{enumerate}[\rm (i)]
    \itshape

    \item
    $(s,t)\mapsto p\Phi(s)+(1-p)\Phi(t)
    -\Phi\bigl(ps+(1-p)t\bigr)$,
    for any $p\in[0,1]$, is jointly convex.

    \item
    $(s,t)\mapsto \Phi(s)-\Phi(t)-\Phi'(t)(s-t)$
    is jointly convex.

    \item
    $(s,t)\mapsto \bigl(\Phi'(s)-\Phi'(t)\bigr)(s-t)$
    is jointly convex.

    \item
    $(s,t)\mapsto \Phi''(s)t^2$ is jointly convex.

    \item
    $1/\Phi''$ is concave.

    \item
    $\Phi''''\Phi''\geq 2(\Phi''')^2$.
\end{enumerate}
\end{definition}

The class $\mathscr{F}$ corresponds to the subadditivity of $\Phi$-entropy. Suppose $(X_1,\dots,X_n)$ are mutually independent. It is shown in \cite{beigi2018phi} that for any $f$, 
\begin{align*}
    \HPhi(f)\leq \sum_{i=1}^n \HPhi \left(f|X_{-i}\right)
\end{align*}
where $X_{-i}=(X_1,\dots,X_{i-1},X_{i+1},\dots,X_n)$. 

Functions in the class $\mathscr{F}$ have found applications in certain concentration of measure inequalities because they are tightly related to subadditivity of $\Phi$-entropy and are of particular importance \cite{chafai2004entropies}. For example, Raginsky uses it to obtain a unified proof of the tensorization of the strong data processing inequality constant \cite{raginsky2016strong}.

\begin{definition}[\cite{beigi2018phi}] \label{def:Phi-ribbon}
    Let $\Phi\in \mathscr{F}$. For finite random variables $(X_1,\dots,X_n)$, the \emph{$\Phi$-ribbon}, denoted as $\fR_{\Phi}(X_1;\dots;X_n)$, is defined as the set of all $(\lambda_1,\dots,\lambda_n) \in [0,1]^n$ such that for any function $f$ of $(X_1,\dots,X_n)$, 
    \begin{align*}  
        \sum_{i=1}^n \lambda_i \HPhi(\E[f|X_i])  
        &\leq \HPhi(f).
    \end{align*}
    The \emph{$\IPhi$-ribbon}, denoted as $\fR_{\IPhi}(X_1;\dots;X_n)$, is defined as the set of all $(\lambda_1,\dots,\lambda_n) \in [0,1]^n$ such that for any distribution $p_{U|X^n}$, 
    \begin{align*}  
        \sum_{i=1}^n \lambda_i \IPhi(U;X_i)
        &\leq \IPhi(U;X^n).
    \end{align*}
\end{definition}

When we impose that $f\geq 0$ and $\E[f]=1$ for the $\Phi$-ribbon, the above two ribbons are the same \cite[Theorem 14]{beigi2018phi}.

Note that we can recover the HC ribbon when $\Phi(t)=t\log t$. The $\Phi$-entropy and the $\Phi$-ribbon have applications in the concentration inequalities, Markov chain mixing times, etc., as shown in \cite{raginsky2016strong}.

\begin{proposition}
\label{prop:phi-ribbon-properties}
    For $\Phi\in \mathscr{F}$, we have 
    \begin{enumerate}[\rm (i)]
        \item (Example 10 \cite{beigi2018phi}) $\{(\lambda_1,\dots,\lambda_n)\in[0,1]^n:\sum_{i=1}^n \lambda_i \leq 1\}$ is always a subset of the $\Phi$-ribbon for any random variables $(X_1,\dots,X_n)$. Moreover, when $X_1=\dots=X_n$ and these random variables are nonconstant, the above set is exactly the $\Phi$-ribbon.
        \item (Proposition 11 \cite{beigi2018phi}) If $X_1,\dots,X_n$ are mutually independent, the $\Phi$-ribbon is exactly $[0,1]^n$. 
    \end{enumerate}
\end{proposition}

\begin{definition}
Suppose $\Phi:\cI\to \R$ is convex where $\cI\subseteq \R_+$. Define
$\eta_{\Phi}(X;Y)$ to be the smallest (the infimum of)
$\lambda \geq 0$ such that, for any function $f=f(X)$
(whose range is in the domain of $\Phi$), we have
\[
\lambda H_{\Phi}(f)
\geq
H_{\Phi}\bigl(\mathbb{E}[f | Y]\bigr).
\]
We call $\eta_{\Phi}(X;Y)$ the \emph{$\Phi$-strong data processing
inequality constant} ($\Phi$-SDPI constant).
\end{definition}
When $X$ is constant or $\Phi$ is affine, we define $\eta_{\Phi}(X;Y)=0$. 

By the chain rule and the nonnegativity of $\Phi$-entropy, 
\begin{align*}
    \HPhi(f)\geq \HPhi \left(\E[f|Y]\right). 
\end{align*}
Thus, $\eta_{\Phi}$ is well defined and $\eta_{\Phi}\in [0,1]$. 

Equivalently, we have 
\begin{align*}
    \eta_{\Phi}(X;Y)=\sup \frac{\HPhi(\E[f|Y])}{\HPhi(f)}
\end{align*}
where the supremum is over all nonnegative $f(X)$ such that $\HPhi(f)>0$. 

For $\Phi\in \mathscr{F}$, the $\Phi$-SDPI constant was first studied by Raginsky \cite{raginsky2016strong} where his version has the constraint $\E[f]=1$. The above version is proposed by Beigi and the second author \cite{beigi2018phi}. 

The following identity was established in
\cite[Theorem 20]{beigi2018phi} under the additional
assumption that the domain of $\Phi$ is compact.
Our proof uses a different technique and removes this assumption.

\begin{theorem}
    \label{thm:equiv-sdpi-ribbon}
    Suppose $\Phi\in \mathscr{F}$ and $X,Y$ are finite
    random variables. Then 
    \begin{align*}
        \eta_{\Phi}(X;Y) &= \inf\frac{1-\lambda_1}{\lambda_2}
    \end{align*}
    where the infimum is over all $(\lambda_1,\lambda_2)\in \fR_{\Phi} (X;Y)$ with $\lambda_2\neq 0$. 
\end{theorem}
The proof is given later in Subsection~\ref{subsec:proof-equiv-scalar}. To understand the proof, we need new $\Phi$-entropic inequalities given in the section below.

\subsection{New $\Phi$-entropic inequalities}

The scalar proofs in this subsection introduce the techniques used later to establish the matrix $\Phi$-entropy inequalities. Reading these proofs first will make the matrix arguments easier to follow.

\begin{theorem}
    \label{thm:scalar-DPhi-ineq}
    Suppose $\Phi\in\mathscr{F}$. 
    Let $(X,Y)$ have distribution $p_{XY}$ on the finite set
    $\cX\times\cY$, and denote $\mu\defeq \min\{p_{XY}(x,y):p_{XY}(x,y)>0\}$. 
    For any function $f:\cX\times\cY\to\dom(\Phi)$, write 
    \begin{align}\label{eqndfy}
        f_X\defeq\E[f|X],
        \qquad
        f_Y\defeq\E[f|Y],
        \qquad
        \widetilde f_Y\defeq\E[f_X|Y].
    \end{align}
    Then
    \begin{align*}
        \HPhi(f_Y)
        &\leq
        \HPhi(\widetilde f_Y)+\HPhi(f|X)
        +\frac{\pi}{\mu}
        \sqrt{\HPhi(f_X)\HPhi(f|X)}.
    \end{align*}
\end{theorem}
Before giving the proof of this theorem, we need some definitions and lemmas.
For $u\in\interior\left(\dom(\Phi)\right)$ and $h\in\R$, write 
\[
    Q_u(h)\defeq\Phi''(u)h^2.
\]
Since $\Phi$ is convex, $\Phi''\geq 0$ and then $h\mapsto Q_u(h)$ is convex. 
The following lemma provides the quadratic estimate used in the proof of Theorem~\ref{thm:scalar-DPhi-ineq}. 

\begin{lemma}
    \label{lem:Phi-bd-quadratic}
    With the same assumptions and notations in Theorem~\ref{thm:scalar-DPhi-ineq}, for every $f:\cX\times \cY\to \interior(\dom(\Phi))$ we have: 
    \begin{align}
        \E_{XY}\left[Q_m (f(X,Y)-m)\right]=\E_{XY} \left[\Phi''(m)(f(X,Y)-m)^2\right]
        \leq
        \frac{2}{\mu}\HPhi(f),\label{uncond}
    \end{align}
    where $m=\E[f]$, and the conditional form:  
    \begin{align*}
        \E\left[Q_{\widetilde f_Y}\left(f_Y-\widetilde f_Y\right)\right]
        &=\E \left[\Phi'' \left(\widetilde f_Y\right) \left(f_Y-\widetilde f_Y\right)^2\right]
        \leq 
        \frac{2}{\mu}\HPhi(f|X)
    \end{align*}
    where $f_Y$ and $\widetilde f_Y$ are defined in \eqref{eqndfy}.
\end{lemma}

\begin{proof}
Recall that 
\[
    Q_u(h)\defeq\Phi''(u)h^2,
    \qquad u>0,\ h\in\R.
\]

We first prove the unconditional inequality. Set $Z=(X,Y)$, and write
$p_Z(z)=p_{XY}(x,y)$ and $f(z)=f(x,y)$ for $z=(x,y)$.
In this part of the proof, all sums are over $\supp(p_Z)$.
If $\mu=1$, then $Z$ is constant and the inequality is immediate.
Hence, assume $0<\mu<1$.

Let $u\in \R$ be any convex combination of $\{f(z):z\in \supp(p_Z)\}$, i.e., 
\[
    u=\sum_z r(z)f(z),
\]
where $r$ is a probability distribution supported on
$\operatorname{supp}(p_Z)$. Since $p_Z(z)\geq\mu$ on this support, we can write
\[
    m=\mu u+(1-\mu)v,
    \qquad
    v\defeq
    \sum_z
    \frac{p_Z(z)-\mu r(z)}{1-\mu}f(z).
\]
The coefficients defining $v$ are nonnegative and sum to one.
Since $\Phi\in\mathscr{F}$, $1/\Phi''$ is concave and positive. Thus, by Jensen's inequality, 
\[
    \frac{1}{\Phi''(m)}
    \geq
    \frac{\mu}{\Phi''(u)}
    +\frac{1-\mu}{\Phi''(v)}
    \geq
    \frac{\mu}{\Phi''(u)}.
\]
Consequently, 
\begin{align} \label{eq:u-m-Phi''}
    \Phi''(u)\geq\mu\Phi''(m)
\end{align}
holds for any $u$ that is a convex combination
of the values of $\{f(z):z\in \supp(p_Z)\}$.
Given $0\leq t\leq 1$, $(1-t)m+tf(z)$ is a convex combination
of the values of $\{f(z):z\in \supp(p_Z)\}$ (as $m$ itself is the expected value of $f$). By \eqref{eq:u-m-Phi''}, we have 
\begin{align} \label{eq:m-f-mu}
    Q_{(1-t)m+tf(z)}(f(z)-m)
    \geq
    \mu Q_m(f(z)-m)
\end{align}
for each $z\in\supp(p_Z)$. Then 
\begin{align*}
    \HPhi(f)
    &=
    \E\left[
        \Phi(f)-\Phi(m)
    \right]\\
    &=
    \E\left[
        \Phi(f)-\Phi(m)-\Phi'(m)(f-m)
    \right]\\
    &=
    \int_0^1(1-t)
    \E\left[
        Q_{(1-t)m+tf}(f-m)
    \right]dt\\
    &\geq \int_0^1(1-t)
    \mu \E \left[Q_m(f-m)\right] dt\\
    &=
    \frac{\mu}{2}\E\left[Q_m(f-m)\right]
\end{align*}
where the second line follows from $\E \left[\Phi'(m)(f-m)\right]=\Phi'(m)\E \left[f-m\right]=0$, the third line follows from the integral Taylor formula, and the fourth line follows from \eqref{eq:m-f-mu}. 
This completes the proof of \eqref{uncond}.

For the conditional form, recall that
$(u,h)\mapsto Q_u(h)$ is jointly convex for
$\Phi\in\mathscr F$. By Jensen's inequality,
\begin{align*}
    \E\left[
        Q_{\widetilde f_Y}(f_Y-\widetilde f_Y)
    \right]
    &=
    \E\left[
        Q_{\E[f_X|Y]}
        \left(\E[f-f_X|Y]\right)
    \right]\\
    &\leq
    \E\left[Q_{f_X}(f-f_X)\right].
\end{align*}

Moreover, since
\[
    p_{Y|X}(y|x)
    =\frac{p_{XY}(x,y)}{p_X(x)}
    \geq\mu
\]
whenever $p_{XY}(x,y)>0$, applying the 
inequality in \eqref{uncond} under each conditional law given $X=x$
and averaging over $X$ gives
\begin{align*}
    \E\left[Q_{f_X}(f-f_X)\right]
    &\leq
    \frac{2}{\mu}
    \E\left[
        \Phi(f)-\Phi(f_X)
    \right]\\
    &=
    \frac{2}{\mu}\HPhi(f|X).
\end{align*}
This proves the conditional form.

For general $f$, consider $f^{(\epsilon)}=(1-\epsilon)f+\epsilon c$ where $c\in \interior(\dom(\Phi))$. 
Then $f^{(\epsilon)}$ takes values in $\interior(\dom(\Phi))$, so the preceding inequalities apply to $f^{(\epsilon)}$. Letting $\epsilon\downarrow0$ and using the continuity of $\Phi$ and the finiteness of the underlying probability space yields the desired inequalities for $f$. This completes the proof. 
\end{proof}

\begin{proof}[Proof of Theorem~\ref{thm:scalar-DPhi-ineq}]
All expectations and conditional expectations below are taken with
respect to $p_{XY}$.
If $\mu=1$, then $(X,Y)$ is deterministic and all the
$\Phi$-entropies in the inequality vanish. Hence, assume $0<\mu<1$.

We first prove the inequalities for $f$ taking values in the interior of $\dom(\Phi)$. 

Write 
$m\defeq\E[f]$. Then $\E[f_X]=\E[f_Y]=\E\left[\widetilde f_Y\right]=m$. 

Since $\Phi\in\mathscr{F}$, the map $(u,h)\mapsto Q_u(h)$ is
jointly convex, as is the map
$(u,v)\mapsto\Phi(u)-\Phi(v)-\Phi'(v)(u-v)$. 
Applying Jensen's inequality given $Y$ to
$(u,v)=(f,f_X)$, and then taking expectation, we obtain
\begin{align}
    \E\left[
        \Phi(f_Y)-\Phi\left(\widetilde f_Y\right)-\Phi'\left(\widetilde f_Y\right)\left(f_Y-\widetilde f_Y\right)
    \right]
    &\leq
    \E\left[
        \Phi(f)-\Phi(f_X)-\Phi'(f_X)(f-f_X)
    \right] \nonumber\\
    &=\HPhi(f|X) \label{eq:scaler-quadratic-Phi}
\end{align}
where the last equality follows from $\E[f-f_X|X]=0$ and the definition of $\Phi$-entropy. 
Since $\E[f_Y]=\E\left[\widetilde f_Y\right]=m$, it follows that 
\begin{align}
    \HPhi(f_Y)-\HPhi\left(\widetilde f_Y\right)
    &= \E \left[\Phi(f_Y)-\Phi \left(\widetilde f_Y\right)\right]\nonumber\\
    &=\E \left[\Phi(f_Y)
    -\Phi \left(\widetilde f_Y\right)
    -\Phi'\left(\widetilde f_Y\right)\left(f_Y-\widetilde f_Y\right)
    \right]
    +\E \left[\Phi'\left(\widetilde f_Y\right)\left(f_Y-\widetilde f_Y\right)\right]\nonumber\\
    &\leq
    \HPhi(f|X)
    +\E \left[\Phi'\left(\widetilde f_Y\right)\left(f_Y-\widetilde f_Y\right)\right]\nonumber\\
    &=
    \HPhi(f|X)
    +\E\left[
        \left(\Phi'\left(\widetilde f_Y\right)-\Phi'(m)\right)\left(f_Y-\widetilde f_Y\right)
    \right]
    \label{eq:scalar-Phi-divergence-linear-term}
\end{align}
where the third line follows from \eqref{eq:scaler-quadratic-Phi} and the last line follows from $\E \left[\Phi'(m)\left(f_Y-\widetilde f_Y\right)\right]=0$. 

It remains to bound the
last term on the right side of \eqref{eq:scalar-Phi-divergence-linear-term}. 

By Jensen's inequality and Lemma~\ref{lem:Phi-bd-quadratic}, we have 
\begin{align}
    \E\left[Q_{m}(\widetilde f_Y-m)\right]
    &\leq 
    \E\left[Q_{m}(f_X-m)\right]
    \nonumber\\
    &\leq
    \frac{2}{\mu}\HPhi(f_X),
    \label{eq:scalar-Phi-divergence-first-quadratic}
\end{align}
and 
\begin{align}
    \E\left[Q_{\widetilde f_Y}\left(f_Y-\widetilde f_Y\right)\right]
    \leq
    \frac{2}{\mu}\HPhi(f|X).
    \label{eq:scalar-Phi-divergence-second-quadratic}
\end{align}

Fix $u,v\in\interior\left(\dom(\Phi)\right)$ and $t\in(0,1)$. We have 
\begin{align*}
    \frac{1}{\Phi''\left((1-t)v+tu\right)}
    &\geq
    \frac{1-t}{\Phi''(v)}
    +\frac{t}{\Phi''(u)}\\
    &\geq
    \frac{2\sqrt{t(1-t)}}
         {\sqrt{\Phi''(v)\Phi''(u)}}.
\end{align*}
where the first line follows from the concavity of $1/\Phi''$ and the second line follows from the AM-GM inequality. Taking the reciprocal, we have 
\begin{align} 
    \label{eq:Phi''-t-v-u}
    \Phi''\left((1-t)v+tu\right)
    \leq
    \frac{\sqrt{\Phi''(v)\Phi''(u)}}
         {2\sqrt{t(1-t)}}.
\end{align}
For every $h\in\R$,
\begin{align}
    \left|
        \left(\Phi'(u)-\Phi'(v)\right)h
    \right|
    &=
    \left|
        \int_0^1
        \Phi''\left((1-t)v+tu\right)(u-v)h\,dt
    \right|\nonumber\\
    &\leq 
    \left|
        \int_0^1
        \frac{\sqrt{\Phi''(v)\Phi''(u)}}
         {2\sqrt{t(1-t)}}(u-v)h\,dt
    \right|\nonumber\\
    &=\sqrt{\Phi''(v)(u-v)^2 \cdot\Phi''(u)h^2}
    \int_0^1\frac{1}{2\sqrt{t(1-t)}}\,dt
    \nonumber\\
    &=
    \frac{\pi}{2}
    \sqrt{Q_v(u-v)Q_u(h)}
    \label{eq:scalar-Phi-divergence-mixed-gradient}
\end{align}
where the first line follows from the fundamental theorem of calculus, the second line follows from \eqref{eq:Phi''-t-v-u} and the last line follows from $\int_0^1\frac{dt}{\sqrt{t(1-t)}}=\pi$. 

Applying \eqref{eq:scalar-Phi-divergence-mixed-gradient}
pointwise with $u=\widetilde f_Y(y)$, $v=m$, and $h=f_Y(y)-\widetilde f_Y(y)$ for each $y$, and then taking expectation, we have 
\begin{align*}
    \left|
        \E\left[
            \left(\Phi'\left(\widetilde f_Y\right)-\Phi'(m)\right)\left(f_Y-\widetilde f_Y\right)
        \right]
    \right|
    &\leq
    \frac{\pi}{2}
    \E\left[
        \sqrt{Q_m\left(\widetilde f_Y-m\right)Q_{\widetilde f_Y}\left(f_Y-\widetilde f_Y\right)}
    \right]\\
    &\leq
    \frac{\pi}{2}
    \sqrt{
        \E\left[Q_m\left(\widetilde f_Y-m\right)\right]
        \E\left[Q_{\widetilde f_Y}\left(f_Y-\widetilde f_Y\right)\right]
    }\\
    &\leq
    \frac{\pi}{\mu}
    \sqrt{\HPhi(f_X)\HPhi(f|X)}
\end{align*}
where the second line follows from the Cauchy--Schwarz inequality, and the last line follows from 
\eqref{eq:scalar-Phi-divergence-first-quadratic}
and
\eqref{eq:scalar-Phi-divergence-second-quadratic}. 
Substituting this into
\eqref{eq:scalar-Phi-divergence-linear-term}
finishes the proof.

For general $f$, consider $f^{(\epsilon)}=(1-\epsilon)f+\epsilon c$ where $c\in \interior(\dom(\Phi))$. 
Then $f^{(\epsilon)}$ takes values in $\interior(\dom(\Phi))$, so the preceding inequalities apply to $f^{(\epsilon)}$. Letting $\epsilon\downarrow0$ and using the continuity of $\Phi$ and the finiteness of the underlying probability space yields the desired inequalities for $f$. This completes the proof. 

\end{proof}

\begin{theorem}
    \label{thm:scalar-DPhi-ineq-delta}
    Suppose $\Phi\in\mathscr{F}$. 
    Let $(X,Y)$ have distribution $p_{XY}$ on the finite set
    $\cX\times\cY$, and denote $\mu\defeq \min\{p_{XY}(x,y):p_{XY}(x,y)>0\}$. 
    For any function $f:\cX\times\cY\to\dom(\Phi)$, write 
    \[
        f_X\defeq\E[f|X],
        \qquad
        f_Y\defeq\E[f|Y],
        \qquad
        \widetilde f_Y\defeq\E[f_X|Y].
    \]
    Then for each $\delta>0$, we have 
    \begin{align*}
        \HPhi(f_Y)
        &\leq
        \HPhi(\widetilde f_Y)
        +\delta \HPhi(f_X)
        +K_{\delta} 
        \HPhi(f|X)
    \end{align*}
    where $K_\delta\defeq 1+\frac{\pi^2}{4\delta\mu^2}$. 
\end{theorem}
\begin{proof}
    By Theorem~\ref{thm:scalar-DPhi-ineq}, we have 
    \begin{align*}
        \HPhi(f_Y)
        &\leq
        \HPhi(\widetilde f_Y)+\HPhi(f|X)
        +\frac{\pi}{\mu}
        \sqrt{\HPhi(f_X)\HPhi(f|X)}.
    \end{align*}
    By Young's inequality, we have 
    \begin{align*}
        \frac{\pi}{\mu}
        \sqrt{\HPhi(f_X)\HPhi(f|X)}
        \leq \delta \HPhi(f_X) +\frac{\pi^2}{4\delta\mu^2}\HPhi(f|X).
    \end{align*}
    Substituting this bound into the preceding inequality gives the claimed inequality with \(K_\delta=1+\pi^2/(4\delta\mu^2)\). 
\end{proof}

\begin{corollary}
    \label{cor:sqrt-DPhi-bd}
    Suppose $\Phi\in\mathscr{F}$ is defined on $\R_+$. 
    Let $p_{XY}$ and $q_{XY}$ be distributions on $\cX\times \cY$ such that $q_{XY}\ll p_{XY}$. Denote $\mu\defeq \min\{p_{XY}(x,y):p_{XY}(x,y)>0\}$. Then 
    \begin{align*}
        \DPhi(q_Y\|p_Y)
        &\leq
        \DPhi(\widetilde q_Y\|p_Y)
        +\DPhi(q_{XY}\|p_{XY})
        -\DPhi(q_X\|p_X)\\
        &\quad 
        +\frac{\pi}{\mu}
        \sqrt{\DPhi(q_X\|p_X)\big(\DPhi(q_{XY}\|p_{XY})
        -\DPhi(q_X\|p_X)\big)}
    \end{align*}
    where $\widetilde q_Y(y)=\sum_{x}q(x)p(y|x)$ for each $y$. 
\end{corollary}
\begin{proof}
    Consider $f_{XY}(x,y)=\frac{q(x,y)}{p(x,y)}$. Substituting it into Theorem~\ref{thm:scalar-DPhi-ineq} gives the desired inequality. 
\end{proof}
Here $\widetilde q_Y$ can be interpreted as the output distribution obtained by passing \(q_X\) through the channel $p_{Y|X}$.

\begin{corollary}\label{corr6}
    Suppose $\Phi\in\mathscr{F}$ is defined on $\R_+$. 
    For random variables $UXY$, define $\wY$ on the same alphabet as $Y$ by the joint distribution $p_{UX\wY}=p_{UX}p_{Y|X}$. We have 
    \begin{align*}
        \IPhi(U;Y)
        \leq 
        \IPhi(U;\wY)+\IPhi(U;Y|X)+\frac{\pi}{\mu} \sqrt{\IPhi(U;X) \IPhi(U;Y|X)}
    \end{align*}
    where $\mu\defeq \min\{p_{XY}(x,y):p_{XY}(x,y)>0\}$.
\end{corollary}

\begin{proof}
We use Corollary \ref{cor:sqrt-DPhi-bd}. Consider the joint distribution $p_{UXY}$. 
Fixing $u\in \cU$, set $q_Y(y)=p_{Y|U}(y|u)$. Then $p_{\wY|U=u}=\widetilde q_Y$. Then the above inequality becomes 
\begin{align*}
    \DPhi(p_{Y|U=u}\|p_Y)
    &\leq
    \DPhi(p_{\wY|U=u}\|p_Y)
    +\DPhi(p_{XY|U=u}\|p_{XY})
    -\DPhi(p_{X|U=u}\|p_X)\\
    &\quad 
    +\frac{\pi}{\mu}
    \sqrt{\DPhi(p_{X|U=u}\|p_X)\big(\DPhi(p_{XY|U=u}\|p_{XY})
    -\DPhi(p_{X|U=u}\|p_X)\big)}.
\end{align*}
Taking the expectation over $u$, we obtain 
\begin{align*}
    \IPhi(U;Y)
    &\leq
    \IPhi(U;\wY)+\IPhi(U;Y|X) 
    +\sum_u p(u)\frac{\pi}{\mu}
    \sqrt{\DPhi(p_{X|U=u}\|p_X)\big(\DPhi(p_{XY|U=u}\|p_{XY})
    -\DPhi(p_{X|U=u}\|p_X)\big)}\\
    &\leq \IPhi(U;\wY)+\IPhi(U;Y|X) \\
    &\quad
    +\frac{\pi}{\mu}
    \sqrt{ \left(\sum_u p(u)\DPhi(p_{X|U=u}\|p_X)\right) \left(\sum_u p(u)\big(\DPhi(p_{XY|U=u}\|p_{XY})
    -\DPhi(p_{X|U=u}\|p_X)\big)\right)}\\
    &=\IPhi(U;\wY)+\IPhi(U;Y|X)+\frac{\pi}{\mu} \sqrt{\IPhi(U;X) \IPhi(U;Y|X)}
\end{align*}
where the second inequality follows from the joint concavity of $(a,b)\mapsto \sqrt{ab}$. It finishes the proof. 
\end{proof}

\begin{corollary}
    \label{cor:delta-DPhi-bd}
    Suppose $\Phi\in \mathscr{F}$ is defined on $\R_+$. 
    Let $p_{XY}$ and $q_{XY}$ be distributions on $\cX\times \cY$ such that $q_{XY}\ll p_{XY}$. Denote $\mu\defeq \min\{p_{XY}(x,y):p_{XY}(x,y)>0\}$. Then for any $\delta>0$, 
    \begin{align*}
        \DPhi(q_Y\|p_Y)
        &\leq
        \DPhi(\widetilde q_Y\|p_Y)
        +\delta \DPhi(q_X\|p_X)
        +\left(1+\frac{\pi^2}{4\delta \mu^2}\right) \big(\DPhi(q_{XY}\|p_{XY})
        -\DPhi(q_X\|p_X)\big)
    \end{align*}
    where $\widetilde q_Y(y)=\sum_{x}q(x)p(y|x)$ for each $y$. 

    Consequently, for random variables $UXY$, define $\wY$ by the joint distribution $p_{UX\wY}=p_{UX}p_{Y|X}$. For any $\delta>0$, we have 
    \begin{align}
        \IPhi(U;Y)
        \leq 
        \IPhi(U;\wY)+\delta\IPhi(U;X)+\left(1+\frac{\pi^2}{4\delta \mu^2}\right) \IPhi(U;Y|X).\label{eqnIP1}
    \end{align}
\end{corollary}
Corollary~\ref{cor:delta-DPhi-bd} is a direct consequence of Theorem~\ref{thm:scalar-DPhi-ineq-delta}, 
following the same idea in Corollary~\ref{cor:sqrt-DPhi-bd}. Therefore, we do not repeat the proof here. 

\begin{remark}Consider \eqref{eqnIP1} for Shannon's mutual information. We can think of $\wY$ as a copy of $Y$. What separates \eqref{eqnIP1} from ordinary Shannon-type inequalities is the use of $\mu$ to write an information inequality. 
\end{remark}

\subsection{Proof of Theorem~\ref{thm:equiv-sdpi-ribbon}}
\label{subsec:proof-equiv-scalar}

If $X$ is constant, then $\eta_\Phi(X;Y)=0$ by convention. In this case, $\fR_\Phi(X;Y)=[0,1]^2$. Taking $(\lambda_1,\lambda_2)=(1,1)$, we have $(1-\lambda_1)/\lambda_2=0$. So the identity is immediate. 
Assume henceforth that $X$ is nonconstant. Write
\[
\eta:=\eta_\Phi(X;Y),
\qquad
\alpha:=
\inf_{\substack{
(\lambda_1,\lambda_2)\in\fR_\Phi(X;Y)\\
\lambda_2>0
}}
\frac{1-\lambda_1}{\lambda_2}.
\]

First, let $f=f(X)$ with $\HPhi(f)>0$. For $(\lambda_1,\lambda_2)\in \fR_{\Phi}(X;Y)$ such that $\lambda_2>0$, we have 
\[
\lambda_1\HPhi(f)
+\lambda_2\HPhi\bigl(\E[f| Y]\bigr)
\leq\HPhi(f).
\]
Then 
\begin{align*}
    \frac{\HPhi\bigl(\E[f| Y]\bigr)}{\HPhi(f)}
    \leq \frac{1-\lambda_1}{\lambda_2}. 
\end{align*}
Taking the supremum over $f(X)$ and then the infimum over
ribbon points, we have $\eta\leq\alpha$.

For the reverse direction, for every admissible $f=f(X,Y)$, set
$f_X:=\E[f| X]$, $f_Y=\E[f|Y]$ and $\widetilde f_Y=\E[f_X|Y]$. Fix $\delta>0$ and take $K_\delta=1+\pi^2/(4\delta \mu^2)$. By Theorem~\ref{thm:scalar-DPhi-ineq-delta}, we have 
\begin{align}
    \HPhi\bigl(f_Y\bigr)
    &\leq\HPhi\left(\widetilde f_Y\right)
    +\delta\HPhi(f_X)+K_\delta\HPhi(f| X)\nonumber\\
    &\leq(\eta+\delta)\HPhi(f_X)
    +K_\delta\HPhi(f| X), \label{eq:apply-lemma-f-X-Y}
\end{align}
where the second inequality follows from the SDPI inequality $\HPhi\left(\widetilde f_Y\right)=\HPhi \left(\E[f_X|Y]\right)\leq \eta \HPhi(f_X)$.

Choose an integer $n>\max\{1,\eta+\delta,K_\delta\}$. Dividing both sides of \eqref{eq:apply-lemma-f-X-Y} by $n$, and adding $(1-(\eta+\delta)/n)\HPhi(f_X)$ to both sides, we obtain 
\begin{align*}
    \left(1-\frac{\eta+\delta}{n}\right)
    \HPhi\bigl(f_X\bigr)
    +\frac1n\HPhi\bigl(f_Y\bigr)
    &\leq\HPhi(f_X)
    +\frac{K_\delta}{n}\HPhi(f| X)\\
    &\leq\HPhi(f_X)+\HPhi(f| X)\\
    &=\HPhi(f),
\end{align*}
where the second inequality follows from $n>K_{\delta}$ and the last equality follows from the chain rule. 
Hence, 
\[
\left(1-\frac{\eta+\delta}{n},\frac1n\right)
\in\fR_\Phi(X;Y),
\]
and then by the definition of $\alpha$, we obtain $\alpha\leq\eta+\delta$. Since $\delta>0$ is arbitrary, $\alpha\leq\eta$. Together with the previously established inequality $\eta\leq \alpha$, it proves the identity.

\section{Matrix \texorpdfstring{$\Phi$}{Phi}-ribbon} \label{sec:matrix-Phi}

\subsection{Preliminaries on matrix functions}

We briefly recall the matrix-function notation used throughout the
paper. For further background, we refer the readers to
\cite{bhatia2013matrix,hiai2014introduction}.

Let $\Phi:\cD\to\R$, where $\cD\subseteq\R$ is an interval and let $d\geq 1$ be an integer. Every $\mt A\in \sa^d$ admits a spectral decomposition $\mt A=\mt U\mt\Lambda\mt U^*$ where $\mt\Lambda=\diag(\lambda_1,\ldots,\lambda_d)$ with $\lambda_i\in \R$ for each $i\in [d]$ 
and $\mt U$ is a unitary matrix. Assume $\lambda_i\in\cD$ for each $i$. Define
\[
\Phi(\mt A)
\defeq
\mt U\Phi(\mt\Lambda)\mt U^*,
\]
where
\[
\Phi(\mt\Lambda)
\defeq
\diag\bigl(\Phi(\lambda_1),\ldots,\Phi(\lambda_d)\bigr).
\]

The map $\Phi$ is called Fr\'echet differentiable at 
$\mt A$ if there exists a linear transformation $\D \Phi(\mt A)$ on $\sa^d$ such that for all $\mt H\in \sa^d$, 
\begin{align*}
    \|\Phi(\mt A+\mt H)-\Phi(\mt A)-\D \Phi(\mt A)[\mt H]\|
    =o(\|\mt H\|),
\end{align*}
where $\|\cdot\|$ is a matrix norm. If it exists, the mapping $\D \Phi(\mt A)[\cdot]$ is called the \emph{Fr\'echet derivative} of $\Phi$ at $\mt A$. Similarly, we write
$\D^2\Phi(\mt A)[\cdot,\cdot]$ for the second-order Fr\'echet
derivative.

For a differentiable function $f:\R\to \R$, define its \emph{first divided
difference} by 
\[
f^{[1]}(\lambda,\mu)
\defeq
\begin{cases}
\dfrac{f(\lambda)-f(\mu)}{\lambda-\mu},
& \lambda\neq\mu,\\[1.2ex]
f'(\lambda),
& \lambda=\mu.
\end{cases}
\]
For $\mt\Lambda=\diag(\lambda_1,\ldots,\lambda_d)$, let
$f^{[1]}(\mt\Lambda)$ denote the matrix whose $(i,j)$-th entry is
$f^{[1]}(\lambda_i,\lambda_j)$.

\begin{fact}[\cite{bhatia2013matrix}]
    \label{fact:Daleckii--Krein formula}
    Suppose $\Phi$ is continuously differentiable. 
Given the spectral decomposition $\mt A=\mt U\mt\Lambda\mt U^*$, 
we have 
\[
\D\Phi(\mt A)[\mt H]
=
\mt U
\left(
\Phi^{[1]}(\mt\Lambda)
\circ
(\mt U^*\mt H\mt U)
\right)
\mt U^*,
\]
where $\circ$ denotes the Schur product (elementwise product).
\end{fact}

\begin{lemma} \label{lem:sa-psd-DPhi'}
    Let $d\geq 1$, and $\Phi:\mathcal I\to\mathbb R$ be convex and twice continuously differentiable on an interval $\mathcal I$.
    If $\mt A\in\PSD^d$ has spectrum contained in $\mathcal I$, then 
    \begin{align} \label{DPhi'-sa}
        \nortr \left(\mt B \D\Phi'(\mt A)[\mt C]\right)=\nortr \left(\D\Phi'(\mt A)[\mt B] \mt C\right),\qquad
        \forall \mt B,\mt C\in \sa^d,
    \end{align}
    and 
    \begin{align}\label{DPhi'-pos}
        \nortr\!\left(
        \mt H\,\D\Phi'(\mt A)[\mt H]
        \right)\geq0,
        \qquad
        \forall\,\mt H\in\sa^d.
    \end{align}
    Moreover, $\mt H\mapsto \nortr\!\left(
        \mt H\,\D\Phi'(\mt A)[\mt H]
    \right)$ is convex. 
\end{lemma}

\begin{proof}
Property~\eqref{DPhi'-sa} follows from the symmetry
\[
\Phi'^{[1]}(s,t)=\Phi'^{[1]}(t,s).
\]
Indeed, let
\[
\mt A=\mt U\mt\Lambda\mt U^*,
\qquad
\mt\Lambda=\operatorname{diag}(\lambda_1,\ldots,\lambda_d),
\]
and denote
\[
\widetilde{\mt B}:=\mt U^*\mt B\mt U,
\qquad
\widetilde{\mt C}:=\mt U^*\mt C\mt U.
\]
By Fact~\ref{fact:Daleckii--Krein formula},
\[
\D\Phi'(\mt A)[\mt C]
=
\mt U
\left(
\Phi'^{[1]}(\mt\Lambda)
\circ
\widetilde{\mt C}
\right)
\mt U^*.
\]
Hence
\begin{align*}
\nortr\left(
\mt B\D\Phi'(\mt A)[\mt C]
\right)
&=
\frac{1}{d}\sum_{i,j}
\Phi'^{[1]}(\lambda_i,\lambda_j)
\widetilde B_{ji}\widetilde C_{ij}\\
&=
\frac{1}{d}\sum_{i,j}
\Phi'^{[1]}(\lambda_i,\lambda_j)
\widetilde B_{ij}\widetilde C_{ji}\\
&=
\nortr\left(
\D\Phi'(\mt A)[\mt B]\mt C
\right),
\end{align*}
where in the second equality we interchange \(i\) and \(j\) and use
\(\Phi'^{[1]}(\lambda_i,\lambda_j)
=
\Phi'^{[1]}(\lambda_j,\lambda_i)\).

For Property~\eqref{DPhi'-pos}, let \(\mt H\in\sa^d\) and denote 
\[
\widetilde{\mt H}:=\mt U^*\mt H\mt U.
\]
Again by Fact~\ref{fact:Daleckii--Krein formula},
\[
\D\Phi'(\mt A)[\mt H]
=
\mt U
\left(
\Phi'^{[1]}(\mt\Lambda)
\circ
\widetilde{\mt H}
\right)
\mt U^*.
\]
Therefore,
\begin{align*}
\nortr\left(
\mt H\D\Phi'(\mt A)[\mt H]
\right)
&=
\nortr\left(
\widetilde{\mt H}
\left(
\Phi'^{[1]}(\mt\Lambda)
\circ
\widetilde{\mt H}
\right)
\right)\\
&=
\frac{1}{d}\sum_{i,j}
\Phi'^{[1]}(\lambda_i,\lambda_j)
\left|\widetilde H_{ij}\right|^2.
\end{align*}
Since \(\Phi\) is convex, \(\Phi'\) is nondecreasing, and hence
\[
\Phi'^{[1]}(s,t)\geq 0,
\qquad
s,t\in\mathcal I.
\]
Thus,
\[
\nortr\left(
\mt H\D\Phi'(\mt A)[\mt H]
\right)\geq 0,
\qquad
\forall \mt H\in\sa^d.
\]

Finally, define
\[
Q(\mt H)
:=
\nortr\left(
\mt H\D\Phi'(\mt A)[\mt H]
\right).
\]
For any \(\mt H_0,\mt H_1\in\sa^d\) and \(\theta\in[0,1]\), by the
linearity of \(\D\Phi'(\mt A)\) and Property~\eqref{DPhi'-sa}, we have 
\begin{align*}
&(1-\theta)Q(\mt H_0)
+\theta Q(\mt H_1)
-Q\left((1-\theta)\mt H_0+\theta\mt H_1\right)\\
=&
\theta(1-\theta)
\nortr\left(
(\mt H_1-\mt H_0)
\D\Phi'(\mt A)[\mt H_1-\mt H_0]
\right)\\
\geq& 0,
\end{align*}
where the last inequality follows from Property~\eqref{DPhi'-pos}. Therefore, $\mt H
\mapsto
\nortr\left(
\mt H\D\Phi'(\mt A)[\mt H]
\right)$
is convex.
\end{proof}

\begin{lemma}[Theorem 3.23 \cite{hiai2014introduction}]
    
\label{lem:trace-spectral-derivatives}
Let $\Phi$ be twice continuously differentiable on an open interval. For $g(t)\defeq\tr\Phi(\mt A+t\mt B)$ where $\mt A,\mt B\in\sa^d$, 
we have
\[
g'(t)
=
\tr\left(
\Phi'(\mt A+t\mt B)\mt B
\right).
\]
If $\mt A+t\mt B
=
\mt U\diag(\lambda_1,\ldots,\lambda_d)\mt U^*$, 
then
\[
g''(t)
=
\tr\left(
\D\Phi'(\mt A+t\mt B)[\mt B]\mt B
\right)
=
\sum_{i,j=1}^d
\Phi'^{[1]}(\lambda_i,\lambda_j)
\left|
(\mt U^*\mt B\mt U)_{ij}
\right|^2.
\]
\end{lemma}

\begin{corollary} \label{cor:integral-HPhi}
    Suppose $\Phi:\cI\to \R$ is twice continuously differentiable for $\cI\subseteq \R_+$ and $\Phi'$ is Fr\'echet differentiable. Let $\mt F:\cX\to \PSD^d$ where $\spec(\mt F(x))\subseteq \cI$ for all $x$ and denote $\mt C=\E[\mt F]$. Then 
    \begin{align*}
        \HPhi(\mt F)
        &=\int_0^1(1-t)\E\left[\nortr\left(
        \D\Phi'\bigl(\mt C+t(\mt F-\mt C)\bigr)
        [\mt F-\mt C]\cdot(\mt F-\mt C)
        \right)\right]dt
    \end{align*}
\end{corollary}
\begin{proof}
For each $x$, define  
\[
g_x(t)\defeq
\nortr\Phi\bigl(\mt C+t(\mt F(x)-\mt C)\bigr).
\]
By Lemma~\ref{lem:trace-spectral-derivatives},
\[
g_x''(t)
=
\nortr\left(
\D\Phi'\bigl(\mt C+t(\mt F(x)-\mt C)\bigr)
[\mt F(x)-\mt C](\mt F(x)-\mt C)
\right).
\]
By Taylor's theorem \cite[Theorem 12.48]{hunter2014introduction}, we have 
\[
g_x(1)-g_x(0)-g_x'(0)
=
\int_0^1(1-t)g_x''(t)\,dt.
\]
Taking expectations and using $\E[\mt F-\mt C]=\mt0$, we obtain the equality.
\end{proof}

\subsection{Matrix \texorpdfstring{\(\Phi\)}{Phi}-entropy and matrix \texorpdfstring{\(\Phi\)}{Phi}-ribbon}

Tropp and Chen define the matrix $\Phi$-entropy as follows. 
\begin{definition}[\cite{tropp2014subadditivity}]
    Suppose $\Phi:\R_+\to \R$ is convex. The \emph{matrix $\Phi$-entropy} of the $d$-dimensional PSD matrix-valued function $\mt F: \cX\to \PSD^d$ is defined as 
    \begin{align*}
        \HPhi(\mt{F})\defeq \nortr(\E[\Phi(\mt{F})]-\Phi(\E[\mt{F}])),
    \end{align*}
    where $\nortr\defeq \tr/d$ is the normalized trace for $d\times d$ matrices and $\spec(\mt F(x))\subseteq \dom(\Phi)$ for all $x\in \cX$. 

    The \emph{conditional matrix $\Phi$-entropy} of the $d$-dimensional PSD matrix-valued function $\mt F: \cX\times \cY\to \PSD^d$ is defined as 
    \begin{align*}
        \HPhi(\mt F|X)\defeq \sum_{x} p(x)\HPhi(\mt F|X=x)
    \end{align*}
    where $\spec(\mt F(x,y))\subseteq \dom(\Phi)$ for all $x\in \cX,y\in \cY$. 
\end{definition}

The matrix $\Phi$-entropy has applications in matrix concentration inequalities (see \cite{tropp2014subadditivity,tropp2015introduction,paulin2016efron}). 

Throughout, matrix-valued functions have spectra contained in $\dom(\Phi)$. 

Since $\Phi$ is convex, the matrix $\Phi$-entropy is nonnegative. Analogous to the scalar case, we have the chain rule for the $\Phi$-entropy, i.e.,  
\begin{align} \label{matrix-Phi:chain-rule}
    \HPhi(\mt F)=\HPhi(\E[\mt F|X])+\HPhi(\mt F|X).
\end{align}
Using induction, the above equation implies that for any function $\mt F(X^n)$, we have 
\begin{align}
    \HPhi(\E[\mt F|\Xn]) &= \sum_{i=1}^n \HPhi \left(\E[\mt F|X_{[i]} ]| X_{[i-1]}\right) \label{eqn:matrix-CRn}.
\end{align}
Moreover, if $\mt F$ is a function of $(X,Y,Z)$, then
\begin{align} \label{eqn:matrix-CR2}
    \HPhi(\mt F|X) &= \HPhi(\mt F|XY) + \HPhi(\E[\mt F|XY]|X).
\end{align}
By the nonnegativity of the matrix $\Phi$-entropy, 
\begin{align}\label{matrix-Phi:cond-reduce}
    \HPhi(\mt F|X)\geq \HPhi(\mt F|XY).
\end{align}
Equation \eqref{matrix-Phi:cond-reduce} means that conditioning reduces the matrix $\Phi$-entropy, analogous to Shannon entropy. By the chain rule, \eqref{matrix-Phi:cond-reduce} is equivalent to 
\begin{align}
    \HPhi(\E[\mt F|XY])\geq \HPhi(\E[\mt F|X]).
\end{align}

\begin{lemma}
\label{lem:Taylor-Phi-Ent}
Let $\Phi:\cD\to\R$ be twice Fr\'echet differentiable on an open
interval $\cD\subseteq\R_+$. Let 
$\mt C\in\sa^d$ be a deterministic matrix satisfying $\operatorname{spec}(\mt C)\subseteq\cD$, 
and let $\mt F=\mt F(X)\in\sa^d$ satisfy $\E[\mt F]=\mt 0$. 
Assume that, for all sufficiently small $|\epsilon|$, $\spec(\mt C+\epsilon\mt F(x))
\subseteq\cD$ for all $x\in\cX$. 
Define $\mt G\defeq\mt C+\epsilon\mt F$. 
Then
\begin{align}
\HPhi(\mt G)
&=
\frac{\epsilon^2}{2}
\E\left[
\nortr\left(
\D^2\Phi(\mt C)[\mt F,\mt F]
\right)
\right]
+
o(\epsilon^2).
\label{eq:matrix-Phi-Taylor}
\end{align}
More explicitly, given the spectral decomposition $\mt C=\mt U\diag(\lambda_1,\ldots,\lambda_d)\mt U^\ast$, 
we have 
\begin{align}
\HPhi(\mt G)
&=
\frac{\epsilon^2}{2d}
\E\left[
\sum_{i,j=1}^d
\Phi'^{[1]}(\lambda_i,\lambda_j)
\left|
(\mt U^\ast\mt F\mt U)_{ij}
\right|^2
\right]
+
o(\epsilon^2).
\label{eq:matrix-Phi-Taylor-explicit}
\end{align}
\end{lemma}

\begin{proof}
By the definition of the Fr\'echet derivative, we have the following second-order Fr\'echet expansion 
\begin{align*}
\Phi(\mt C+\epsilon\mt F)
&=
\Phi(\mt C)
+
\epsilon\D\Phi(\mt C)[\mt F]
+
\frac{\epsilon^2}{2}
\D^2\Phi(\mt C)[\mt F,\mt F]
+
o(\epsilon^2).
\end{align*}
Moreover, $\E[\mt G]=\mt C$.
And by linearity,
\[
\E\left[\D\Phi(\mt C)[\mt F]\right]
=
\D\Phi(\mt C)[\E[\mt F]]
=
\mt 0.
\]
Taking expectation and normalized trace proves
\eqref{eq:matrix-Phi-Taylor}.

Finally, by Lemma~\ref{lem:trace-spectral-derivatives}, 
\[
\nortr\left(
\D^2\Phi(\mt C)[\mt F,\mt F]
\right)
=
\frac1d
\sum_{i,j=1}^d
\Phi'^{[1]}(\lambda_i,\lambda_j)
\left|
(\mt U^\ast\mt F\mt U)_{ij}
\right|^2.
\]
Substituting this identity into \eqref{eq:matrix-Phi-Taylor} proves
\eqref{eq:matrix-Phi-Taylor-explicit}. 

\end{proof}

\begin{definition} \label{def:class-FM}
    For $d\geq 1$, define $\mathscr{F}_M^{d}$ to be the set of $\Phi:\cI\rightarrow \R$ where $\cI\subseteq \R_+$ such that $\Phi$ is continuous and convex, is twice continuously differentiable in the interior of $\cI$, is not affine, and satisfies one of the following equivalent conditions (see \cite[Theorem 3.3]{cheng2016characterizations}): 
    \setlist{label=\rm (\roman*)}
    \begin{enumerate}
        \item Define $\Psi(t)=\Phi'(t)$. The Fr\'echet derivative $\D\Psi$ is invertible and the mapping $\mt A\mapsto (\D\Psi(\mt A))^{-1}$ is concave. 
        \item $(\mt A,\mt B)\mapsto \tr(\Phi(\mt A+\mt B)-\Phi(\mt A)-\D\Phi(\mt A)[\mt B])$ is jointly convex. 
        \item Define $\Psi(t)=\Phi'(t)$. The mapping $(\mt A,\mt B)\mapsto \tr(\D^2\Phi(\mt A)[\mt B,\mt B])=\tr(\mt B \D \Psi(\mt A)[\mt B])$ is jointly convex. 
    \end{enumerate}

    Then define $\mathscr{F}_M\defeq \bigcap_{d\geq 1}\mathscr{F}_M^d$. 
\end{definition}

Since it is required that $\D \Psi$ is invertible, for any $\Phi\in \mathscr{F}_M^d$, $\Phi''(t)>0$ for all $t\in \interior(\cI)$.

As remarked by Cheng and Hsieh \cite{cheng2016characterizations}, the above equivalent conditions can be seen as the non-commutative version of $\mathscr{F}$ and correspond to the subadditivity of the matrix $\Phi$-entropy \cite{tropp2014subadditivity}. That is, for any $\Phi\in \mathscr F_M$ and mutually independent $X_1,\dots,X_n$, we have 
\begin{align*}
    \HPhi(\mt F) \leq \sum_{i=1}^n \HPhi \left(\mt F|X_{\widehat{i}}\right),
\end{align*}
where $X_{\widehat{i}}=(X_1,\dots,X_{i-1},X_{i+1},\dots,X_n)$. 

By definition, note that $\mathscr{F}=\mathscr{F}_M^1$. 

\begin{example}
    The functions $t\log t$ and $t^p/p$ for $p\in (1,2]$ belong to $\mathscr{F}_M$ \cite{tropp2014subadditivity}. When affine functions are allowed in $\mathscr F_M^d$, 
    Hansen and Zhang prove that $\mathscr F_M^d$ is a convex cone for every $d$ and conjecture that if $\Phi$ is defined on $[0,\infty)$, the aforementioned two functions, their shifts and their nonnegative combinations exhaust $\mathscr F_M$ based on numerical calculations \cite{hansen2015characterisation}. 
\end{example}
\begin{remark}\label{remark-matrix-mi-form}
    Assume $\Phi(t)=\psi(t)\defeq t\log t$. Consider a mapping $\mt F: \cX\to \PSD^d$ such that $\tr(\mt F(x))=1$ for all $x\in\cX$. Then, one can view $\mt F(x)$ as the state of a $d$-dimensional quantum system $\mt U$, conditioned on $X=x$. Cheng and Hsieh find that in this setting, $(X,\mt U)$ is a classical-quantum ensemble \cite{cheng2019matrix}. Then, $\E[\mt{F}]$ is the average state of $\mt U$ and the matrix $\Phi$-entropy is 
    \begin{align*}
        H_{\psi}(\mt{F}) &=\nortr(\E[\mt{F}\log\mt{F}]-\E[\mt{F}]\log(\E[\mt{F}]))\\
        &=\frac{1}{d}(-H(\mt U|X)+H(\mt U)),
    \end{align*}
    where, by abuse of notation, $H(\cdot)$ denotes the von Neumann entropy and the definition of quantum mutual information is standard. For the classical-quantum ensemble $(X,\mt U)$, define 
    \begin{align*}
        I(\mt U;X) 
        \defeq H(\mt U)-H(\mt U|X).
    \end{align*}
    Then 
    \begin{align*}
        H_{\psi}(\mt{F})=\frac{1}{d}I(\mt U;X).
    \end{align*}
\end{remark}
Inspired by the above, we define  the matrix $\Phi$-mutual information as 
\begin{align*}
    \IPhi(\mt F;X)
    \defeq \HPhi(\E[\mt F|X])
\end{align*}
where $\mt F$ is a PSD matrix-valued function of $(X,Y)$. By the chain rule, we have  
\begin{align*}
    \IPhi(\mt F;X)
    =\HPhi(\mt F)-\HPhi(\mt F|X). 
\end{align*}

We define the matrix $\Phi$-ribbon for the first time as follows: 
\begin{definition} \label{def:matrix-Phi-ribbon}
    For $\Phi\in \mathscr{F}_M^d$, the \emph{$d$-dimensional matrix $\Phi$-ribbon} for the random variables $(X_1,\dots,X_n)$ taking values in $\prod_{i=1}^n \cX_i$, denoted as $\fR_{\Phi}^{M,d}(X_1;\dots;X_n)$, is the set of all $(\lambda_1,\dots,\lambda_n)\in [0,1]^n$ such that 
    \begin{align} \label{ineq:matrix-Phi-ribbon-ineq}
        \sum_{i=1}^n \lambda_i \HPhi(\E[\mt{F}|X_i])\leq \HPhi(\mt{F}),
    \end{align}
    for all PSD matrix-valued $\mt F: \prod_{i=1}^n \cX_i\to \PSD^d$. 

    Suppose $\Phi\in \mathscr{F}_M$. 
    The \emph{matrix $\Phi$-ribbon} for $(X_1,\dots,X_n)$ is defined as 
    \begin{align} \label{eq:matrix-Phi-ribbon}
        \fR_{\Phi}^{M}(X_1;\dots;X_n) \defeq\bigcap_{d\geq 1} \fR_{\Phi}^{M,d}(X_1;\dots;X_n). 
    \end{align}
\end{definition}

\begin{proposition}\label{propdpd}
    Suppose $X_1,\dots,X_n$ are finite random variables. Then 
    $\fR_{\Phi}^{M,1}(X_1;\dots;X_n)=\fR_{\Phi}(X_1;\dots;X_n)$. Let $d'\geq d\geq 1$. Then 
    \begin{align*}
        \fR_{\Phi}^{M}(X_1;\dots;X_n) 
        \subseteq 
        \fR_{\Phi}^{M,d'}(X_1;\dots;X_n) 
        \subseteq 
        \fR_{\Phi}^{M,d}(X_1;\dots;X_n). 
    \end{align*} 
\end{proposition}
\begin{proof}
    Note that $\mathscr{F}_M^1=\mathscr{F}$, and in dimension 1, it reduces to the scalar case. Thus $\fR_{\Phi}^{M,1}(X_1;\dots;X_n)=\fR_{\Phi}(X_1;\dots;X_n)$. By the definition in \eqref{eq:matrix-Phi-ribbon}, we immediately have $\fR_{\Phi}^{M}(X_1;\dots;X_n) 
        \subseteq 
        \fR_{\Phi}^{M,d'}(X_1;\dots;X_n)$ for any $d'$. Since $d'\geq d$, we can embed an arbitrary $d$-dimensional matrix $\mt F$ into a $d'$-dimensional matrix as 
        \begin{align*}
            \mt F'=\begin{bmatrix}
                \mt F & \mt 0\\
                \mt 0 & c\mt I_{d'-d}
            \end{bmatrix}
        \end{align*}
        where $c\in \dom(\Phi)$. By the definition of the matrix $\Phi$-entropy, $\HPhi(\mt F')=\frac{d}{d'}\HPhi(\mt F)$. 
        
        Now take $(\lambda_1,\dots,\lambda_n)\in \fR_{\Phi}^{M,d'}(X_1;\dots;X_n)$. We have  
        \begin{align*}
            \sum_{i=1}^n \lambda_i \HPhi(\E[\mt{F}'|X_i])\leq \HPhi(\mt{F}').
        \end{align*}
        Cancelling the common factor \(d/d'\) from both sides, we obtain  
        \begin{align*}
            \sum_{i=1}^n \lambda_i \HPhi(\E[\mt{F}|X_i])\leq \HPhi(\mt{F}).
        \end{align*}
        Therefore, $(\lambda_1,\dots,\lambda_n)\in\fR_{\Phi}^{M,d}(X_1;\dots;X_n)$. Thus $\fR_{\Phi}^{M,d'}(X_1;\dots;X_n) 
        \subseteq 
        \fR_{\Phi}^{M,d}(X_1;\dots;X_n)$. 
\end{proof}

Suppose $(X^n,\mt F)$ is a classical-quantum ensemble. 
In the settings of Remark \ref{remark-matrix-mi-form}, Equation \eqref{ineq:matrix-Phi-ribbon-ineq} becomes
\begin{align}
    \sum_{i=1}^n \lambda_i I(\mt F;X_i)  
    \leq I(\mt F;X^n).\label{mi-form-u-quantum}
\end{align}
When $\Phi(t)=t\log t$, we will also call $\fR_{\Phi}^{M,d}(X_1;\dots;X_n)$ the $d$-dimensional matrix hypercontractivity ribbon.

\begin{remark}
Consider the special case of $n=3$ and $\lambda_1=\lambda_2=\lambda_3=1/2$. Then, the condition in \eqref{mi-form-u-quantum}  is related to the key agreement problem in information theoretic security. It is equivalent to the condition that a sophisticated upper bound on the  secrecy capacity involving a quantum auxiliary system meets a trivial upper bound, see \cite[Remark 1\&4]{abin2025source}.
\end{remark}

The following $\Phi$-entropy inequalities are the key ingredients in the tensorization and data-processing proofs. The proofs follow similar lines as in \cite{beigi2018phi}. 
\begin{lemma} \label{lem:matrix-Phi-ineq}
    Suppose $\Phi\in \mathscr{F}_M^d$. 
    \begin{enumerate}[\rm (i)]
        \item \label{matrix-Phi-ineq-1} Assume $X$ and $Y$ are independent random variables and $\mt F=\mt F(X,Y)$ is a $d$-dimensional PSD matrix-valued function. Then 
        \begin{align*}
            \HPhi(\mt F|X)\geq \HPhi(\E[\mt F|Y]).
        \end{align*}
        \item \label{matrix-Phi-ineq-2} More generally, suppose $X,Y,Z$ are three random variables satisfying the Markov chain condition $X\rightarrow Z \rightarrow Y$ and $\mt F=\mt F(X,Y,Z)$ is a $d$-dimensional PSD matrix-valued function. Then 
        \begin{align*}
             \HPhi(\mt F|XZ)\geq \HPhi(\E[\mt F|YZ]|Z).
        \end{align*}
    \end{enumerate}
\end{lemma}

We remark that Lemma~\ref{lem:matrix-Phi-ineq} (i) was previously established in \cite[Lemma 4.3]{tropp2014subadditivity}. For completeness, we also provide the proof later.

Then we have the following results for the matrix $\Phi$-ribbon. 
\begin{proposition}
\label{prop:matrix-phi-ribbon-properties}
    For $\Phi\in \mathscr{F}_M^d$, we have 
    \begin{enumerate}[\rm (i)]
        \item $\{(\lambda_1,\dots,\lambda_n)\in[0,1]^n:\sum_{i=1}^n \lambda_i \leq 1\}$ is always a subset of the $d$-dimensional matrix $\Phi$-ribbon for any $p_{X^n}$. Moreover, when $X_1=\dots=X_n$ and these random variables are nonconstant, the above set is exactly the $d$-dimensional matrix $\Phi$-ribbon. 
        \item If $X_1,\dots,X_n$ are mutually independent, the $d$-dimensional matrix $\Phi$-ribbon is exactly $[0,1]^n$. 
    \end{enumerate}
\end{proposition}

\begin{theorem} \label{thm:tensor-dpi}
    For any $\Phi\in \mathscr{F}_M^d$, the $d$-dimensional matrix $\Phi$-ribbon satisfies tensorization and data processing as follows:
    \begin{enumerate}[\rm (i)]
        \item {\rm (Tensorization)} If $(X_1,Y_1)$ is independent of $(X_2,Y_2)$, then
        \begin{align*}
            \fR_{\Phi}^{M,d}(X_1X_2;Y_1Y_2) =\fR_{\Phi}^{M,d}(X_1;Y_1)\cap \fR_{\Phi}^{M,d}(X_2;Y_2).
        \end{align*}
        \item {\rm (Data processing)} If $(X_2,Y_2)$ is generated from $(X_1,Y_1)$ by $p(x_2,y_2|x_1,y_1)=p(x_2|x_1)p(y_2|y_1)$, then 
        \begin{align*}
            \fR_{\Phi}^{M,d}(X_1;Y_1) \subseteq \fR_{\Phi}^{M,d}(X_2;Y_2). 
        \end{align*}
    \end{enumerate} 
\end{theorem}
By taking intersections, the tensorization and the data processing also hold for the matrix $\Phi$-ribbon. 

The proofs can be found in Appendix \ref{app:proof-matrix-phi}.

\begin{proposition}[Supermodularity of matrix $\Phi$-entropy]
\label{lem:matrix-Phi-entropy-supermodularity}
Let \(X_1,\ldots,X_n\) be mutually independent finite random variables, let
\(\Phi\in\mathscr F_M^d\), and let
\begin{align*}
\mt F
=
\mt F(X_1,\ldots,X_n)
\in
\PSD^d.
\end{align*}
For each \(A\subseteq[n]\), define
\begin{align*}
h_{\mt F}(A)
\defeq
\HPhi\left(
\E[\mt F| X_A]
\right).
\end{align*}
Then \(h_{\mt F}\) is supermodular. That is, for every
\(A,B\subseteq[n]\),
\begin{align}
h_{\mt F}(A)
+
h_{\mt F}(B)
\leq
h_{\mt F}(A\cup B)
+
h_{\mt F}(A\cap B).
\label{eq:matrix-Phi-entropy-supermodularity}
\end{align}

\end{proposition}

\begin{proof}
Fix \(A,B\subseteq[n]\), and define the pairwise disjoint sets
\begin{align*}
C
\defeq
A\cap B,
\qquad
U
\defeq
A\setminus B,
\qquad
V
\defeq
B\setminus A.
\end{align*}
Thus,
\begin{align*}
A=C\cup U,
\qquad
B=C\cup V,
\qquad
A\cup B=C\cup U\cup V.
\end{align*}
Since \(X_1,\ldots,X_n\) are mutually independent, $X_U\to X_C\to X_V$ is a Markov chain. Since \(\Phi\in\mathscr F_M^d\), applying Lemma~\ref{lem:matrix-Phi-ineq} (ii) with $X=X_U$, $Z=X_C$, $Y=X_V$ and $\mt F=\E[\mt F|X_CX_UX_V]$, we obtain  
\begin{align*}
    \HPhi(\E[\mt F|X_CX_UX_V]|X_UX_C)
    \geq \HPhi(\E[\mt F|X_CX_V]|X_C). 
\end{align*}
Applying the chain rule to both sides, we have 
\begin{align*}
    \HPhi(\E[\mt F|X_CX_UX_V]) -\HPhi(\E[\mt F|X_CX_U])
    \geq 
    \HPhi(\E[\mt F|X_CX_V])-\HPhi(\E[\mt F|X_C]), 
\end{align*}
which is exactly \eqref{eq:matrix-Phi-entropy-supermodularity}.

\end{proof}

\begin{theorem}
\label{thm:matrix-t2-ribbon-dimension-independent}
For any \(X_1,\ldots,X_n\) and any \(d\geq1\),
\begin{align*}
\fR_{t^2}^{M,d}(X_1;\ldots;X_n)
=
\fR_{t^2}^{M,1}(X_1;\ldots;X_n).
\end{align*} 
\end{theorem}

\begin{proof}
By Proposition \ref{propdpd}, we have \begin{align*}
\fR_{t^2}^{M,d}(X_1;\ldots;X_n)
\subseteq
\fR_{t^2}^{M,1}(X_1;\ldots;X_n).
\end{align*}
We prove the reverse inclusion.
Suppose that
\begin{align*}
\lambda
\in
\fR_{t^2}^{M,1}(X_1;\ldots;X_n).
\end{align*}
Let
\begin{align*}
\mt F:
\prod_{i=1}^n\cX_i
\longrightarrow
\PSD^d
\end{align*}
be arbitrary, and define
\begin{align*}
\mt G
\defeq
\mt F-\E[\mt F].
\end{align*}
Let \(\{\mt B_r\}_{r=1}^{d^2}\) be an orthonormal basis of the real vector space \(\sa^d\) with respect to the normalized Hilbert--Schmidt inner product 
\begin{align*}
\langle\mt A,\mt B\rangle_{\mathrm{HS}}
\defeq
\frac{1}{d}\operatorname{tr}(\mt A\mt B).
\end{align*}
Write
\begin{align*}
\mt G(X_1,\ldots,X_n)
=
\sum_{r=1}^{d^2}
g_r(X_1,\ldots,X_n)\mt B_r,
\qquad
g_r
\defeq
\langle\mt G,\mt B_r\rangle_{\mathrm{HS}}.
\end{align*}
Each \(g_r=g_r(X_1,\ldots,X_n)\) is real-valued. Since
\(\E[\mt G]=\mt 0\), we have \(\E[g_r]=0\) for every \(r\in[d^2]\).

Note that 
\begin{align*}
H_{t^2}(\mt F)
&=
\E\left[
\frac{1}{d}
\operatorname{tr}(\mt G^2)
\right]
\\
&=
\sum_{r=1}^{d^2}
\E[g_r^2]
\\
&=
\sum_{r=1}^{d^2}
\Var(g_r).
\end{align*}
Moreover, for every \(i\in[n]\),
\begin{align*}
\E[\mt G| X_i]
=
\sum_{r=1}^{d^2}
\E[g_r| X_i]\mt B_r,
\end{align*}
and hence
\begin{align*}
H_{t^2}
\left(
\E[\mt F| X_i]
\right)
=
\sum_{r=1}^{d^2}
\Var\left(
\E[g_r| X_i]
\right).
\end{align*}

Since the alphabets are finite, for every \(r\in[d^2]\), there exists a constant \(c_r\) such that \(g_r+c_r\geq0\). Moreover,
\begin{align*}
\Var(g_r+c_r)
&=
\Var(g_r),
\\
\Var\left(
\E[g_r+c_r| X_i]
\right)
&=
\Var\left(
\E[g_r| X_i]
\right).
\end{align*}
Since
\(\lambda\in\fR_{t^2}^{M,1}(X_1;\ldots;X_n)\), it follows that, for every \(r\in[d^2]\),
\begin{align*}
\sum_{i=1}^n
\lambda_i
\Var\left(
\E[g_r| X_i]
\right)
\leq
\Var(g_r).
\end{align*}
Summing over \(r\) gives
\begin{align*}
&
\sum_{i=1}^n
\lambda_i
H_{t^2}
\left(
\E[\mt F| X_i]
\right)
\\
&\qquad=
\sum_{r=1}^{d^2}
\sum_{i=1}^n
\lambda_i
\Var\left(
\E[g_r| X_i]
\right)
\\
&\qquad\leq
\sum_{r=1}^{d^2}
\Var(g_r)
\\
&\qquad=
H_{t^2}(\mt F).
\end{align*}
Since \(\mt F\) is arbitrary,
\begin{align}
\fR_{t^2}^{M,1}(X_1;\ldots;X_n)
\subseteq
\fR_{t^2}^{M,d}(X_1;\ldots;X_n).
\label{eq:t2-scalar-to-matrix}
\end{align}

\end{proof}

It is known in \cite[Theorem 23]{beigi2018phi} that 
\begin{align} \label{HC-MC}
    \fR_{\Phi}^{M,1}(X_1;\dots;X_n)
        \subseteq \fR_{t^2}^{M,1}(X_1;\dots;X_n). 
\end{align}
We extend it to the matrix version. 
\begin{corollary}
        Suppose $d\geq 1$ and $\Phi\in \mathscr{F}_M^d$. Then for any $X_1,\dots,X_n$, 
    \begin{align*}
        \fR_{\Phi}^{M,d}(X_1;\dots;X_n)
        \subseteq \fR_{t^2}^{M,d}(X_1;\dots;X_n). 
    \end{align*}
\end{corollary}
\begin{proof}
    First, by Proposition \ref{propdpd} we have 
    \begin{align*}
        \fR_{\Phi}^{M,d}(X_1;\dots;X_n)
        \subseteq \fR_{\Phi}^{M,1}(X_1;\dots;X_n). 
    \end{align*}
    By \eqref{HC-MC}, 
    \begin{align*}
        \fR_{\Phi}^{M,1}(X_1;\dots;X_n)
        \subseteq \fR_{t^2}^{M,1}(X_1;\dots;X_n). 
    \end{align*}
    By Theorem~\ref{thm:matrix-t2-ribbon-dimension-independent}, we have 
    \begin{align*}
        \fR_{t^2}^{M,1}(X_1;\dots;X_n)
        =\fR_{t^2}^{M,d}(X_1;\dots;X_n). 
    \end{align*}
    Combining all above, we obtain the desired inclusion. 
\end{proof}

\section{Strong data processing inequality constants for scalar and matrix entropies} \label{sec:SDPI}
    The scalar strong data processing inequalities are well studied \cite{raginsky2016strong,beigi2018phi}. In this section, we show that many scalar results can be generalized to the matrix case. 
\subsection{Basic properties}
\begin{definition}
    Suppose $\Phi:\cI\to \R$ is convex where $\cI\subseteq \R_+$. Define the \emph{$d$-dimensional matrix $\Phi$-SDPI constant} to be 
    \begin{align*}
        \eta_{\Phi}^{M,d}(X;Y) &\defeq \sup \frac{\HPhi(\E[\mt F(X)|Y])}{\HPhi(\mt F(X))}.
    \end{align*}
    where the supremum is over all $d$-dimensional PSD matrix-valued $\mt F(X)$ such that $\spec(\mt F(x))\subseteq \dom(\Phi)$ for all $x$ and $\HPhi(\mt F(X))>0$. Then define the \emph{matrix $\Phi$-SDPI constant} to be 
    \begin{align*}
        \eta_{\Phi}^{M}(X;Y) &\defeq \sup_{d\geq 1} \eta_{\Phi}^{M,d}(X;Y).
    \end{align*}
    When $X$ is constant or $\Phi$ is affine, we define $\eta_{\Phi}^{M,d}(X;Y)=0$. 
\end{definition}

By the chain rule and the nonnegativity of matrix $\Phi$-entropy, 
\begin{align*}
    \HPhi \left(\mt F\right)\geq \HPhi \left(\E[\mt F|Y]\right). 
\end{align*}
Thus, $\eta_{\Phi}^{M,d}$ and $\eta_{\Phi}^{M}$ are well defined and $\eta_{\Phi}^{M}\in [0,1]$. By definition and using the manipulation in the proof of Proposition \ref{propdpd}, it is immediate that for $d'\geq d\geq 1$, we have 
\begin{align} \label{eq:order-Phi-SDPI}
    \eta_{\Phi}^{M,d}
    \leq 
    \eta_{\Phi}^{M,d'}
    \leq 
    \eta_{\Phi}^{M}.
\end{align}

The matrix SDPI constants have applications in studying the exponential decay in the quantum Markov semigroup \cite{cheng2017exponential}

The maximal correlation of $X$ and $Y$ is defined as (see \cite{hirschfeld1935connection,gebelein1941statistische,renyi1959new,renyi1959measures,gacs1973common})
\begin{align*}
    \rho(X;Y)
    \defeq \sup \E[f(X)g(Y)]
\end{align*}
where the supremum is over all $f,g$ such that $\E[f]=\E[g]=0$ and $\Var(f)=\Var(g)=1$. If $X$ or $Y$ is constant, define $\rho(X;Y)=0$. Alternatively, the maximal correlation is defined as 
\begin{align} \label{MC2}
    \rho(X;Y)
    =
    \sup \sqrt{\frac{\E[\E[f|Y]^2]}{\E[f^2]}} 
\end{align}
where the supremum is over all nonconstant $f(X)$ such that $\E[f]=0$. It is known that the squared maximal correlation is a $\Phi$-SDPI constant when $\Phi(t)=t^2$ , i.e., $\rho(X;Y)^2=\eta_{t^2}(X;Y)$ \cite{beigi2018phi}. 

In the scalar case, it is known that the squared maximal correlation, $\rho^2(X;Y)$, lower bounds every $\Phi$-SDPI constant for $\Phi\in \mathscr F$, i.e., $\eta_{\Phi}(X;Y)\geq \rho^2(X;Y)$ \cite{raginsky2016strong}. We have similar results in the matrix case.

\begin{theorem}
\label{thm:matrix-Phi-SDPI-data-processing-tensorization}
Fix \(d\geq 1\) and let \(\Phi\) be convex on $[0,\infty)$. 

Then the following properties hold.
\begin{enumerate}[\rm (i)]

\item \textbf{Data processing.}
Suppose that $X\to Y\to Z$ 
forms a Markov chain. Then
\begin{align}
\eta_{\Phi}^{M,d}(X;Z)
\leq
\eta_{\Phi}^{M,d}(X;Y)
\eta_{\Phi}^{M,d}(Y;Z).
\label{eq:matrix-Phi-SDPI-data-processing}
\end{align}
In particular,
\begin{align}
\eta_{\Phi}^{M,d}(X;Z)
\leq
\min\left\{
\eta_{\Phi}^{M,d}(X;Y),
\eta_{\Phi}^{M,d}(Y;Z)
\right\}. \label{eq:matrix-Phi-SDPI-data-processing-2}
\end{align}

\item \textbf{Tensorization.} Suppose in addition that \(\Phi\in\mathscr F_M^d\). 
Let $(X_1,Y_1),\ldots,(X_n,Y_n)$ 
be mutually independent pairs. Then
\begin{align}
\eta_{\Phi}^{M,d}
\left(
X_1\cdots X_n;
Y_1\cdots Y_n
\right)
=
\max_{i\in[n]}
\eta_{\Phi}^{M,d}(X_i;Y_i).
\label{eq:matrix-Phi-SDPI-tensorization}
\end{align}

\end{enumerate}
\end{theorem}

The proof is in Appendix \ref{app:proof-matrix-Phi-SDPI-data-processing-tensorization}.

\begin{theorem} \label{thm:seperation-SDPI}
    Suppose 
    \begin{align*}
    \Phi(t)
    \defeq
    \frac{1}{2}t\sqrt{1+t^2}
    +
    t\ln\left(
    \frac{t}{1+\sqrt{1+t^2}}
    \right)
    -
    \frac{1}{2}\operatorname{arsinh}(t),
    \qquad
    t>0, 
    \end{align*}
    and 
    \begin{align*}
    p_{XY}
    =
    \frac{1}{8}
    \begin{pmatrix}
    1 & 0 & 1 & 0\\
    1 & 0 & 0 & 1\\
    0 & 1 & 1 & 0\\
    0 & 1 & 0 & 1
    \end{pmatrix}.
    \end{align*}
    Then $\Phi\in \mathscr{F}\setminus \mathscr{F}_M^2$ and $\eta_{\Phi}^{M,2}>1/2=\eta_{\Phi}^{M,1}$. 
\end{theorem}

The proof is in Appendix \ref{app:proof-seperation-SDPI}.

\begin{remark} \label{rmk:separation}
    In Theorem~\ref{thm:seperation-SDPI}, the function $\Phi$ belongs to $\mathscr{F}$ but does not belong to $\mathscr{F}_M^2$. It remains open whether there exists $\Phi\in\mathscr{F}_M^{d'}$ and $p_{XY}$ such that $\eta_{\Phi}^{M,d'}>\eta_{\Phi}^{M,d}$ where $d'>d$. 
\end{remark}

\subsection{Matrix $\Phi$-entropy inequalities}

We extend the scalar $\Phi$-entropy bounds in Theorem~\ref{thm:scalar-DPhi-ineq} and Theorem~\ref{thm:scalar-DPhi-ineq-delta} to matrix-valued functions. These estimates are the key ingredients in proving the equality between the $d$-dimensional matrix $\Phi$-SDPI constant and the boundary slope of the $d$-dimensional matrix $\Phi$-ribbon.

\begin{restatable}{theorem}{matrixPhiineq}
\label{thm:matrix-Phi-conditional-stability}
    Let $d\geq 1$ be an integer. 
    Suppose $\Phi\in\mathscr{F}_M^d$. 
    Let $(X,Y)$ have distribution $p_{XY}$ on the finite set
    $\cX\times\cY$, and denote $\mu\defeq \min\{p_{XY}(x,y):p_{XY}(x,y)>0\}$. 
    For any function $\mt F:\cX\times\cY\to\sa^d$ with spectra
    inside $\dom(\Phi)$, write $\mt F_X\defeq\E[\mt F|X]$, $\mt F_Y\defeq\E[\mt F|Y]$ and $\widetilde{\mt F}_Y\defeq\E[\mt F_X|Y]$. 
    Then
    \begin{align*}
        \HPhi(\mt F_Y)
        &\leq
        \HPhi(\widetilde{\mt F}_Y)+\HPhi(\mt F|X)
        +\frac{\pi}{\mu}
        \sqrt{\HPhi(\mt F_X)\HPhi(\mt F|X)}.
    \end{align*}
    Consequently, for any $\delta>0$, we have 
    \begin{align*}
        \HPhi(\mt F_Y)
        &\leq
        \HPhi(\widetilde{\mt F}_Y)
        +\delta \HPhi(\mt F_X)
        +K_{\delta} 
        \HPhi(\mt F|X)
    \end{align*}
    where $K_\delta\defeq 1+\frac{\pi^2}{4\delta\mu^2}$. 
\end{restatable}

The proof is in Appendix~\ref{app:proof-matrix-SDPI-ribbon}. The proof follows exactly the idea in Theorem~\ref{thm:scalar-DPhi-ineq} and Theorem~\ref{thm:scalar-DPhi-ineq-delta}.

\begin{theorem}
\label{thm:matrix-SDPI-ribbon}
Fix \(d\geq 1\). Let \(\Phi\in\mathscr F_M^d\). 
Then
\[
\eta_\Phi^{M,d}(X;Y)
=\inf\frac{1-\lambda_1}{\lambda_2}
\]
where the infimum is over all $(\lambda_1,\lambda_2)\in\fR_\Phi^{M,d}(X;Y)$ with $\lambda_2>0$. 
\end{theorem}

\begin{proof}
If $X$ is constant, then $\eta_\Phi^{M,d}(X;Y)=0$ by convention
and $\fR_\Phi^{M,d}(X;Y)=[0,1]^2$, so the identity is immediate.
Assume henceforth that $X$ is nonconstant. Write
\[
\eta:=\eta_\Phi^{M,d}(X;Y),
\qquad
\alpha:=
\inf_{\substack{
(\lambda_1,\lambda_2)\in\fR_\Phi^{M,d}(X;Y)\\
\lambda_2>0
}}
\frac{1-\lambda_1}{\lambda_2}.
\]

First, let $\mt F=\mt F(X)$ with $\HPhi(\mt F)>0$. For $(\lambda_1,\lambda_2)\in \fR_{\Phi}^{M,d}$ with $\lambda_2>0$, we have 
\[
\lambda_1\HPhi(\mt F)
+\lambda_2\HPhi\bigl(\E[\mt F| Y]\bigr)
\leq\HPhi(\mt F).
\]
Then 
\begin{align*}
    \HPhi\bigl(\E[\mt F| Y]\bigr)
    \leq \frac{1-\lambda_1}{\lambda_2} \HPhi(\mt F). 
\end{align*}
Taking the supremum over $\mt F(X)$ and then the infimum over
ribbon points, we have $\eta\leq\alpha$.

For the reverse direction, fix $\delta>0$ and take $K_\delta$
from Theorem~\ref{thm:matrix-Phi-conditional-stability}.
For every admissible $\mt F=\mt F(X,Y)$, set
$\mt A:=\E[\mt F| X]$. 
Then we have 
\begin{align}
\HPhi\bigl(\E[\mt F| Y]\bigr)
&\leq\HPhi\bigl(\E[\mt A| Y]\bigr)
+\delta\HPhi(\mt A)+K_\delta\HPhi(\mt F| X)\nonumber\\
&\leq(\eta+\delta)\HPhi(\mt A)
+K_\delta\HPhi(\mt F| X), \label{eq:apply-lemma-A-B}
\end{align}
where the first inequality follows from Theorem~\ref{thm:matrix-Phi-conditional-stability}, and the second inequality follows from the definition of $\eta$.

Choose an integer $n>\max\{1,\eta+\delta,K_\delta\}$. Dividing both sides of \eqref{eq:apply-lemma-A-B} by $n$, and adding $(1-(\eta+\delta)/n)\HPhi(\mt A)$ to both sides, we obtain 
\begin{align*}
    \left(1-\frac{\eta+\delta}{n}\right)
    \HPhi\bigl(\mt A\bigr)
    +\frac1n\HPhi\bigl(\E[\mt F| Y]\bigr)
    &\leq\HPhi(\mt A)
    +\frac{K_\delta}{n}\HPhi(\mt F| X)\\
    &\leq\HPhi(\mt A)+\HPhi(\mt F| X)\\
    &=\HPhi(\mt F),
\end{align*}
where the second inequality follows from $n>K_{\delta}$ and the last equality follows from the chain rule. 
Hence, 
\[
\left(1-\frac{\eta+\delta}{n},\frac1n\right)
\in\fR_\Phi^{M,d}(X;Y),
\]
and then by the definition of $\alpha$, we obtain $\alpha\leq\eta+\delta$.
Letting $\delta\downarrow0$ gives $\alpha\leq\eta$. It finishes the proof. 
\end{proof}

\begin{corollary} \label{cor:Phi-SDPI-MC}
    Let $d\geq 1$. For every $\Phi\in\mathscr F_M^d$, 
    \begin{align} \label{sdpi-mc}
        \eta_\Phi^{M,d}(X;Y)\geq \rho^2(X;Y).
    \end{align}
    Moreover, 
    \begin{align} \label{sdpi-mc-dim-ind}
        \eta_{t^2}^{M,d}(X;Y)=\rho^2(X;Y). 
    \end{align}
\end{corollary}

\begin{proof}
By \eqref{eq:order-Phi-SDPI}, we have $\eta_\Phi^{M,d}(X;Y)\geq \eta_\Phi^{M,1}(X;Y)$. 
Moreover, it is shown in \cite[Theorem 3.6]{raginsky2016strong} that $\eta_\Phi^{M,1}(X;Y)\geq \rho^2(X;Y)$. Then we obtain \eqref{sdpi-mc}. 

For the second part, note that by  Theorem~\ref{thm:matrix-t2-ribbon-dimension-independent}, 
\begin{align*}
\fR_{t^2}^{M,d}(X_1;\ldots;X_n)
=
\fR_{t^2}^{M,1}(X_1;\ldots;X_n).
\end{align*} 
Moreover, by Theorem~\ref{thm:matrix-SDPI-ribbon}, 
\[
\eta_{t^2}^{M,d}(X;Y)
=\inf\frac{1-\lambda_1}{\lambda_2},
\]
where the infimum is over all $(\lambda_1,\lambda_2)\in\fR_{t^2}^{M,d}(X;Y)$ with $\lambda_2>0$. 
Thus, we obtain \eqref{sdpi-mc-dim-ind}. 
\end{proof}

\subsection{Doubly symmetric binary source}
Suppose $(X,Y)$ is a doubly symmetric binary source (DSBS) with parameter $\lambda$, denoted as $\DSBS(\lambda)$, where the joint distribution is given by 
\begin{align*}
    p_{XY}=\begin{pmatrix}
        \frac{1+\lambda}{4} &\frac{1-\lambda}{4}\\
        \frac{1-\lambda}{4} &\frac{1+\lambda}{4}
    \end{pmatrix},
\end{align*}
where $\lambda\in [0,1]$. 

For the scalar case, it is shown that $\eta_{\Phi}(X;Y)=\lambda^2$ for $\Phi\in \mathscr{F}$ \cite{beigi2018phi}. We extend the proof to the matrix case.  

\begin{proposition} \label{prop:M-Phi-SDPI}
    Let $d\geq 1$ and let $\Phi \in \mathscr{F}_M^d$ be three times Fr\'echet differentiable. 
    If $(X,Y)$ are distributed according to $\DSBS(\lambda)$, we have
    \begin{align*}
        \eta_{\Phi}^{M,d}(X;Y)=\lambda^2. 
    \end{align*}
\end{proposition}
The proof is in Appendix \ref{app:proof-SDPI}.

\subsection{Sums of i.i.d. random variables}

Let \(X_1,\ldots,X_n\) be nondegenerate i.i.d.\ finite real-valued random variables. For \(0\leq \ell\leq n\), define
\begin{align*}
S_\ell
\defeq
\sum_{i=1}^{\ell}X_i,
\end{align*}
where \(S_0\defeq 0\). It is shown in \cite{dembo2001remarks} that for any $1\leq m\leq n$, $\rho^2(S_n,S_m)=m/n$ and in \cite{kamath2015strong} that $\eta_{t\log t}=m/n$. Moreover, it is shown in \cite{beigi2018phi} that for all $\Phi\in \mathscr{F}$, $\eta_{\Phi}=m/n$. Here we give the matrix version, mimicking the proof idea in \cite{beigi2018phi}. 

\begin{theorem}
\label{thm:matrix-Phi-SDPI-partial-sums}
For every \(1\leq m\leq n\), every \(d\geq 1\), and every 
\(\Phi\in\mathscr F_M^d\), we have
\begin{align}
\eta_\Phi^{M,d}(S_n;S_m)
=
\frac{m}{n}.
\label{eq:matrix-Phi-SDPI-partial-sums}
\end{align}
\end{theorem}

\begin{proof}

By Corollary~\ref{cor:Phi-SDPI-MC}, we know that 
\begin{align*}
    \eta_\Phi^{M,d}(S_n;S_m)
    \geq \rho^2(S_n;S_m)=\frac{m}{n}.
\end{align*}
Thus, it remains to prove the reverse inequality. 

Fix a matrix-valued function $\mt F:\supp(S_n)\to \PSD^d$ 
such that \(\HPhi(\mt F)>0\). For \(A\subseteq[n]\), define
\begin{align*}
g(A)
\defeq
\HPhi\left(
\E[\mt F(S_n)| X_A]
\right).
\end{align*}
Because \(X_1,\ldots,X_n\) are independent and
\(\Phi\in\mathscr F_M^d\), by Proposition~\ref{lem:matrix-Phi-entropy-supermodularity}, \(g\) is supermodular. In particular, for every
\(A\subseteq B\subseteq[n]\) and every \(i\notin B\),
\begin{align}
g(A\cup\{i\})-g(A)
\leq
g(B\cup\{i\})-g(B).
\label{eq:supermodularity-marginal-increments}
\end{align}

For \(0\leq\ell\leq n\), define
\begin{align*}
c_\ell
\defeq
\HPhi\left(
\E[\mt F(S_n)| X_1,\ldots,X_\ell]
\right), 
\end{align*}
with the convention $c_0=\HPhi\left(
\E[\mt F(S_n)]\right)$. 
Since \(\mt F(S_n)\) is invariant under permutations of
\(X_1,\ldots,X_n\), the quantity \(g(A)\) depends only on \(\lvert A\rvert\). Moreover, writing
\begin{align*}
S_n
=
S_\ell
+
\sum_{i=\ell+1}^{n}X_i
\end{align*}
and using the independence of \(S_\ell\) and
\(\sum_{i=\ell+1}^{n}X_i\), we obtain
\begin{align*}
\E[\mt F(S_n)| X_1,\ldots,X_\ell]
=
\E[\mt F(S_n)| S_\ell].
\end{align*}
Consequently,
\begin{align*}
c_\ell
=
\HPhi\left(
\E[\mt F(S_n)| S_\ell]
\right).
\end{align*}

Applying \eqref{eq:supermodularity-marginal-increments} to $A,B$ such that $|A|=\ell$ and $|B|=\ell+1$, we obtain 
\begin{align*}
    c_{\ell+1}-c_{\ell}\leq c_{\ell+2}-c_{\ell+1}. 
\end{align*}
Thus $\{c_{\ell+1}-c_\ell\}_{0\leq\ell\leq n-1}$ 
is a nondecreasing sequence. Since $c_0=0$ and $c_n=\HPhi(\mt F(S_n))$, 
it follows that 
\begin{align*}
\frac{c_m}{m}
&=
\frac{1}{m}
\sum_{\ell=0}^{m-1}
\left(
c_{\ell+1}-c_\ell
\right)
\\
&\leq
\frac{1}{n}
\sum_{\ell=0}^{n-1}
\left(
c_{\ell+1}-c_\ell
\right)
=
\frac{c_n}{n}.
\end{align*}
Therefore,
\begin{align*}
\HPhi\left(
\E[\mt F(S_n)| S_m]
\right)
\leq
\frac{m}{n}
\HPhi(\mt F(S_n)).
\end{align*}
Taking the supremum over all admissible \(\mt F\) gives
\begin{align}
\eta_\Phi^{M,d}(S_n;S_m)
\leq
\frac{m}{n}.
\label{eq:matrix-Phi-SDPI-partial-sums-upper-bound}
\end{align}

\end{proof}

\begin{corollary}[Contraction of cq mutual information]
\label{cor:cq-mutual-information-partial-sums}
Let \(Q\) be a finite-dimensional quantum system, and suppose that, conditionally on \(X_1,\ldots,X_n\), the state of \(Q\) depends only on \(S_n\). Thus, for each \(s\in\operatorname{supp}(S_n)\), let
\(\mt \rho_Q(s)\) be a density matrix. 
Then
\begin{align}
I(S_m;\mt \rho_Q)
\leq
\frac{m}{n}I(S_n;\mt \rho_Q).
\label{eq:cq-mutual-information-partial-sums}
\end{align}
\end{corollary}

\section{Wirings of no-signaling boxes} \label{sec:wirings}

We will define the matrix $\Phi$-ribbon for no-signaling boxes and then prove it has the tensorization and the data processing properties. This generalizes the results of \cite{beigi2015monotone} and \cite{beigi2018phi}. While the matrix $\Phi$-entropy shares some properties with Shannon's entropy, it does not permit all the same manipulations, and the repetitive use of Lemma \ref{lem:matrix-Phi-ineq} is necessary. 
\begin{definition}
    Given a no-signaling box $p_{AB|XY}$, define the corresponding \emph{$d$-dimensional matrix $\Phi$-ribbon} to be the intersection of the $d$-dimensional matrix $\Phi$-ribbons of its outputs conditioned on all possible inputs, i.e., 
    \begin{align*}
        \fR_{\Phi}^{M,d}(A;B|X;Y)\defeq \bigcap_{x,y} \fR_{\Phi}^{M,d}(A;B|X=x,Y=y).
    \end{align*}

    The \emph{matrix $\Phi$-ribbon} of the no-signaling box $p_{AB|XY}$ is defined as 
    \begin{align*}
        \fR_{\Phi}^{M}(A;B|X;Y)\defeq \bigcap_{d\ge 1} \fR_{\Phi}^{M,d}(A;B|X;Y). 
    \end{align*}
\end{definition}

If a no-signaling box $p_{AB|XY}$ can be simulated by the box with $p_{A'B'|XY}$, given any $x,y$, we have 
$$p_{AA'BB'|X=xY=y}=p_{A'B'|X=xY=y}q_{A|A'X=x}q_{B|B'Y=y}.$$ 
By the data processing property of the matrix $\Phi$-ribbon, 
\begin{align*}
    \fR_{\Phi}^{M,d} (A';B'|X=x,Y=y) \subseteq \fR_{\Phi}^{M,d} (A;B|X=x,Y=y).
\end{align*}
By the definition of the $\Phi$-ribbon for the no-signaling boxes, we then have  
\begin{align*}
    \fR_{\Phi}^{M,d} (A';B'|X;Y) \subseteq \fR_{\Phi}^{M,d} (A;B|X;Y).
\end{align*}
Thus, the $\Phi$-ribbon for no-signaling boxes has the data processing inequality.

Suppose the two parties, Alice and Bob, have the inputs $X_i,Y_i$ for the $i$-th box respectively and the associated outputs $A_i,B_i$ for $i\in [n]$. The $i$-th box is associated with the conditional probability $p_{A_iB_i|X_iY_i}$. The wiring of boxes allows one party to choose which box to use and the corresponding input based on all the past information, including all the boxes that have been used and the corresponding inputs and outputs. At each action, each party chooses a box and an input using only its own local transcript and private randomness. The wiring can be arbitrary, which means that in the same action, the two parties can use different boxes, as shown in Fig.~\ref{fig2}. 

Given a no-signaling box $p_{A'B'|X'Y'}$, Alice and Bob use $n$ no-signaling boxes, i.e., $p_{A_iB_i|X_iY_i}$ for $i\in[n]$, to simulate it, which means that the outputs $A'B'$ depend on all the information generated in the actions with the $n$ boxes $p_{A_iB_i|X_iY_i}$. 

\begin{figure}
\begin{center}
\includegraphics[width=3in]{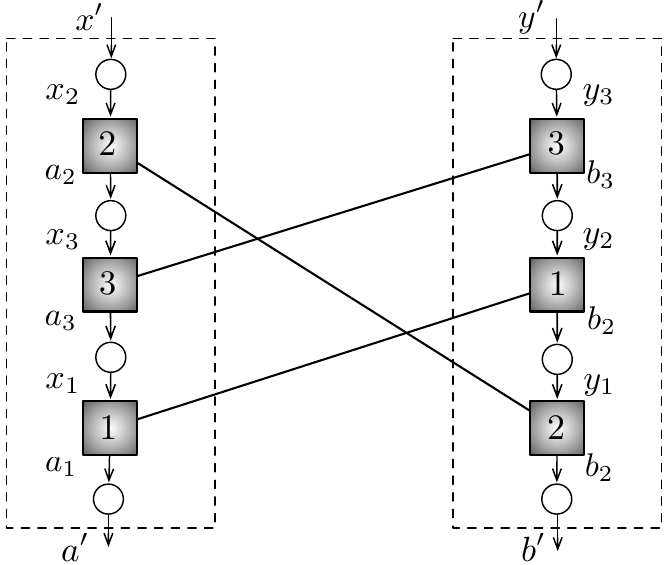}
\caption{\small (\cite{beigi2015monotone}) Due to the no-signaling property, Alice and Bob can choose boxes in different orders. For example, Alice can use the order $2,3,1$ to simulate the result given $x'$. Alice simulates the box $2$ with input $x_2$ determined by $x'$ and gets the result $a_2$. She then uses the box $3$ with input $x_3$ determined by the former output $a_2$ and gets the result $a_3$. Then she uses the box $1$ with input $x_1$ determined by $a_3$ and gets the result $a_1$. Finally she simulates the result $a'$ by applying a stochastic map. Bob can use the order of boxes as $3,1,2$ given $y'$ and follow the same steps. }
\label{fig2}
\end{center}
\end{figure}

Now we state the main theorem for the tensorization of the wiring of no-signaling boxes. 

\begin{theorem} \label{thm:tensorization-for-boxes}
    Let $\Phi\in \mathscr{F}_M^d$. 
    Suppose a no-signaling box $p_{A'B'|X'Y'}$ can be generated from $n$ no-signaling boxes $\{p_{A_iB_i|X_iY_i}\}$ where $i\in[n]$, under wirings. Then we have 
    \begin{align*}
        \bigcap_{i=1}^n \fR_{\Phi}^{M,d} (A_i;B_i|X_i;Y_i)\subseteq \fR_{\Phi}^{M,d} (A';B'|X';Y').
    \end{align*}
    If $\Phi\in \mathscr{F}_M$, we have  
    \begin{align*}
        \bigcap_{i=1}^n \fR_{\Phi}^{M} (A_i;B_i|X_i;Y_i)\subseteq \fR_{\Phi}^{M} (A';B'|X';Y'). 
    \end{align*}
\end{theorem}
The proof is in Appendix \ref{apx:proof-main-theorem}. The proof is similar to \cite{beigi2015monotone}, but it has to be carefully adapted because the original proof used Shannon's mutual information, and not the matrix $\Phi$ entropy.

\section{Ribbons from independence structures} \label{sec:characterization-ribbon}

\subsection{General independence structures}

Given a hypergraph $\cH=(\cV,\cE)$ with $\cV=[n]$ and $\cE\subseteq \cP([n])$, we say that $\cH$ specifies the independence structure of the random variables $(X_1,\dots,X_n)$ if 
\begin{align*}
    \cE
    =
    \left\{S\in \cP([n]): (X_i)_{i\in S} \textrm{ are mutually independent}
    \right\}.
\end{align*}
Conversely, every independence structure can be represented by hypergraphs. Note that $\{i\}\in \cE$ for $i\in [n]$. Moreover, $\cE$ is downward closed: if $S\in \cE$, then for any $T\subseteq S$, $T\in \cE$. When $\cE=\cP([n])$, we say that $\cH$ is complete.

\begin{theorem} \label{thm:Phi-ribbon-general-hypergraph}
    Suppose $\Phi\in \mathscr{F}_M^d$. Let $\cH=([n],\cE)$ be the hypergraph that specifies the independence structure of $X_1,\dots,X_n$. Define 
    $$J_{\cE}\defeq\{\vec{0},\vec{e}_1,\dots,\vec{e}_n\}\cup \{\vec{e}_{S}\}_{S\in \cE}$$
    where $\vec{e}_i\in\mathbb{R}^n$ stands for the $i$-th standard basis vector and $\vec{e}_S=\sum_{i\in S}\vec{e}_i$. Then 
    $$\Conv(J_{\cE})\subseteq \fR_{\Phi}^{M,d}(X_1;\dots;X_n).$$ 
\end{theorem}
\begin{proof}
    Let $\mt F=\mt F(X^n)$ be a $d$-dimensional PSD matrix-valued function. 
    For each $i\in [n]$, we have
    \begin{align} \label{ineq:proof-i}
        \HPhi(\E[\mt F|X_i]) \leq \HPhi(\mt F),
    \end{align}
    which follows from the chain rule and the nonnegativity of the matrix entropy. 
    By Proposition \ref{prop:matrix-phi-ribbon-properties} (ii), for each $S\in \cE$, since $(X_i)_{i\in S}$ are mutually independent, we have 
    \begin{align} \label{ineq:proof-S}
        \sum_{i\in S} \HPhi(\E[\mt F|X_i]) \leq \HPhi(\E[\mt F|X_S])\leq \HPhi(\mt F).
    \end{align}
    Taking the convex combination of \eqref{ineq:proof-i} and  \eqref{ineq:proof-S} for each $i$ and $S$, we obtain the desired inequality. 
\end{proof}

The above convex hull characterization is tight when $\Phi(t)=t\log t$. 

\begin{theorem}
\label{thm:tlogt-arbitrary-independence-converse}
Suppose $\Phi(t)=t\log t$ and $d\geq 1$. Let $\cH=([n],\cE)$ be a hypergraph that is not complete where $\cE$ contains all singletons and is downward closed. Define 
$$J_{\cE}\defeq\{\vec{0},\vec{e}_1,\dots,\vec{e}_n\}\cup \{\vec{e}_{S}\}_{S\in \cE}$$
If $\vec\lambda=(\lambda_1,\dots,\lambda_n)\in [0,1]^n\setminus \Conv(J_{\cE})$, 
then there exist finite random variables $X_1,\ldots,X_n$ satisfying the independence structure of $\cH$, but 
\[
    \vec\lambda\notin \fR_{\Phi}^{M,d}(X_1;\ldots;X_n).
\]
\end{theorem}

The proof is in Appendix \ref{app:proof-tlogt-arbitrary-independence-converse}. The assumption that $\cH$ is not complete is to exclude the trivial case. If $\cH$ is complete, $\Conv(J_{\cE})=[0,1]^n$. Thus there is nothing to prove.

\section{Future work}

\textbf{Separation of matrix SDPI constants.} As discussed in Remark~\ref{rmk:separation}, the function $\Phi$ in Theorem~\ref{thm:seperation-SDPI} belongs to $\mathscr{F}$ but not to $\mathscr{F}_M^2$. It remains open whether the matrix SDPI constants can strictly increase in some matrix dimensions when $\Phi$ is always in the subadditivity class, that is, whether there exists $d'>d$, $\Phi\in\mathscr{F}_M^{d'}$ and $p_{XY}$ for which $\eta_{\Phi}^{M,d'}>\eta_{\Phi}^{M,d}$.

\section*{Acknowledgement}

Theorems~\ref{thm:seperation-SDPI} and \ref{thm:tlogt-arbitrary-independence-converse}  are developed with the assistance of ChatGPT 5.6 Sol. The authors have verified AI-generated content. The authors also use ChatGPT 5.6 Sol for proofreading.

\bibliographystyle{IEEEtran}
\bibliography{refs}

\appendices

\section{Proof of the properties of the matrix \texorpdfstring{$\Phi$}{Phi}-ribbon} \label{app:proof-matrix-phi}

\begin{proof}[Proof of Lemma \ref{lem:matrix-Phi-ineq}]
    (i) By Definition \ref{def:class-FM} (ii), the mapping
    \begin{align*}
        (\mt F(x))_{x\in \cX} \mapsto \nortr \left(\sum_{x}p(x)\Phi(\mt F(x)) -\Phi\left(\sum_{x}p(x) \mt F(x)\right)\right) 
    \end{align*}
    is convex. Then by Jensen's inequality, for any $q_X$ and $\mt F(X,Y)$, we have 
    \begin{align*}
        &\sum_{x}q(x)
        \nortr\Bigg(
        \sum_{y}p(y)\Phi(\mt F(x,y)) -\Phi\left(\sum_{y}p(y) \mt F(x,y)\right)
        \Bigg)\\
        \geq&
        \nortr \left(\sum_{y}p(y)\Phi\left(\sum_x q(x) \mt F(x,y)\right) -\Phi\left(\sum_{x,y}q(x)p(y) \mt F(x,y)\right)\right)
        ,
    \end{align*}
    which is equivalent to
    \begin{align*}
        \HPhi(\mt F|X)\geq \HPhi(\E[\mt F|Y]).
    \end{align*}

    (ii) By (i), fixing $Z=z$, we have the inequalities for the function $\mt G_{XY}^{(z)}(x, y) = \mt F(x,y,z)$. Then taking the average over $z$, we obtain the result. 
\end{proof}

\begin{proof}[Proof of Proposition \ref{prop:matrix-phi-ribbon-properties}]
    \begin{enumerate}[(i)]
        \item By \eqref{matrix-Phi:cond-reduce}, we have 
        \begin{align*}
            \HPhi(\E[\mt F|X_i])\leq \HPhi(\mt F). 
        \end{align*}
        Then for $\sum_{i=1}^n \lambda_i\leq 1$, we have
        \begin{align*}
            \sum_{i=1}^n \lambda_i \HPhi(\E[\mt F|X_i])
            \leq \sum_{i=1}^n \lambda_i \HPhi(\mt F)\leq\HPhi(\mt F).
        \end{align*}
        Thus $\{(\lambda_1,\dots,\lambda_n)\in [0,1]^n:\sum_{i=1}^n \lambda_i\leq 1\}$ is always in the ribbon. 

        Now suppose $X_1=\dots=X_n$ and these random variables are nonconstant. Let $(\lambda_1,\dots,\lambda_n)\in [0,1]^n$ such that $\sum_{i=1}^n \lambda_i>1$. Then for nonconstant $\mt F(X^n)$, we have
        \begin{align*}
            \HPhi(\mt F) = \HPhi(\E[\mt F|X_i])
        \end{align*}
        for all $i\in [n]$. Then 
        \begin{align*}
            \sum_{i=1}^n \lambda_i \HPhi(\E[\mt F|X_i])
            &= \sum_{i=1}^n \lambda_i \HPhi(\mt F)\\
            &>\HPhi(\mt F).
        \end{align*}
        Therefore, this choice of $(\lambda_1,\dots,\lambda_n)$ is not in the matrix $\Phi$-ribbon of $X^n$. Thus 
        \begin{align*}
            \fR_{\Phi}^{M,d}(X_1;\dots;X_n)=
            \left\{(\lambda_1,\dots,\lambda_n)\in [0,1]^n:
            \sum_{i=1}^n \lambda_i\leq 1\right\}.
        \end{align*}
        \item Since $X_1$ and $X_2$ are independent, by Lemma \ref{lem:matrix-Phi-ineq} (i), 
        \begin{align*}
            \HPhi(\E[\mt F|X_1])
            \leq
            \HPhi(\E[\mt F|X_1X_2]|X_2)
        \end{align*}
        By the chain rule \eqref{eqn:matrix-CR2} and rearranging the terms, we get 
        \begin{align*}
            \HPhi(\E[\mt F|X_1])
            +\HPhi(\E[\mt F|X_2])
            \leq \HPhi(\E[\mt F|X_1X_2]). 
        \end{align*}
        By induction, we obtain that 
        \begin{align*}
            \sum_{i=1}^n \HPhi(\E[\mt F|X_i])
            \leq
            \HPhi(\mt F). 
        \end{align*}
    \end{enumerate}
    Thus the matrix $\Phi$-ribbon is $[0,1]^n$. 
\end{proof}

\begin{proof}[Proof of Theorem \ref{thm:tensor-dpi}]
    By letting $\mt F$ be an arbitrary mapping of $(X_1,Y_1)$ or $(X_2,Y_2)$, we have
    \begin{align*}
        \fR_{\Phi}^{M,d}(X_1X_2;Y_1Y_2) \subseteq \fR_{\Phi}^{M,d}(X_1;Y_1)\cap \fR_{\Phi}^{M,d}(X_2;Y_2).
    \end{align*}
    To prove the other direction, let $(\lambda_1,\lambda_2)\in \fR_{\Phi}^{M,d}(X_1;Y_1)\cap \fR_{\Phi}^{M,d}(X_2;Y_2)$.  Then it suffices to show that 
    \begin{align} \label{eq:tensor-direction}
        \HPhi(\mt F)\geq \lambda_1 \HPhi(\E[\mt F|X_1X_2]) + \lambda_2 \HPhi(\E[\mt F|Y_1Y_2]).
    \end{align}
    Since $\E[\mt F|X_1 Y_1]$ is a PSD matrix-valued mapping of $(X_1,Y_1)$ and $(\lambda_1,\lambda_2)\in \fR_{\Phi}^{M,d}(X_1;Y_1)$, we have
    \begin{align} \label{eq:tens-EX1Y1}
         \HPhi(\E[\mt F|X_1 Y_1])\geq \lambda_1 \HPhi(\E[\mt F|X_1]) + \lambda_2 \HPhi(\E[\mt F|Y_1]).
    \end{align}
    Moreover, by fixing $X_1=x$, $Y_1=y$, $\mt F$ is a mapping of $(X_2,Y_2)$. 
    Since $(\lambda_1,\lambda_2)\in \fR_{\Phi}^{M,d}(X_2;Y_2)$, taking the expectation over $x,y$, we have
    \begin{align} \label{eq:tens-X1Y1}
        \HPhi(\mt F|X_1 Y_1)\geq \lambda_1 \HPhi(\E[\mt F|X_1 X_2 Y_1]|X_1 Y_1) +\lambda_2 \HPhi(\E[\mt F|X_1 Y_1 Y_2]|X_1 Y_1).
    \end{align}
    Combining \eqref{eq:tens-EX1Y1} and \eqref{eq:tens-X1Y1}, we have
    \begin{align} \label{eq:F-la12}
        \HPhi(\mt F) 
        &=\HPhi(\E[\mt F|X_1 Y_1])+\HPhi(\mt F|X_1 Y_1) \nonumber\\
        &\geq \lambda_1 \left(\HPhi(\E[\mt F|X_1])+\HPhi(\E[\mt F|X_1 X_2 Y_1]|X_1 Y_1)\right)
        +\lambda_2 \left(\HPhi(\E[\mt F|Y_1])+\HPhi(\E[\mt F|X_1 Y_1 Y_2]|X_1 Y_1)\right). 
    \end{align}
    Note that we have the Markov chains $X_2\rightarrow X_1\rightarrow Y_1$ and $X_1\rightarrow Y_1\rightarrow Y_2$. By Lemma \ref{lem:matrix-Phi-ineq} (ii), we have
    \begin{align} \label{eq:X1X2Y1}
        \HPhi(\E[\mt F|X_1 X_2 Y_1]|X_1 Y_1)\geq \HPhi(\E[\mt F|X_1 X_2]|X_1),
    \end{align}
    and
    \begin{align} \label{eq:X1Y1Y2}
        \HPhi(\E[\mt F|X_1 Y_1 Y_2]|X_1 Y_1)\geq \HPhi(\E[\mt F|Y_1Y_2]|Y_1),
    \end{align}
    Combining \eqref{eq:F-la12}, \eqref{eq:X1X2Y1} and \eqref{eq:X1Y1Y2}, we obtain 
    \begin{align*}
        \HPhi(\mt F) 
        &\geq \lambda_1 \left(\HPhi(\E[\mt F|X_1])+\HPhi(\E[\mt F|X_1 X_2 Y_1]|X_1 Y_1)\right)
        +\lambda_2 \left(\HPhi(\E[\mt F|Y_1])+\HPhi(\E[\mt F|X_1 Y_1 Y_2]|X_1 Y_1)\right)\\
        &\geq \lambda_1 \left(\HPhi(\E[\mt F|X_1])+\HPhi(\E[\mt F|X_1 X_2]|X_1)\right)
        +\lambda_2 \left(\HPhi(\E[\mt F|Y_1])+\HPhi(\E[\mt F|Y_1Y_2]|Y_1)\right)\\
        &=\lambda_1 \HPhi(\E[\mt F|X_1X_2]) + \lambda_2 \HPhi(\E[\mt F|Y_1Y_2]).
    \end{align*}
    Thus \eqref{eq:tensor-direction} is proved.

    (ii) Suppose $(\lambda_1,\lambda_2)\in \fR_{\Phi}^{M,d}(X_1;Y_1)$. Let $\mt F=\mt F(X_2,Y_2)$. We then have
    \begin{align} \label{dpi-matrix-1}
        \HPhi(\E[\mt F|X_1Y_1])\geq \lambda_1 \HPhi(\E[\mt F|X_1]) +\lambda_2 \HPhi(\E[\mt F|Y_1]).
    \end{align} 
    By the chain rule, 
    \begin{align} \label{eq:dp-FX1Y1}
        \HPhi(\mt F|X_1Y_1)
        &= \HPhi(\E[\mt F|X_1X_2Y_1]|X_1Y_1)
        +\HPhi(\mt F|X_1X_2Y_1) \nonumber\\
        &\geq \HPhi(\E[\mt F|X_1X_2Y_1]|X_1Y_1)
        +\HPhi(\E[\mt F|X_1Y_1Y_2]|X_1Y_1) \nonumber\\
        &\geq \lambda_1 \HPhi(\E[\mt F|X_1X_2Y_1]|X_1Y_1)
        +\lambda_2\HPhi(\E[\mt F|X_1Y_1Y_2]|X_1Y_1), 
    \end{align}
    where the second line follows from applying Lemma~\ref{lem:matrix-Phi-ineq}~(ii) with the Markov chain $X_2\to X_1Y_1\to Y_2$, and the last line follows from $\lambda_1,\lambda_2\leq 1$ and the nonnegativity of the matrix $\Phi$-entropy.  
    Note that we have the Markov chains $X_2\rightarrow X_1\rightarrow Y_1$ and $X_1\rightarrow Y_1\rightarrow Y_2$. Applying Lemma \ref{lem:matrix-Phi-ineq} (ii) with $\E[\mt F|X_1 X_2 Y_1]$ and $Y_1\rightarrow X_1\rightarrow X_2$, we have
    \begin{align} \label{eq:dp-X1X2Y1}
        \HPhi(\E[\mt F|X_1 X_2 Y_1]|X_1 Y_1)\geq \HPhi(\E[\mt F|X_1 X_2]|X_1).
    \end{align}
    Similarly, applying Lemma \ref{lem:matrix-Phi-ineq} (ii) with $\E[\mt F|X_1 Y_1 Y_2]$ and $X_1\rightarrow Y_1\rightarrow Y_2$
    \begin{align} \label{eq:dp-X1Y1Y2}
        \HPhi(\E[\mt F|X_1 Y_1 Y_2]|X_1 Y_1)\geq \HPhi(\E[\mt F|Y_1Y_2]|Y_1),
    \end{align}
    Combining \eqref{eq:dp-FX1Y1}, \eqref{eq:dp-X1X2Y1} and \eqref{eq:dp-X1Y1Y2}, we obtain 
    \begin{align} \label{eq:F|X1X1}
        \HPhi(\mt F|X_1Y_1)
        \geq \lambda_1 \HPhi(\E[\mt F|X_1 X_2]|X_1)
        +\lambda_2 \HPhi(\E[\mt F|Y_1Y_2]|Y_1).
    \end{align}
    Then 
    \begin{align*}
        \HPhi(\mt F)  
        &\overset{\rm (a)}{=}\HPhi(\E[\mt F|X_1 Y_1])+\HPhi(\mt F|X_1 Y_1)\\
        &\overset{\rm (b)}{\geq} 
        \lambda_1 \left(\HPhi(\E[\mt F|X_1])+\HPhi(\E[\mt F|X_1 X_2]|X_1)\right)  +\lambda_2 \left(\HPhi(\E[\mt F|Y_1])+\HPhi(\E[\mt F|Y_1Y_2]|Y_1)\right) \\
        &\overset{\rm (c)}{=} \lambda_1 \HPhi(\E[\mt F|X_1X_2]) +\lambda_2 \HPhi(\E[\mt F|Y_1Y_2])\\
        &\overset{\rm (d)}{\geq} \lambda_1 \HPhi(\E[\mt F|X_2]) +\lambda_2 \HPhi(\E[\mt F|Y_2])
    \end{align*}
    where (a) follows from the chain rule, (b) follows from \eqref{dpi-matrix-1} and \eqref{eq:F|X1X1}, (c) follows from the chain rule, and (d) follows from the fact that conditioning reduces the matrix $\Phi$-entropy. 

    Thus, $(\lambda_1,\lambda_2)\in \fR_{\Phi}^{M,d}(X_2;Y_2)$, which finishes the proof of the data processing property for the matrix $\Phi$-ribbon. 
\end{proof}

\section{Proof of Theorem \ref{thm:tensorization-for-boxes}} \label{apx:proof-main-theorem}
\begin{table*}
\begin{center}
\begin{small}
\caption{\small (\cite{beigi2015monotone}) Alice and Bob interact with the  $n$ no-signaling boxes in sequences that may differ and could be chosen randomly. The table outlines the notations used to represent the random variables related to these sequences, as well as the inputs and outputs of the boxes. The term “Alice’s transcript” refers to the information Alice has gathered through her observations up to a given point in time.
}
\begin{tabular}{|c|c|c|}
\hline
Notation & Description   & { Corresponding variable of Bob}\\
\hline
$\Pi_i$ & Alice uses the $i$-th box in her $\Pi_i$-th action& $\Omega_i$\\ \hline   
 & Index of the box Alice uses in her $i$-th action: & \\
$\wPi_i$&$~\Pi_{\wPi_i}=i, \qquad \wPi_{\Pi_i}=i$&$\wOmega_i$\\
\hline
$X_i$ & Alice's input of the $i$-th box & $Y_i$\\\hline
$A_i$ & Alice's output of the $i$-th box & $B_i$\\\hline
& Alice's input in her $i$-th action:  & \\
$\wX_i$ &$~\wX_i=X_{\wPi_i}$&$\wY_i$\\
\hline
 & Alice's output in her $i$-th action: & \\
$\wA_i$&$~\wA_i=A_{\wPi_i}$&$\wB_i$\\
\hline
$T_i$ & Alice's transcript before using the $i$-th box & $S_i$\\\hline
 & Alice's transcript before her $i$-th action: & \\
$\wT_i$ & $\wT_i=T_{\wPi_i}$& $\wS_i$
\\
\hline
& Alice's information before getting the output of the $i$-th box: & \\
$\Tie$ & $\Tie=T_iX_i\Pi_i$ & $\Sie$
\\
\hline
 & Alice's information before getting the output of the $i$-th action: & \\
$\wTie$ & $\wTie\defeq \wT_i\wX_i\wPi_i$ & $\wSie$ 
\\
\hline
\end{tabular}
\end{small}
\end{center}

\label{table:summary}
\end{table*}

We will introduce some notations to formulate the wiring. Let $(\Pi_1,...,\Pi_n)$ denote the random variables for the order of boxes used by Alice. For example, in the $\Pi_i$-th action, Alice uses the $i$-th box. Next, define the inverse permutation of $(\Pi_1,...,\Pi_n)$ as $(\wPi_1,...,\wPi_n)$, i.e., 
\begin{align*}
    \wPi_{\Pi_i}=i, \qquad \qquad \Pi_{\wPi_i}=i
\end{align*}
which means that in the $j$-th action, Alice uses the $\wPi_j$-th box. The corresponding input is then $X_{\wPi_j}$, and the output is $A_{\wPi_j}$. For simplicity, define
\begin{align*}
    \wX_j \defeq X_{\wPi_j}, \qquad \qquad \wA_j \defeq A_{\wPi_j}.
\end{align*}

Similarly, define $(\Omega_1,...,\Omega_n)$ to be the random variables for the order of boxes used by Bob, and define the inverse permutation of $(\Omega_1,...,\Omega_n)$ as $(\wOmega_1,...,\wOmega_n)$, i.e.,  
\begin{align*}
    \wOmega_{\Omega_i}=i, \qquad \qquad \Omega_{\wOmega_i}=i.
\end{align*}
We use 
\begin{align*}
    \wY_i\defeq Y_{\wOmega_i}, \qquad \qquad \wB_i \defeq B_{\wOmega_i}.
\end{align*}

Denote $T_i$ as the \textit{transcript} of the information that Alice has before using the $i$-th box (before the $\Pi_i$-th action), i.e., 
\begin{align*}
    T_i &\defeq \wPi_1\dots \wPi_{\Pi_{i}-1}X_{\wPi_1}\dots X_{\wPi_{\Pi_{i}-1}}A_{\wPi_1}\dots A_{\wPi_{\Pi_{i}-1}}\\
    &= \wPi_{[\Pi_{i}-1]} \wX_{[\Pi_{i}-1]} \wA_{[\Pi_{i}-1]}.
\end{align*}
Denote $\wT_i$ as the transcript of Alice before her $i$-th action, i.e., 
\begin{align*}
    \wT_i &\defeq T_{\wPi_i} = \wPi_1\dots \wPi_{i-1}X_{\wPi_1} \dots X_{\wPi_{i-1}} A_{\wPi_1} \dots A_{\wPi_{i-1}}\\
    &= \wPi_{[i-1]} \wX_{[i-1]} \wA_{[i-1]}.
\end{align*}
Similarly, define $S_i$ and $\wS_i$ as transcripts for Bob, i.e., 
\begin{align*}
S_i \defeq \wOmega_{[\Omega_{i}-1]} \wY_{[\Omega_{i}-1]} \wB_{[\Omega_{i}-1]}, \qquad \wS_i \defeq \wOmega_{[i-1]} \wY_{[i-1]} \wB_{[i-1]}.
\end{align*}
{Moreover, define $\Tie\defeq T_iX_i\Pi_i$, which can be interpreted as all the information at hand for Alice before she gets the output for the $i$-th box. Then define $\wTie\defeq \wT_i\wX_i\wPi_i$. Similarly, for Bob, define $\Sie\defeq S_i Y_i\Omega_i$ and $\wSie\defeq \wS_i\wY_i\wOmega_i$. }  

Hence, we can formulate the wiring in this way. In the $i$-th action, Alice has the information $\wT_i$ at hand and Bob has the information $\wS_i$. Then Alice chooses the $\wPi_i$-th box and the corresponding input $\wX_i$ based on $\wT_i$, and Bob chooses the $\wOmega_i$-th box and the corresponding input $\wY_i$ based on $\wS_i$ using only their respective local transcripts and private randomness. The $\wPi_i$-th box generates the output $\wA_i$ for Alice, and the $\wOmega_i$-th box generates the output $\wB_i$ for Bob. 

Based on the wiring and the no-signaling property, we can state the following lemma. 
\begin{lemma}[Lemma 2 in \cite{beigi2015monotone}] \label{lem:mc} 
    For the wiring of no-signaling boxes, given $x'y'$, we have the following Markov chains: 
    \begin{enumerate} [label=\rm{(\roman*)}]
        \item $A_iB_i\rightarrow X_iY_i\rightarrow T_iS_i\Pi_i\Omega_i$.
        \item $B_i \rightarrow S_i^e \rightarrow T_i^e$.
        \item $B_i \rightarrow A_i T_i^e S_i^e \rightarrow \An \Xn \Pin$.
        \item $\wY_i\wOmega_i \rightarrow \wS_i \rightarrow \An\Xn\Pin$.
    \end{enumerate}
\end{lemma}

Combining Lemmas \ref{lem:matrix-Phi-ineq} and \ref{lem:mc}, we immediately get the following corollary, which is important in the proof of Theorem \ref{thm:tensorization-for-boxes}. 
\begin{corollary} \label{cor:mc}
    Suppose $\Phi\in \mathscr{F}_M^d$. For any $\mt F:\cX\to \PSD^d$, with Lemma \ref{lem:matrix-Phi-ineq} (ii), conditioned on $x',y'$, we have the following inequalities: 
    \begin{enumerate} [label=\rm{(\roman*)}]
        \item from Lemma \ref{lem:mc} (ii), 
        \begin{align*}
            \HPhi(\E[\mt F|B_i\Tie\Sie]|\Tie\Sie)\geq
                \HPhi(\E[\mt F|B_i\Sie]|\Sie).
        \end{align*}
        \item from Lemma \ref{lem:mc} (iii), 
        \begin{align*}
            \HPhi(\E[\mt F|\An\Xn\Pin B_i\Sie]|\An\Xn\Pin \Sie)
            \geq 
            \HPhi(\E[\mt F|A_iB_i\Tie\Sie]|A_i\Tie\Sie).
        \end{align*}
        \item from Lemma \ref{lem:mc} (iv), 
        \begin{align*}
            \HPhi(\E[\mt F|\An\Xn\Pin\wSie]|\An\Xn\Pin\wS_i)
            \geq
            \HPhi(\E[\mt F|\wS_i^e]|\wS_i). 
        \end{align*}
    \end{enumerate}
\end{corollary}

We will also use the following lemma to prove the tensorization for the matrix $\Phi$-ribbon of no-signaling boxes. 
\begin{lemma} \label{lem:expand}
   For any $\PSD^d$-valued $\mt F$, 
   \begin{enumerate}[\rm (i)]
       \item \begin{align*}
            \HPhi(\E[\mt F|\An\Xn\Pin]) 
            =& \sum_{i=1}^n \left(
            \HPhi(\E[\mt F|\wT_i^e]|\wT_i)
            +\HPhi(\E[\mt F|A_i \Tie]|\Tie)
            \right),
        \end{align*}
        and similarly, 
        \begin{align*}
            \HPhi(\E[\mt F|\Bn\Yn\Omegan])
            =& \sum_{i=1}^n
            \left(\HPhi(\E[\mt F|\wS_i^e]|\wS_i)
            +\HPhi(\E[\mt F|B_i \Sie]|\Sie)\right).
        \end{align*}
       \item \begin{align*}
           &\HPhi(\E[\mt F|\An\Xn\Pin\Bn\Yn\Omegan]|\An\Xn\Pin) \\
           =& \sum_{i=1}^n \bigg( \HPhi(\E[\mt F|\An\Xn\Pin\wSie]|\An\Xn\Pin\wS_i)
           +\HPhi(\E[\mt F|\An\Xn\Pin B_i \Sie]|\An\Xn\Pin S_i^e) \bigg).
       \end{align*}
   \end{enumerate}
\end{lemma}
We postpone the proof of Lemma \ref{lem:expand} to the end of this appendix. 

\begin{proof}[Proof of Theorem \ref{thm:tensorization-for-boxes}]
    By the data processing property, if the box $p_{A'B'|X'Y'}$ can be generated by the box 
    $$p_{\An\Xn\Pin\Bn\Yn\Omegan|X'Y'},$$ 
    then for any $x',y'$, 
    \begin{align*}
        \fR_{\Phi}^{M,d}(\An\Xn\Pin;\Bn\Yn\Omegan|x';y') \subseteq \fR_{\Phi}^{M,d}(A';B'|x',y').
    \end{align*}

    Thus, it is sufficient to prove
    \begin{equation*}
        \bigcap_{i=1}^n \fR_{\Phi}^{M,d} (A_i;B_i|X_i;Y_i)\subseteq \fR_{\Phi}^{M,d} (A_{[n]}X_{[n]}\Pi_{[n]};B_{[n]}Y_{[n]}\Omega_{[n]}|x',y')
    \end{equation*}
    for any $x',y'$.  That is to say, for any $(\lambda_1,\lambda_2)\in \bigcap_{i=1}^n \fR_{\Phi}^{M,d} (A_i;B_i|X_i,Y_i)$, we want to show 
    $$\lambda_1 \HPhi(\E[\mt F|A_{[n]}X_{[n]}\Pi_{[n]}]) + \lambda_2 \HPhi(\E[\mt F|B_{[n]}Y_{[n]}\Omega_{[n]}])\leq \HPhi(\mt F)$$
    for any PSD matrix-valued function $\mt F(\An,\Xn,\Pin,\Bn,\Yn,\Omegan)$. If  $\lambda_1+\lambda_2\leq 1$, the inequality is trivially satisfied because of the data processing inequality for the matrix $\Phi$-entropy. Thus, we only need to show the inequality for pairs $(\lambda_1,\lambda_2)\in \bigcap_{i=1}^n \fR_{\Phi}^{M,d}(A_i;B_i|X_i,Y_i)$ such that $\lambda_1+\lambda_2\geq 1$. 

    Define 
    \begin{align*}
        \chi(\lambda_1,\lambda_2) &\triangleq -\lambda_1 \HPhi(\E[\mt F|\An\Xn\Pin]) - \lambda_2 \HPhi(\E[\mt F|\Bn\Yn\Omegan]) + \HPhi(\mt F).
    \end{align*}
    We need to show that $\chi(\lambda_1,\lambda_2)\geq 0$. 
    By Lemma \ref{lem:expand}, we have 
    \begin{align*}
        \HPhi(\E[\mt F|\An\Xn\Pin])
        &= \sum_{i=1}^n \left(\HPhi(\E[\mt F|\wT_i^e]|\wT_i)
        +\HPhi(\E[\mt F|A_i \Tie]|\Tie) \right) 
        \\
        &=\sum_{i=1}^{n}
        \bigg(\HPhi(\E[\mt F|\wA_{[i-1]}\wX_{[i]}\wPi_{[i]}]|\wA_{[i-1]}\wX_{[i-1]}\wPi_{[i-1]})
        +\HPhi(\E[\mt F|A_{i}\Tie S_i^e]|T_i^eS_i^e)\nonumber \\
        &\qquad +\HPhi(\E[\mt F|A_i \Tie]|\Tie)
        -\HPhi(\E[\mt F|A_{i}\Tie S_i^e]|T_i^eS_i^e)\bigg)
    \end{align*}
    and
    \begin{align*}
        \HPhi(\E[\mt F|\Bn\Yn\Omegan]) 
        &=\sum_{i=1}^n \bigg(\HPhi(\E[\mt F|\wS_i^e]|\wS_i)
        +\HPhi(\E[\mt F|B_i\Sie]|\Sie) \bigg)\\
        &=\sum_{i=1}^{n} \bigg( \HPhi(\E[\mt F|\wB_{[i-1]}\wY_{[i]}\wOmega_{[i]}]|\wB_{[i-1]}\wY_{[i-1]}\wOmega_{[i-1]})
        +\HPhi(\E[\mt F|B_{i}\Tie S_i^e]|T_i^eS_i^e)\nonumber\\
        &\quad +\HPhi(\E[\mt F|B_i \Sie]|\Sie)
        -\HPhi(\E[\mt F|B_{i}\Tie S_i^e]|T_i^eS_i^e) \bigg). 
    \end{align*}

    Using the fact that $\mt F$ is a function of $(\An,\Xn,\Pin,\Bn,\Yn,\Omegan)$, we can rewrite the above expression as follows:
    \begin{align}
        \chi(\lambda_1,\lambda_2) 
        &= -\sum_{i=1}^{n} \bigg( \lambda_1 \HPhi(\E[\mt F|\wA_{[i-1]}\wX_{[i]}\wPi_{[i]}]|\wA_{[i-1]}\wX_{[i-1]}\wPi_{[i-1]})
        +\lambda_1 \HPhi(\E[\mt F|A_{i}\Tie S_i^e]|T_i^eS_i^e)\nonumber \\
        &\quad +\lambda_1 \HPhi(\E[\mt F|A_i \Tie]|\Tie)
        -\lambda_1 \HPhi(\E[\mt F|A_{i}\Tie S_i^e]|T_i^eS_i^e)\nonumber \\
        &\quad 
        +\lambda_2 \HPhi(\E[\mt F|\wB_{[i-1]}\wY_{[i]}\wOmega_{[i]}]|\wB_{[i-1]}\wY_{[i-1]}\wOmega_{[i-1]})
        +\lambda_2 \HPhi(\E[\mt F|B_{i}\Tie S_i^e]|T_i^eS_i^e)\nonumber\\
        &\quad +\lambda_2 \HPhi(\E[\mt F|B_i \Sie]|\Sie)
        -\lambda_2\HPhi(\E[\mt F|B_{i}\Tie S_i^e]|T_i^eS_i^e)
        \bigg) + \HPhi(\E[\mt F|\An\Xn\Pin\Bn\Yn\Omegan])\label{eqnchif2}.
    \end{align} 
    Next, we claim that
    \begin{equation}
    \lambda_1\HPhi(\E[\mt F|A_i\Tie\Sie]|\Tie\Sie)
        +\lambda_2\HPhi(\E[\mt F|B_i\Tie\Sie]|\Tie\Sie)\leq \HPhi(\E[\mt F|A_iB_i\Tie\Sie]|\Tie\Sie).\label{eqnStep}
    \end{equation}
    To show this inequality, observe that from Lemma \ref{lem:mc} (i), we have the Markov chain $A_iB_i\rightarrow X_iY_i \rightarrow T_iS_i\Pi_i\Omega_i$. Fix $t_i,s_i,\pi_i,\omega_i$, then we have $p_{A_iB_i|X_iY_it_is_i\pi_i\omega_i}=p_{A_iB_i|X_iY_i}$. Thus, $p(A_iB_i|X_iY_it_is_i\pi_i\omega_i)$ is the same box as $p_{A_iB_i|X_iY_i}$. 

    As $(\lambda_1,\lambda_2)\in \bigcap_{i=1}^n \fR_{\Phi}^{M,d} (A_i;B_i|X_i,Y_i)$, we have
    \begin{align*}
        &\lambda_1\HPhi(\E[\mt F|A_iX_iY_i t_i s_i \pi_i \omega_i]|X_iY_i t_i s_i \pi_i \omega_i)
        +\lambda_2\HPhi(\E[\mt F|B_iX_iY_i t_i s_i \pi_i \omega_i]|X_iY_i t_i s_i \pi_i \omega_i)\\
        &\quad \leq\HPhi(\E[\mt F|A_iB_iX_iY_i t_i s_i \pi_i \omega_i]|X_iY_i t_i s_i \pi_i \omega_i).
    \end{align*}
    By averaging over $t_i,s_i,\pi_i,\omega_i$, we obtain \eqref{eqnStep}.

    Next, observe that \eqref{eqnStep} and \eqref{eqnchif2} imply that
    $$\chi(\lambda_1,\lambda_2)\geq \chi'(\lambda_1,\lambda_2)$$
    where
    \begin{align*}
    \chi'(\lambda_1,\lambda_2)
        &= -\sum_{i=1}^{n} \bigg[ \lambda_1 \HPhi(\E[\mt F|\wA_{[i-1]}\wX_{[i]}\wPi_{[i]}]|\wA_{[i-1]}\wX_{[i-1]}\wPi_{[i-1]})
        \nonumber \\
        &\quad +\lambda_1 \HPhi(\E[\mt F|A_i \Tie]|\Tie)
        -\lambda_1 \HPhi(\E[\mt F|A_{i}\Tie S_i^e]|T_i^eS_i^e)\nonumber \\
        &\quad 
        +\lambda_2 \HPhi(\E[\mt F|\wB_{[i-1]}\wY_{[i]}\wOmega_{[i]}]|\wB_{[i-1]}\wY_{[i-1]}\wOmega_{[i-1]})
        \nonumber\\
        &\quad +\lambda_2 \HPhi(\E[\mt F|B_i \Sie]|\Sie)
        -\lambda_2\HPhi(\E[\mt F|B_{i}\Tie S_i^e]|T_i^eS_i^e) + \HPhi(\E[\mt F|A_iB_i\Tie\Sie]|\Tie\Sie)
        \bigg]\nonumber\\
        &\quad + \HPhi(\E[\mt F|\An\Xn\Pin\Bn\Yn\Omegan]).
    \end{align*}
    To show that $\chi(\lambda_1,\lambda_2)\geq 0$, it suffices to establish that
    $\chi'(\lambda_1,\lambda_2)\geq 0$ for any $\lambda_1, \lambda_2\in[0,1]$ satisfying $\lambda_1+\lambda_2\geq 1$. Since $\chi'(\lambda_1,\lambda_2)$ is affine in $\lambda_1$ and $\lambda_2$, it suffices to show this for the corner points of the set, i.e., for the three points $(1,0),(0,1)$ and $(1,1)$. By symmetry, we show this when 
    \begin{align*}
        \lambda_1=1,\qquad \lambda_2\in\{0,1\}.
    \end{align*}
    The proof for the other corner point is similar. 
    We have
   \begin{align}
    \chi'(1,\lambda_2)
        =& -\sum_{i=1}^{n} \bigg[ \HPhi(\E[\mt F|\wA_{[i-1]}\wX_{[i]}\wPi_{[i]}]|\wA_{[i-1]}\wX_{[i-1]}\wPi_{[i-1]})
        \nonumber \\
        &\quad + \HPhi(\E[\mt F|A_i \Tie]|\Tie)
        - \HPhi(\E[\mt F|A_{i}\Tie S_i^e]|T_i^eS_i^e)\nonumber \\
        &\quad 
        +\lambda_2 \HPhi(\E[\mt F|\wB_{[i-1]}\wY_{[i]}\wOmega_{[i]}]|\wB_{[i-1]}\wY_{[i-1]}\wOmega_{[i-1]})
        \nonumber\\
        &\quad +\lambda_2 \HPhi(\E[\mt F|B_i \Sie]|\Sie)
        -\lambda_2\HPhi(\E[\mt F|B_{i}\Tie S_i^e]|T_i^eS_i^e) + \HPhi(\E[\mt F|A_iB_i\Tie\Sie]|\Tie\Sie)
        \bigg]\nonumber\\
        & + \HPhi(\E[\mt F|\An\Xn\Pin\Bn\Yn\Omegan]).\label{eqnn14}
    \end{align}
    Next, consider the following expansion by the chain rule, 
    \begin{align*}
        \HPhi(\E[\mt F|\An\Xn\Pin\Bn\Yn\Omegan])=\HPhi(\E[\mt F|\An\Xn\Pin]) + \HPhi(\E[\mt F|\An\Xn\Pin\Bn\Yn\Omegan]|\An\Xn\Pin)
    \end{align*}
    The first term on the right hand side can be expanded using part (i) of Lemma \ref{lem:expand} as follows:
    \begin{align*}
        \HPhi(\E[\mt F|\An\Xn\Pin]) &=\sum_{i=1}^{n} \bigg( \HPhi(\E[\mt F|\wA_{[i-1]}\wX_{[i]}\wPi_{[i]}]|\wA_{[i-1]}\wX_{[i-1]}\wPi_{[i-1]})
            + \HPhi(\E[\mt F|A_i \Tie]|\Tie) \bigg).
    \end{align*}
    Therefore, the above two equations imply
    \begin{align*}
    \HPhi(\E[\mt F|\An\Xn\Pin\Bn\Yn\Omegan]|\An\Xn\Pin)
    &=
    \HPhi(\E[\mt F|\An\Xn\Pin\Bn\Yn\Omegan])\\&\quad-\sum_{i=1}^{n}  \HPhi(\E[\mt F|\wA_{[i-1]}\wX_{[i]}\wPi_{[i]}]|\wA_{[i-1]}\wX_{[i-1]}\wPi_{[i-1]})
            \\&\quad-\sum_{i=1}^{n} \HPhi(\E[\mt F|A_i \Tie]|\Tie).
    \end{align*}

    Therefore, we can rewrite \eqref{eqnn14} as
    \begin{align*}
    \chi'(1,\lambda_2)
        &= -\sum_{i=1}^{n} \bigg[-\HPhi(\E[\mt F|A_{i}\Tie S_i^e]|T_i^eS_i^e) +\lambda_2 \HPhi(\E[\mt F|\wB_{[i-1]}\wY_{[i]}\wOmega_{[i]}]|\wB_{[i-1]}\wY_{[i-1]}\wOmega_{[i-1]})
        \nonumber\\
        &\quad +\lambda_2 \HPhi(\E[\mt F|B_i \Sie]|\Sie)
        -\lambda_2\HPhi(\E[\mt F|B_{i}\Tie S_i^e]|T_i^eS_i^e) + \HPhi(\E[\mt F|A_iB_i\Tie\Sie]|\Tie\Sie)
        \bigg]\nonumber\\
        &\quad + \HPhi(\E[\mt F|\An\Xn\Pin\Bn\Yn\Omegan]|\An\Xn\Pin).
    \end{align*}
Next, note that by the chain rule \eqref{eqn:matrix-CR2}, we have
\begin{align*}
\HPhi(\E[\mt F|A_iB_i\Tie\Sie]|\Tie\Sie) &= \HPhi(\E[\mt F|A_{i}\Tie S_i^e]|T_i^eS_i^e) + \HPhi(\E[\mt F|A_iB_i\Tie\Sie]|A_i \Tie\Sie)
\end{align*}
    and therefore we obtain
    \begin{align*}
    \chi'(1,\lambda_2)
        &=  -\sum_{i=1}^{n} \bigg[\lambda_2 \HPhi(\E[\mt F|\wB_{[i-1]}\wY_{[i]}\wOmega_{[i]}]|\wB_{[i-1]}\wY_{[i-1]}\wOmega_{[i-1]})
        \nonumber\\
        &\quad +\lambda_2 \HPhi(\E[\mt F|B_i \Sie]|\Sie)
        -\lambda_2\HPhi(\E[\mt F|B_{i}\Tie S_i^e]|T_i^eS_i^e) + \HPhi(\E[\mt F|A_iB_i\Tie\Sie]|A_i \Tie\Sie)
        \bigg]\nonumber\\
        &\quad + \HPhi(\E[\mt F|\An\Xn\Pin\Bn\Yn\Omegan]|\An\Xn\Pin).
    \end{align*}

    By Lemma \ref{lem:expand} (ii), we have 
    \begin{align*}
       \HPhi(\E[\mt F|\An\Xn\Pin\Bn\Yn\Omegan]|\An\Xn\Pin) = \sum_{i=1}^n &
       \bigg(\HPhi(\E[\mt F|\An\Xn\Pin\wB_{[i-1]}\wY_{[i]}\wOmega_{[i]}]|\An\Xn\Pin\wS_i)\\
       &+\HPhi(\E[\mt F|\An\Xn\Pin B_i\Sie]|\An\Xn\Pin S_i^e)\bigg).
    \end{align*}
    Then 
    \begin{align*}
        \chi'(1,\lambda_2) &= \sum_{i=1}^{n} \bigg[-\lambda_2 \HPhi(\E[\mt F|\wB_{[i-1]}\wY_{[i]}\wOmega_{[i]}]|\wB_{[i-1]}\wY_{[i-1]}\wOmega_{[i-1]})
        \nonumber\\
        &\qquad -\lambda_2 \HPhi(\E[\mt F|B_i \Sie]|\Sie)
        +\lambda_2\HPhi(\E[\mt F|B_{i}\Tie S_i^e]|T_i^eS_i^e) -\HPhi(\E[\mt F|A_iB_i\Tie\Sie]|A_i \Tie\Sie)
        \nonumber\\
        &\qquad + \HPhi(\E[\mt F|\An\Xn\Pin\wB_{[i-1]}\wY_{[i]}\wOmega_{[i]}]|\An\Xn\Pin\wS_i)\\
        &\qquad+\HPhi(\E[\mt F|\An\Xn\Pin B_i\Sie]|\An\Xn\Pin S_i^e)\bigg].
    \end{align*}
    
    The three parts of Corollary \ref{cor:mc} state that
    \begin{align}
        \HPhi(\E[\mt F|B_i\Tie\Sie]|\Tie\Sie)&\geq\HPhi(\E[\mt F|B_i\Sie]|\Sie),\\
        \HPhi(\E[\mt F|\An\Xn\Pin B_i\Sie]|\An\Xn\Pin S_i^e)
        &\geq 
        \HPhi(\E[\mt F|A_iB_i\Tie\Sie]|A_i\Tie\Sie), \label{eq:wirings-ineq-2}\\
        \HPhi(\E[\mt F|\An\Xn\Pin\wB_{[i-1]}\wYi\wOmegai]|\An\Xn\Pin\wS_i)
        &\geq
        \HPhi(\E[\mt F|\wB_{[i-1]}\wY_{[i]}\wOmega_{[i]}]|\wB_{[i-1]}\wY_{[i-1]}\wOmega_{[i-1]}).
    \end{align}
For $\lambda_2=0$, \eqref{eq:wirings-ineq-2} implies that $\chi'(1,\lambda_2)\geq 0$. For $\lambda_2=1$, we can use all of the above three inequalities to see that $\chi'(1,\lambda_2)\geq 0$. This establishes the desired inequality. The proof is complete.

\end{proof}

\begin{proof}[Proof of Lemma \ref{lem:expand}]
    (i) We have 
    \begin{align*}
        \HPhi(\E[\mt F|\An\Xn\Pin])
        &\stackrel{\rm (a)}{=} \HPhi(\E[\mt F|\wAn\wXn\wPin])\\
        &\stackrel{\rm (b)}{=} \sum_{i=1}^n \HPhi(\E[\mt F|\wA_{[i]} \wX_{[i]} \wPi_{[i]}]| \wA_{[i-1]} \wX_{[i-1]} \wPi_{[i-1]}) \\
        &\stackrel{\rm (c)}{=} \sum_{i=1}^{n} 
        \left(
            \HPhi(\E[\mt F|\wA_{[i-1]}\wX_{[i]}\wPi_{[i]}]|\wA_{[i-1]}\wX_{[i-1]}\wPi_{[i-1]}) +\HPhi(\E[\mt F|\wA_{[i]} \wX_{[i]} \wPi_{[i]}]|\wA_{[i-1]}\wX_{[i]}\wPi_{[i]})
        \right) \\
        &\stackrel{\rm (d)}{=} \sum_{i=1}^{n} 
        \left(
            \HPhi(\E[\mt F|\wA_{[i-1]}\wX_{[i]}\wPi_{[i]}]|\wA_{[i-1]}\wX_{[i-1]}\wPi_{[i-1]})
        +\HPhi(\E[\mt F|A_i \Tie]|\Tie)
        \right)
    \end{align*}
    where (a) follows from taking a permutation, (b) follows from the chain rule in \eqref{eqn:matrix-CRn}, (c) follows from the chain rule in \eqref{matrix-Phi:chain-rule}, and (d) follows from the following equation: 
    \begin{align*}
        &\sum_{i=1}^n \HPhi(\E[\mt F|\wA_{[i]} \wX_{[i]} \wPi_{[i]}]|\wA_{[i-1]}\wX_{[i]}\wPi_{[i]}) \\
        =& \sum_{i=1}^n \sum_{j=1}^n \HPhi\left(\E[\mt F|\wA_{[i]} \wX_{[i]} \wPi_{[i]}]| \wT_i \wX_i \wPi_i=j\right) p\left(\wPi_i=j\right)\\
        =& \sum_{i=1}^n \sum_{j=1}^n \HPhi\left(\E[\mt F|A_j T_j X_j \Pi_j]| T_j X_j \Pi_j=i\right) p\left(\wPi_i=j\right) \\
        =& \sum_{i=1}^n \sum_{j=1}^n \HPhi\left(\E[\mt F|A_j T_j X_j \Pi_j]| T_j X_j \Pi_j=i\right) p\left(\Pi_j=i\right) \\
        =& \sum_{j=1}^n \HPhi\left(\E\left[\mt F|A_j T_j^e\right]| T_j^e\right).
    \end{align*}

    The proof for $\HPhi(\E[\mt F|\Bn\Yn\Omegan])$ is similar. 

    (ii) Similarly, we have 
    \begin{align*}
        &\HPhi(\E[\mt F|\An\Xn\Pin\Bn\Yn\Omegan]|\An\Xn\Pin)  \\
        \overset{\rm (a)}{=}&\HPhi(\E[\mt F|\wAn\wXn\wPin\wBn\wYn\wOmegan]|\wAn\wXn\wPin) \\
        \overset{\rm (b)}{=}& \sum_{i=1}^n \HPhi(\E[\mt F|\wAn\wXn\wPin\wB_{[i]}\wY_{[i]}\wOmega_{[i]}]|\wAn\wXn\wPin\wB_{[i-1]}\wY_{[i-1]}\wOmega_{[i-1]}) \\
        \overset{\rm (c)}{=}& \sum_{i=1}^n \bigg( \HPhi(\E[\mt F|\wAn\wXn\wPin\wB_{[i-1]}\wY_{[i]}\wOmega_{[i]}]|\wAn\wXn\wPin\wB_{[i-1]}\wY_{[i-1]}\wOmega_{[i-1]}) \\
        & \quad+ \HPhi(\E[\mt F|\wAn\wXn\wPin\wB_{[i]}\wY_{[i]}\wOmega_{[i]}]|\wAn\wXn\wPin\wB_{[i-1]}\wY_{[i]}\wOmega_{[i]})\bigg) \\
        \overset{\rm (d)}{=}& \sum_{i=1}^n  \HPhi(\E[\mt F|\wAn\wXn\wPin\wB_{[i-1]}\wY_{[i]}\wOmega_{[i]}]|\wAn\wXn\wPin\wB_{[i-1]}\wY_{[i-1]}\wOmega_{[i-1]}) \\
        & \quad+ \sum_{j=1}^n \HPhi(\E[\mt F|\wAn\wXn\wPin B_j S_j Y_j \Omega_j]|\wAn\wXn\wPin S_j Y_j \Omega_j)
    \end{align*}
    where (a) follows from taking a permutation, 
    (b) follows from the chain rule in \eqref{eqn:matrix-CRn}, (c) follows from the chain rule in \eqref{matrix-Phi:chain-rule}, and (d) follows from the following equation: 
    \begin{align*}
        &\sum_{i=1}^n \HPhi(\E[\mt F|\wAn\wXn\wPin\wB_{[i]}\wY_{[i]}\wOmega_{[i]}]|\wAn\wXn\wPin\wB_{[i-1]}\wY_{[i]}\wOmega_{[i]}) \\
        =& \sum_{i=1}^n \sum_{j=1}^n \HPhi(\E[\mt F|\wAn\wXn\wPin\wB_{[i]}\wY_{[i]}\wOmega_{[i]}]|\wAn\wXn\wPin\wS_i\wY_i\wOmega_i=j)p(\wOmega_i=j)\\
        =& \sum_{i=1}^n \sum_{j=1}^n \HPhi(\E[\mt F|\wAn\wXn\wPin B_j S_j Y_j \Omega_j]|\wAn\wXn\wPin S_j Y_j \Omega_j=i) p(\wOmega_i=j)\\
        =& \sum_{j=1}^n \HPhi(\E[\mt F|\wAn\wXn\wPin B_j S_j Y_j \Omega_j]|\wAn\wXn\wPin S_j Y_j \Omega_j). 
    \end{align*}
\end{proof}

\section{Proof of Theorem \ref{thm:matrix-Phi-SDPI-data-processing-tensorization}}
\label{app:proof-matrix-Phi-SDPI-data-processing-tensorization}

We first prove the data-processing property. Suppose that $X\to Y\to Z$
forms a Markov chain, and fix $\mt F:\cX\to\PSD^d$.
Set $\mt G
\defeq
\E[\mt F| Y]$.
The Markov property implies
\begin{align*}
\E[\mt F| Z]
&=
\E\left[
\E[\mt F| Y]
| Z
\right]
\\
&=
\E[\mt G| Z].
\end{align*}
Therefore,
\begin{align*}
\HPhi\left(
\E[\mt F| Z]
\right)
&=
\HPhi\left(
\E[\mt G| Z]
\right)
\\
&\leq
\eta_{\Phi}^{M,d}(Y;Z)
\HPhi(\mt G)
\\
&=
\eta_{\Phi}^{M,d}(Y;Z)
\HPhi\left(
\E[\mt F| Y]
\right)
\\
&\leq
\eta_{\Phi}^{M,d}(Y;Z)
\eta_{\Phi}^{M,d}(X;Y)
\HPhi(\mt F).
\end{align*}
Taking the supremum over all \(\mt F:\cX\to\PSD^d\) satisfying \(\HPhi(\mt F)>0\) proves
\eqref{eq:matrix-Phi-SDPI-data-processing}.

Then
\eqref{eq:matrix-Phi-SDPI-data-processing-2} follows from $0
\leq
\eta_{\Phi}^{M,d}(X;Y),\eta_{\Phi}^{M,d}(Y;Z)
\leq
1$.

We next prove tensorization. It suffices to consider two independent pairs, since the general statement follows by induction.

Let $(X_1,Y_1)$ and $(X_2,Y_2)$
be independent. Define $\eta_i
\defeq
\eta_{\Phi}^{M,d}(X_i;Y_i)$ for $i\in\{1,2\}$ and $\eta
\defeq
\max\{\eta_1,\eta_2\}$.

We first prove that
\begin{align}
\eta_{\Phi}^{M,d}(X_1X_2;Y_1Y_2)
\leq
\eta.
\label{eq:matrix-Phi-SDPI-tensor-upper}
\end{align}
Fix $\mt F:
\cX_1\times\cX_2
\to
\PSD^d$
and set $\mt G
\defeq
\E[\mt F| Y_1,Y_2]$.
By the chain rule,
\begin{align}
\HPhi(\mt G)
=
\HPhi(\mt G| Y_2)
+
\HPhi\left(
\E[\mt G| Y_2]
\right).
\label{eq:matrix-Phi-SDPI-output-chain}
\end{align}
First we consider the first term in \eqref{eq:matrix-Phi-SDPI-output-chain}.
For each \(y_2\in\cY_2\), define
\begin{align*}
\mt F_{y_2}(x_1)
\defeq
\E\left[
\mt F(x_1,X_2)
|
Y_2=y_2
\right].
\end{align*}
Since the pairs \((X_1,Y_1)\) and \((X_2,Y_2)\) are independent,
\begin{align*}
\mt G(y_1,y_2)
=
\E\left[
\mt F_{y_2}(X_1)
|
Y_1=y_1
\right].
\end{align*}
Applying the SDPI to $\mt F_{y_2}(X_1)$ for every $y_2$ and taking the average, we have
\begin{align}
\HPhi(\mt G| Y_2)
\leq
\eta_1
\E_{Y_2}
\left[
\HPhi\left(
\mt F_{Y_2}(X_1)
\right)
\right],
\label{eq:matrix-Phi-SDPI-first-coordinate}
\end{align}
where the matrix \(\Phi\)-entropy inside the expectation is taken with respect to \(X_1\).

For every \(y_2\in\cY_2\),
\begin{align*}
\mt F_{y_2}(x_1)
=
\sum_{x_2\in\cX_2}
p_{X_2| Y_2}(x_2| y_2)
\mt F(x_1,x_2).
\end{align*}
Since \(X_1\) is independent of \((X_2,Y_2)\), the convexity of the matrix \(\Phi\)-entropy functional gives
\begin{align*}
\HPhi\left(
\mt F_{y_2}(X_1)
\right)
\leq
\sum_{x_2\in\cX_2}
p_{X_2| Y_2}(x_2| y_2)
\HPhi\left(
\mt F(X_1,x_2)
\right)
\end{align*}
where the inequality follows from the convexity of $(\mt A_x)_{x\in \cX} \to \HPhi(\mt A)$
for $\Phi\in \mathscr{F}_M^d$.
Averaging over \(Y_2\), we obtain
\begin{align}
\E_{Y_2}
\left[
\HPhi\left(
\mt F_{Y_2}(X_1)
\right)
\right]
&\leq
\E_{X_2}
\left[
\HPhi\left(
\mt F(X_1,X_2)
\right)
\right]
\nonumber\\
&=
\HPhi(\mt F| X_2).
\label{eq:matrix-Phi-SDPI-convexity-average}
\end{align}
Combining \eqref{eq:matrix-Phi-SDPI-first-coordinate} and \eqref{eq:matrix-Phi-SDPI-convexity-average}, it follows that
\begin{align}
\HPhi(\mt G| Y_2)
\leq
\eta_1
\HPhi(\mt F| X_2).
\label{eq:matrix-Phi-SDPI-first-part-final}
\end{align}

For the second term in \eqref{eq:matrix-Phi-SDPI-output-chain}, we have
\begin{align*}
\E[\mt G| Y_2]
&=
\E[\mt F| Y_2]
\nonumber\\
&=
\E\left[
\E[\mt F| X_2]
|
Y_2
\right]
\end{align*}
since $(X_1,Y_1)$ and $(X_2,Y_2)$
are independent.
Since \(\E[\mt F| X_2]\) is a function of \(X_2\), we obtain
\begin{align}
\HPhi\left(
\E[\mt G| Y_2]
\right)
\leq
\eta_2
\HPhi\left(
\E[\mt F| X_2]
\right).
\label{eq:matrix-Phi-SDPI-second-part-final}
\end{align}

Substituting
\eqref{eq:matrix-Phi-SDPI-first-part-final} and
\eqref{eq:matrix-Phi-SDPI-second-part-final} into
\eqref{eq:matrix-Phi-SDPI-output-chain}, we obtain
\begin{align}
\HPhi\left(
\E[\mt F| Y_1,Y_2]
\right)
&\leq
\eta_1
\HPhi(\mt F| X_2)
+
\eta_2
\HPhi\left(
\E[\mt F| X_2]
\right)
\nonumber\\
&\leq
\eta
\left(
\HPhi(\mt F| X_2)
+
\HPhi\left(
\E[\mt F| X_2]
\right)
\right)
\nonumber\\
&=
\eta
\HPhi(\mt F).
\label{eq:matrix-Phi-SDPI-product-upper-calculation}
\end{align}
Thus,
\begin{align} \label{eq:SDPI-tensorize-UB}
\eta_{\Phi}^{M,d}(X_1X_2;Y_1Y_2)
\leq
\max\{\eta_1,\eta_2\}.
\end{align}

For the reverse inequality, fix \(i\in\{1,2\}\) and let $\mt F_i:\cX_i\to\PSD^d$.
Suppose \(\mt F_i(X_i)\) is a function of \((X_1,X_2)\) that depends only on \(X_i\). Since $(X_1,Y_1)$ and $(X_2,Y_2)$
are independent, we have
\begin{align*}
\E\left[
\mt F_i(X_i)
|
Y_1,Y_2
\right]
=
\E\left[
\mt F_i(X_i)
|
Y_i
\right].
\end{align*}
Then
\begin{align*}
    \HPhi
    \left(
        \E\left[
    \mt F_i(X_i)
    |
    Y_i
    \right]
    \right)
    &=
    \HPhi
    \left(
        \E\left[
    \mt F_i(X_i)
    |
    Y_1 Y_2
    \right]
    \right)\\
    &\leq
    \eta_{\Phi}^{M,d}(X_1X_2;Y_1Y_2)
    \HPhi(\mt F_i(X_i)).
\end{align*}
Taking the supremum over admissible \(\mathbf F_i\) with positive entropy gives
\begin{align*}
    \eta_{\Phi}^{M,d}(X_i;Y_i)
    \leq
    \eta_{\Phi}^{M,d}(X_1X_2;Y_1Y_2).
\end{align*}
Taking the maximum over \(i\in\{1,2\}\) gives
\begin{align}
    \max\{\eta_1,\eta_2\}
    \leq
\eta_{\Phi}^{M,d}(X_1X_2;Y_1Y_2).
\label{eq:matrix-Phi-SDPI-tensor-lower}
\end{align}
Combining \eqref{eq:SDPI-tensorize-UB} and \eqref{eq:matrix-Phi-SDPI-tensor-lower}, we obtain
\begin{align*}
\eta_{\Phi}^{M,d}(X_1X_2;Y_1Y_2)
=
\max\left\{
\eta_{\Phi}^{M,d}(X_1;Y_1),
\eta_{\Phi}^{M,d}(X_2;Y_2)
\right\}.
\end{align*}
It finishes the proof.

\section{Proof of Theorem \ref{thm:seperation-SDPI}}
\label{app:proof-seperation-SDPI}

    Construct $p_{XY}$ in the following way.
    Let $U,V,J$ be independent $\operatorname{Bernoulli}(1/2)$ random variables. Define $X \defeq (U,V)$ and $Y \defeq (J,Z)$
    where
    \begin{align*}
    Z
    =
    \begin{cases}
    U, & J=0,\\
    V, & J=1.
    \end{cases}
    \end{align*}
The joint distribution is
\begin{align*}
p_{XY}
=
\frac{1}{8}
\begin{pmatrix}
1 & 0 & 1 & 0\\
1 & 0 & 0 & 1\\
0 & 1 & 1 & 0\\
0 & 1 & 0 & 1
\end{pmatrix}.
\end{align*}

Consider the smooth strictly convex function
\begin{align*}
\Phi(t)
\defeq
\frac{1}{2}t\sqrt{1+t^2}
+
t\ln\left(
\frac{t}{1+\sqrt{1+t^2}}
\right)
-
\frac{1}{2}\operatorname{arsinh}(t),
\qquad
t>0.
\end{align*}
It satisfies
\begin{align*}
\Phi''(t)
=
\frac{\sqrt{1+t^2}}{t},
\qquad
\frac{1}{\Phi''(t)}
=
\frac{t}{\sqrt{1+t^2}},
\end{align*}
and
\begin{align*}
\frac{d^2}{dt^2}
\left(
\frac{1}{\Phi''(t)}
\right)
=
-\frac{3t}{(1+t^2)^{5/2}}
<0.
\end{align*}
Hence $1/\Phi''$ is concave and thus $\Phi\in \mathscr{F}$.

For every positive scalar function $f=f(U,V)$,
\begin{align}
H_\Phi\left(\E[f| Y]\right)
&=
H_\Phi\left(\E[f| JZ]\right)\nonumber\\
&=\frac{1}{2}H_\Phi\left(\E[f| J=0,Z]\right)
+\frac{1}{2}H_\Phi\left(\E[f| J=1,Z]\right)\nonumber\\
&=
\frac{1}{2}
H_\Phi\left(\E[f| U]\right)
+
\frac{1}{2}
H_\Phi\left(\E[f| V]\right). \label{eq:Phi-YUV}
\end{align}
Since $U$ and $V$ are independent, by Proposition~\ref{prop:phi-ribbon-properties}~(ii), we have
\begin{align*}
H_\Phi\left(\E[f| U]\right)
+
H_\Phi\left(\E[f| V]\right)
\leq
H_\Phi(f).
\end{align*}
Combined with \eqref{eq:Phi-YUV}, we have  $H_\Phi\left(\E[f| Y]\right)\leq \frac12 H_\Phi(f)$. Thus,
\begin{align*}
\eta_\Phi^{M,1}(X;Y)
\leq
\frac{1}{2}.
\end{align*}
Taking any nonconstant positive function $f=f(U)$, we have the equality,
\begin{align*}
    \HPhi(\E[f|Y])
    &=\frac{1}{2}
    H_\Phi\left(\E[f| U]\right)
    +
    \frac{1}{2}
    H_\Phi\left(\E[f| V]\right)\\
    &=\frac{1}{2}
    H_\Phi\left(\E[f| U]\right)\\
    &=\frac{1}{2}
    H_\Phi\left(f\right).
\end{align*}
Hence
\begin{align}
\eta_\Phi^{M,1}(X;Y)
=
\frac{1}{2}.
\label{eq:scalar-Phi-SDPI-half}
\end{align}

Now define a $2\times 2$ positive definite matrix-valued function
$\mt F=\mt F(U,V)$ by
\begin{align*}
\mt F_{00}
&=
\begin{pmatrix}
0.51 & 0.23\\
0.23 & 0.87
\end{pmatrix},
&
\mt F_{01}
&=
\begin{pmatrix}
2.20 & 1.12\\
1.12 & 1.34
\end{pmatrix},
\\
\mt F_{10}
&=
\begin{pmatrix}
0.56 & -0.11\\
-0.11 & 3.06
\end{pmatrix},
&
\mt F_{11}
&=
\begin{pmatrix}
2.26 & 0.78\\
0.78 & 3.53
\end{pmatrix}.
\end{align*}
These matrices are not simultaneously diagonalizable. Using the normalized trace convention, direct evaluation gives
\begin{align*}
H_\Phi(\mt F)
&\approx
0.7136850541658435,
\\
H_\Phi\left(\E[\mt F| U]\right)
&\approx
0.3682717697171982,
\\
H_\Phi\left(\E[\mt F| V]\right)
&\approx
0.3454371869114984.
\end{align*}
In particular,
\begin{align*}
&
H_\Phi\left(\E[\mt F| U]\right)
+
H_\Phi\left(\E[\mt F| V]\right)
-
H_\Phi(\mt F)
\\
&\qquad\approx
2.3902462853\times 10^{-5}
>0.
\end{align*}
Since
\begin{align*}
H_\Phi\left(\E[\mt F| Y]\right)
=
\frac{1}{2}
H_\Phi\left(\E[\mt F| U]\right)
+
\frac{1}{2}
H_\Phi\left(\E[\mt F| V]\right),
\end{align*}
we obtain
\begin{align*}
\frac{
H_\Phi\left(\E[\mt F| Y]\right)
}{
H_\Phi(\mt F)
}
\geq
0.5000167458059
>
\frac{1}{2}.
\end{align*}
Consequently,
\begin{align}
\eta_\Phi^{M,2}(X;Y)
\geq
0.5000167458059
>
\frac{1}{2}
=
\eta_\Phi^{M,1}(X;Y).
\end{align}
Thus, the matrix $\Phi$-SDPI constant can depend strictly on the matrix dimension, even for a smooth strictly convex function $\Phi$ with concave $1/\Phi''$.

\section{Matrix $\Phi$-entropy inequalities}
\label{app:proof-matrix-SDPI-ribbon}

We restate Theorem~\ref{thm:matrix-Phi-conditional-stability}. 
\matrixPhiineq*

The following lemma will be useful. 
For $\mt J$ with spectrum in $\interior(\dom(\Phi))$ and
$\mt H\in\sa^d$, define the quadratic form 
\begin{align*}
    Q_{\mt J}(\mt H)
    &\defeq\nortr\left(\mt H \D\Phi'(\mt J)[\mt H]\right).
\end{align*}

\begin{lemma} \label{lem:matrix-quadratic-ub}
    Let $\Phi\in \mathscr{F}_M^d$. Suppose $\mt F:\cX\times \cY\to\PSD^d$ with spectra contained in $\interior(\dom(\Phi))$ and $\mu\defeq \min\{p_{XY}(x,y):p_{XY}(x,y)>0\}$. Write $\mt M\defeq\E[\mt F]$ and $\mt F_X\defeq\E[\mt F|X]$. Then
    \begin{equation}
        \E[Q_{\mt M}(\mt F-\mt M)]
        \leq \frac{2}{\mu}\HPhi(\mt F),
        \label{eq:mean-quadratic-bound}
    \end{equation}
    and
    \begin{equation}
        \E[Q_{\mt F_X}(\mt F-\mt F_X)]
        \leq \frac{2}{\mu}\HPhi(\mt F|X),
        \label{eq:mean-quadratic-bound-conditional}
    \end{equation}
\end{lemma}

\begin{proof}
If $\mu=1$, then $(X,Y)$ is constant and both
inequalities are trivial. Hence assume $\mu<1$.

We first prove \eqref{eq:mean-quadratic-bound}. For all $(x,y)\in \supp(p)$, define
\begin{align*}
    \mt V(x,y)
    \defeq
    \frac{p(x,y)-\mu}{1-\mu}\mt F(x,y)
    +
    \sum_{(x',y')\in\supp(p):(x',y')\neq(x,y)}
    \frac{p(x',y')}{1-\mu}\mt F(x',y').
\end{align*}
Since $p(x,y)\geq\mu$ for all $(x,y)\in \supp(p)$, the coefficients are nonnegative and sum to
one. Moreover,
\begin{align*}
    \mt M
    =
    \mu\mt F(x,y)+(1-\mu)\mt V(x,y).
\end{align*}
Therefore, for every $t\in[0,1]$,
\begin{align*}
    \mt M
    &=
    \mu\bigl((1-t)\mt M+t\mt F(x,y)\bigr)
    +(1-\mu)\bigl((1-t)\mt M+t\mt V(x,y)\bigr).
\end{align*}
Since $\Phi\in \mathscr{F}_M^d$, by definition, $(\mt U,\mt H)\mapsto Q_{\mt U}(\mt H)$ is jointly convex. 
Then by Jensen's inequality, for every
$\mt H\in\sa^d$, we have 
\begin{align*}
    Q_{\mt M}(\mt H)
    &\leq
    \mu
    Q_{(1-t)\mt M+t\mt F(x,y)}
    \left(\frac{\mt H}{\mu}\right)
    +(1-\mu)
    Q_{(1-t)\mt M+t\mt V(x,y)}(\mt 0)\\
    &=
    \frac{1}{\mu}
    Q_{(1-t)\mt M+t\mt F(x,y)}(\mt H),
\end{align*}
where the last equality follows from the fact that $Q_{\mt M}(a\mt H)=a^2 Q_{\mt M}(\mt H)$ for $a\geq 0$. Then 
\begin{align}
    Q_{(1-t)\mt M+t\mt F(x,y)}(\mt H)
    \geq
    \mu Q_{\mt M}(\mt H).
    \label{eq:quadratic-comparison-unconditional}
\end{align}

By Corollary~\ref{cor:integral-HPhi}, 
\begin{align*}
    \HPhi(\mt F)
    &=
    \E \left[
        \int_0^1(1-t)
        Q_{(1-t)\mt M+t\mt F}
        \bigl(\mt F-\mt M\bigr)\,dt
    \right]
    \\
    &\geq
    \mu
    \E \left[
        \int_0^1(1-t)
        Q_{\mt M}\bigl(\mt F-\mt M\bigr)\,dt
    \right]
    \\
    &=
    \frac{\mu}{2}
    \E\left[
        Q_{\mt M}(\mt F-\mt M)
    \right],
\end{align*}
where the second line follows from \eqref{eq:quadratic-comparison-unconditional}. 
It finishes the proof of 
\eqref{eq:mean-quadratic-bound}. The proof of \eqref{eq:mean-quadratic-bound-conditional} follows similarly.
\end{proof}

\begin{proof}[Proof of Theorem~\ref{thm:matrix-Phi-conditional-stability}]
If $\mu=1$, then $(X,Y)$ is constant and the inequality is trivial. Hence assume $0<\mu<1$. We first assume that all values
of $\mt F$ have spectra in $\interior(\dom(\Phi))$.
The boundary case will follow by perturbation and continuity.

Set
\[
    \mt M\defeq\E[\mt F],\qquad
    \mt F_X\defeq\E[\mt F|X],\qquad
    \mt F_Y\defeq\E[\mt F|Y],\qquad
    \widetilde{\mt F}_Y\defeq\E[\mt F_X|Y].
\]
Then
\begin{align*}
    \E[\mt F_X]=\E[\mt F_Y]&=\E[\widetilde{\mt F}_Y]=\mt M,\\
    \mt F_Y-\widetilde{\mt F}_Y&=\E[\mt F-\mt F_X|Y].
\end{align*}

By Lemma~\ref{lem:sa-psd-DPhi'}, $Q_{\mt U}(\mt H)\geq 0$.
Since $\Phi\in\mathscr F_M^d$, by definition, the mapping
$(\mt U,\mt H)\mapsto Q_{\mt U}(\mt H)$ is jointly convex. Equivalently in the definition, $(\mt U,\mt V)\mapsto \nortr\left(
    \Phi(\mt U)-\Phi(\mt V)-\Phi'(\mt V)(\mt U-\mt V)
    \right)$ is jointly convex. 
Applying Jensen's inequality to the jointly convex mapping $(\mt U,\mt V)\to \nortr\left(\Phi(\mt U)-\Phi(\mt V)-\Phi'(\mt V)(\mt U-\mt V)\right)$ 
with $(\mt U,\mt V)=(\mt F,\mt F_X)$ conditioned on $Y$,
and taking expectations, we obtain
\begin{align}
    \E\left[\nortr\left(
    \Phi(\mt F_Y)-\Phi(\widetilde{\mt F}_Y)-\Phi'(\widetilde{\mt F}_Y)(\mt F_Y-\widetilde{\mt F}_Y)
    \right)\right]\nonumber
    &\leq\E\left[\nortr\left(
    \Phi(\mt F)-\Phi(\mt F_X)-\Phi'(\mt F_X)(\mt F-\mt F_X)
    \right)\right]\nonumber\\
    &=\HPhi(\mt F|X),
    \label{eq:second-F|X}
\end{align}
where the last equality follows from
$\E[\mt F-\mt F_X|X]=\mt 0$ and the definition of conditional
matrix $\Phi$-entropy. Since $\E[\mt F_Y]=\E[\widetilde{\mt F}_Y]$, we have 
\begin{align}
    \HPhi(\mt F_Y)-\HPhi(\widetilde{\mt F}_Y)
    &=\E\left[\nortr \left(
        \Phi(\mt F_Y)-\Phi(\widetilde{\mt F}_Y)
    \right)
    \right]\nonumber\\
    &=\E\left[\nortr\left(
    \Phi(\mt F_Y)-\Phi(\widetilde{\mt F}_Y)-\Phi'(\widetilde{\mt F}_Y)(\mt F_Y-\widetilde{\mt F}_Y)+\Phi'(\widetilde{\mt F}_Y)(\mt F_Y-\widetilde{\mt F}_Y)
    \right)\right]
    \nonumber\\
    &\leq\HPhi(\mt F|X)
    +\E\left[\nortr\left(
    \Phi'(\widetilde{\mt F}_Y)(\mt F_Y-\widetilde{\mt F}_Y)
    \right)\right]\nonumber \\
    &=\HPhi(\mt F|X)
    +\E\left[\nortr\left(
    \bigl(\Phi'(\widetilde{\mt F}_Y)-\Phi'(\mt M)\bigr)(\mt F_Y-\widetilde{\mt F}_Y)
    \right)\right].
    \label{eq:matrix-Phi-stability-linear-term}
\end{align}
where the third line follows from \eqref{eq:second-F|X} and the last line follows from $\E \left[\nortr \left(\Phi'(\mt M) \left(\mt F_Y-\widetilde{\mt F}_Y\right)\right)\right]=0$. 
Then it remains to bound $\E\left[\nortr\left(
    \bigl(\Phi'(\widetilde{\mt F}_Y)-\Phi'(\mt M)\bigr)(\mt F_Y-\widetilde{\mt F}_Y)
    \right)\right]$. 

For $0<t<1$, define $\mt G_t(y)\defeq(1-t)\mt M+t\widetilde{\mt F}_Y(y)$ for each $y$. 
\begin{claim}
    For $0<t<1$, we have 
    \begin{align}
        \E\left[\left|\nortr\left(
        \D\Phi'(\mt G_t)[\widetilde{\mt F}_Y-\mt M](\mt F_Y-\widetilde{\mt F}_Y)
        \right)\right|\right]
        &\leq\frac{1}{\mu\sqrt{t(1-t)}}
        \sqrt{\HPhi(\mt F_X)\HPhi(\mt F|X)}.
        \label{eq:matrix-Phi-stability-mixed-Hessian}
    \end{align}
\end{claim}
\begin{proof}
    
Note that $p_X(x)\geq\mu$ and
$p_{Y|X}(y|x)=p_{XY}(x,y)/p_X(x)\geq\mu$ whenever
$p_{XY}(x,y)>0$. By Lemma~\ref{lem:sa-psd-DPhi'}, for fixed $\mt M$, $\mt H\to Q_{\mt M}(\mt H)$ is convex. Then by Jensen's inequality, we have
\begin{align}
    \E\left[Q_{\mt M}(\widetilde{\mt F}_Y-\mt M)\right]
    &=\E\left[Q_{\mt M}\bigl(\E[\mt F_X-\mt M|Y]\bigr)\right]
    \nonumber\\
    &\leq\E\left[Q_{\mt M}(\mt F_X-\mt M)\right]\nonumber\\
    &\leq\frac{2}{\mu}\HPhi(\mt F_X),
    \label{eq:matrix-Phi-stability-first-quadratic}
\end{align}
where the last line follows from Lemma~\ref{lem:matrix-quadratic-ub}. 
Similarly, by the joint convexity of $(\mt J,\mt H)\mapsto Q_{\mt J}(\mt H)$ and Jensen's inequality, we have 
\begin{align*}
    \E\left[Q_{\widetilde{\mt F}_Y}(\mt F_Y-\widetilde{\mt F}_Y)\right]
    &=\E\left[Q_{\E[\mt F_X|Y]}
    \bigl(\E[\mt F-\mt F_X|Y]\bigr)\right]\\
    &\leq\E\left[Q_{\mt F_X}(\mt F-\mt F_X)\right].
\end{align*}
Similarly, applying 
Lemma~\ref{lem:matrix-quadratic-ub} for each $x$ with $p_X(x)>0$, we have
\begin{align*}
    \E\left[Q_{\mt F_X(x)}(\mt F-\mt F_X(x))\,\middle|\,X=x\right]
    &\leq\frac{2}{\mu}\HPhi(\mt F|X=x).
\end{align*}
Averaging over $x$, we obtain
\begin{align}
    \E\left[Q_{\widetilde{\mt F}_Y}(\mt F_Y-\widetilde{\mt F}_Y)\right]
    &\leq\frac{2}{\mu}\HPhi(\mt F|X).
    \label{eq:matrix-Phi-stability-second-quadratic}
\end{align}

For every $r>0$, 
by the joint convexity of $(\mt U,\mt H)\mapsto Q_{\mt U}(\mt H)$ and Jensen's inequality, we have 
\begin{align}
    Q_{\mt G_t}\left(
    r(\widetilde{\mt F}_Y-\mt M)+ r^{-1}(\mt F_Y-\widetilde{\mt F}_Y)
    \right)
    &\leq (1-t) Q_{\mt M}\left(\frac{r}{1-t}(\widetilde{\mt F}_Y-\mt M)\right) 
    +t Q_{\widetilde{\mt F}_Y}\left(\frac{1}{rt}(\mt F_Y-\widetilde{\mt F}_Y)\right)\nonumber \\
    &=
    \frac{r^2}{1-t}Q_{\mt M}(\widetilde{\mt F}_Y-\mt M)
    +\frac{1}{r^2t}Q_{\widetilde{\mt F}_Y}(\mt F_Y-\widetilde{\mt F}_Y),
    \label{eq:matrix-Phi-stability-polarization-bound}
\end{align}
where the second line follows from $Q_{\mt J}(a\mt H)=a^2Q_{\mt J}(\mt H)$ for $a\geq 0$. 
Similarly, 
\begin{align}
    Q_{\mt G_t}\left(
    r(\widetilde{\mt F}_Y-\mt M)- r^{-1}(\mt F_Y-\widetilde{\mt F}_Y)
    \right)
    &\leq\frac{r^2}{1-t}Q_{\mt M}(\widetilde{\mt F}_Y-\mt M)
    +\frac{1}{r^2t}Q_{\widetilde{\mt F}_Y}(-\mt F_Y+\widetilde{\mt F}_Y)\nonumber \\
    &=\frac{r^2}{1-t}Q_{\mt M}(\widetilde{\mt F}_Y-\mt M)
    +\frac{1}{r^2t}Q_{\widetilde{\mt F}_Y}(\mt F_Y-\widetilde{\mt F}_Y)
    \label{eq:matrix-Phi-stability-polarization-bound2}
\end{align}
where the second line uses $Q_{\mt J}(\mt H)=Q_{\mt J}(-\mt H)$. 

By direct calculation, we have the following identity 
\begin{align*}
    4\nortr\left(
    \D\Phi'(\mt G_t)[\widetilde{\mt F}_Y-\mt M](\mt F_Y-\widetilde{\mt F}_Y)
    \right)&=Q_{\mt G_t}\left(
    r(\widetilde{\mt F}_Y-\mt M)+r^{-1}(\mt F_Y-\widetilde{\mt F}_Y)
    \right)\\
    &\quad-Q_{\mt G_t}\left(
    r(\widetilde{\mt F}_Y-\mt M)-r^{-1}(\mt F_Y-\widetilde{\mt F}_Y)
    \right).
\end{align*}
Since $Q_{\mt G_t}(\cdot)\geq 0$ and both items on the right are bounded above by the
right side of \eqref{eq:matrix-Phi-stability-polarization-bound} and \eqref{eq:matrix-Phi-stability-polarization-bound2}, we have 
\begin{align*}
    \bigg|4\nortr\left(
    \D\Phi'(\mt G_t)[\widetilde{\mt F}_Y-\mt M](\mt F_Y-\widetilde{\mt F}_Y)
    \right)\bigg|
    \leq 
    \frac{r^2}{1-t}Q_{\mt M}(\widetilde{\mt F}_Y-\mt M)
    +\frac{1}{r^2t}Q_{\widetilde{\mt F}_Y}(\mt F_Y-\widetilde{\mt F}_Y).
\end{align*}
Dividing both sides by 4 and taking expectations on both sides, we obtain
\begin{align*}
    \E\left[\left|\nortr\left(
    \D\Phi'(\mt G_t)[\widetilde{\mt F}_Y-\mt M](\mt F_Y-\widetilde{\mt F}_Y)
    \right)\right|\right]
    &\leq \frac{1}{4} \left(
        \frac{r^2}{1-t}
    \E\left[Q_{\mt M}(\widetilde{\mt F}_Y-\mt M)\right]
    +\frac{1}{r^2t}\E\left[Q_{\widetilde{\mt F}_Y}(\mt F_Y-\widetilde{\mt F}_Y)\right]
    \right)\\
    &\leq\frac{1}{2\mu}\left(
    \frac{r^2}{1-t}\HPhi(\mt F_X)
    +\frac{1}{r^2t}\HPhi(\mt F|X)\right)
\end{align*}
where the last line follows from Lemma~\ref{lem:matrix-quadratic-ub}. When $\HPhi(\mt F_X)$ and $\HPhi(\mt F|X)$ are positive, 
taking $r^2=\sqrt{\frac{(1-t)\HPhi(\mt F|X)}{t\HPhi(\mt F_X)}}$, the above inequality becomes 
\begin{align*}
    \E\left[\left|\nortr\left(
    \D\Phi'(\mt G_t)[\widetilde{\mt F}_Y-\mt M](\mt F_Y-\widetilde{\mt F}_Y)
    \right)\right|\right]
    &\leq\frac{1}{\mu\sqrt{t(1-t)}}
    \sqrt{\HPhi(\mt F_X)\HPhi(\mt F|X)}.
\end{align*}
When either $\HPhi(\mt F_X)$ or $\HPhi(\mt F|X)$ is zero, this inequality holds trivially. 
It finishes the proof of the claim. 
\end{proof}

By elementary calculus for matrices \cite{hiai2014introduction}, we have 
\begin{align*}
    \Phi'(\widetilde{\mt F}_Y)-\Phi'(\mt M)
    &=\int_0^1\D\Phi'(\mt G_t)[\widetilde{\mt F}_Y-\mt M]\,dt.
\end{align*}
Then 
\begin{align*}
    \left|\E\left[\nortr\left(
    \bigl(\Phi'(\widetilde{\mt F}_Y)-\Phi'(\mt M)\bigr)(\mt F_Y-\widetilde{\mt F}_Y)
    \right)\right]\right|
    &= \left|
    \E\left[\int_0^1\nortr\left(
    \D\Phi'(\mt G_t)[\widetilde{\mt F}_Y-\mt M](\mt F_Y-\widetilde{\mt F}_Y)
    \right)dt\right]\right|\\
    &\leq\int_0^1
    \E\left[\left|\nortr\left(
    \D\Phi'(\mt G_t)[\widetilde{\mt F}_Y-\mt M](\mt F_Y-\widetilde{\mt F}_Y)
    \right)\right|\right]dt\\
    &\leq \int_0^1\frac{1}{\mu\sqrt{t(1-t)}}
    \sqrt{\HPhi(\mt F_X)\HPhi(\mt F|X)}dt\\
    &=\frac{\pi}{\mu}
    \sqrt{\HPhi(\mt F_X)\HPhi(\mt F|X)},
\end{align*}
where the second line follows from the triangle inequality and exchange of expectation and integral, the third line follows from 
\eqref{eq:matrix-Phi-stability-mixed-Hessian} and the last line follows from $\int_0^1(t(1-t))^{-1/2}\,dt=\pi$.
Substituting into \eqref{eq:matrix-Phi-stability-linear-term},
we obtain
\begin{align}
    \HPhi(\mt F_Y)
    &\leq\HPhi(\widetilde{\mt F}_Y)+\HPhi(\mt F|X)
    +\frac{\pi}{\mu}\sqrt{\HPhi(\mt F_X)\HPhi(\mt F|X)}\nonumber\\
    &\leq\HPhi(\widetilde{\mt F}_Y)+\HPhi(\mt F|X)
    +\delta\HPhi(\mt F_X)
    +\frac{\pi^2}{4\delta\mu^2}\HPhi(\mt F|X)
    \label{eq:matrix-Phi-square-root-stability}
\end{align}
where the last line follows from Young's inequality. This proves the two inequalities for interior-valued $\mt F$. 

Finally, for general $\mt F$, choose
$c\in\interior(\dom(\Phi))$ and define
\begin{align*}
    \mt F_\epsilon\defeq(1-\epsilon)\mt F+\epsilon c\mt I_d,
    \qquad 0<\epsilon<1.
\end{align*}
Then all spectra of $\mt F_\epsilon$ lie in
$\interior(\dom(\Phi))$. 
Applying \eqref{eq:matrix-Phi-square-root-stability} to $\mt F_\epsilon$ and letting
$\epsilon\downarrow0$, by the continuity of $\Phi$ and finiteness of
the probability space, both inequalities hold for $\mt F$. It finishes the proof. 

\end{proof}

\section{Proof of Proposition \ref{prop:M-Phi-SDPI}} \label{app:proof-SDPI}

By perturbation and continuity, it suffices to consider matrix functions with spectra contained in $\interior(\dom(\Phi))$. 
Let $\mt F:\cX\to \sa^d$ be a matrix-valued function whose values have spectra contained in $\interior(\dom(\Phi))$. 
It suffices to show that 
\begin{align*}
    \lambda^2 \HPhi(\mt F) \geq \HPhi(\E[\mt F|Y]).
\end{align*}
Suppose $\mt F$ takes values from $\mt S_0$ and $\mt S_1$. Let $\mt M=\E[\mt F]$. Then we write $\mt S_0=\mt M-\mt Z$ and $\mt S_1=\mt M+\mt Z$. Thus we will show
\begin{align*}
    \lambda^2 \tr\left(\Phi(\mt M+\mt Z)+\Phi(\mt M-\mt Z)-2\Phi(\mt M)\right)\geq \tr\left(\Phi(\mt M+\lambda \mt Z)+\Phi(\mt M-\lambda \mt Z)-2\Phi(\mt M)\right).
\end{align*}
Let 
\begin{align*}
    \psi(t)= \tr\left(\Phi(\mt M+\sqrt{t}\mt Z)+\Phi(\mt M-\sqrt{t}\mt Z)-2\Phi(\mt M)\right).
\end{align*}
It is equivalent to showing that
\begin{align*}
    \lambda^2 \psi(1)\geq \psi(\lambda^2).
\end{align*}
Since $\psi(0)=0$, it suffices to show that $\psi(t)$ is convex in $t$. By the chain rule for the Fr\'echet derivative, we compute
\begin{align*}
    \psi'(t)=&\tr\left(\D\Phi(\mt M+\sqrt{t}\mt Z)\left[\frac{1}{2\sqrt{t}}\mt Z\right]+\D\Phi(\mt M-\sqrt{t}\mt Z)\left[-\frac{1}{2\sqrt{t}}\mt Z\right]
    \right). 
\end{align*}
And then
\begin{align*}
    \psi''(t)=\tr\left(
    \begin{aligned}
        &\D^2\Phi(\mt M+\sqrt{t}\mt Z)\left[\frac{1}{2\sqrt{t}}\mt Z,\frac{1}{2\sqrt{t}}\mt Z\right]
        +\D\Phi(\mt M+\sqrt{t}\mt Z)\left[-\frac{1}{4t^{3/2}}\mt Z\right]\\
        &+\D^2\Phi(\mt M-\sqrt{t}\mt Z)\left[-\frac{1}{2\sqrt{t}}\mt Z,-\frac{1}{2\sqrt{t}}\mt Z\right]
        +\D\Phi(\mt M-\sqrt{t}\mt Z)\left[\frac{1}{4t^{3/2}}\mt Z\right]
    \end{aligned}
    \right).
\end{align*}
Let $s=\sqrt{t}$. It is equivalent to showing that
\begin{align*}
    \tr\left(\D^2\Phi(\mt M+s\mt Z)[\mt Z,\mt Z]+\D^2\Phi(\mt M-s\mt Z)[-\mt Z,-\mt Z]\right)\geq
    \frac{1}{s}\tr\left(\D\Phi(\mt M+s\mt Z)[\mt Z]-\D\Phi(\mt M-s\mt Z)[\mt Z]\right).
\end{align*}
Let 
\begin{align*}
    \xi(s)=\tr\left(\D\Phi(\mt M+s\mt Z)[\mt Z]-\D\Phi(\mt M-s\mt Z)[\mt Z]\right).
\end{align*}
Thus it suffices to show that $\xi(s)$ is convex for $s\in [0,1]$. 
\begin{align*}
    \xi'(s)=\tr\left(\D^2\Phi(\mt M+s\mt Z)[\mt Z,\mt Z]+\D^2\Phi(\mt M-s\mt Z)[-\mt Z,-\mt Z]\right).
\end{align*}
\begin{align*}
    \xi''(s)=\tr\left(\D^3\Phi(\mt M+s\mt Z)[\mt Z,\mt Z,\mt Z]-\D^3\Phi(\mt M-s\mt Z)[\mt Z,\mt Z,\mt Z]\right).
\end{align*}
First we show that $\xi'(s)$ is convex in $s\in [0,1]$, i.e., 
\begin{align*}
    &\lambda \tr\left(\D^2\Phi(\mt M+s_1\mt Z)[\mt Z,\mt Z]+\D^2\Phi(\mt M-s_1\mt Z)[-\mt Z,-\mt Z]\right)\\
    +&(1-\lambda) \tr\left(\D^2\Phi(\mt M+s_2\mt Z)[\mt Z,\mt Z]+\D^2\Phi(\mt M-s_2\mt Z)[-\mt Z,-\mt Z]\right)\\
    \geq & \tr\left(\D^2\Phi\left(\mt M+\left(\lambda s_1 +(1-\lambda)s_2\right)\mt Z\right)[\mt Z,\mt Z]+\D^2\Phi\left(\mt M-\left(\lambda s_1 +(1-\lambda)s_2\right)\mt Z\right)[-\mt Z,-\mt Z]\right)
\end{align*}
for $0\leq \lambda \leq 1$ and $s_1,s_2\in [0,1]$. But this is implied by the convexity of 
\begin{align*}
    (\mt A,\mt B) \mapsto \tr\left(\D^2\Phi(\mt A)[\mt B,\mt B]\right), 
\end{align*}
which is an equivalent definition for $\mathscr{F}_M^d$ in Definition~\ref{def:class-FM}. 
Then, $\xi''(s)$ is increasing in $s$ and thus
\begin{align*}
    \xi''(s)\geq \xi''(0)=0
\end{align*}
for $s\in [0,1]$. Thus, $\xi(s)$ is convex. It concludes the proof.

\section{Proof of Theorem \ref{thm:tlogt-arbitrary-independence-converse}}
\label{app:proof-tlogt-arbitrary-independence-converse}

The following lemma is a consequence of Reed--Solomon codes \cite{reed1960polynomial,huffman2003fundamentals}. The proof is by a probabilistic argument, although there is an explicit construction by the Vandermonde matrix. 

\begin{lemma}\label{lem:ext-lnr-map}
    Let $\cE\subseteq 2^{[n]}$ be nonempty. 
    Let $c_i$, $i\in [n]$ be nonnegative integers and $r\defeq \max_{S\in \cE} \left\{\sum_{i\in S} c_i\right\}$. For a sufficiently large prime number $q$, there exist linear maps $\{L_i:\F_q^r\to \F_q^{c_i}\}_{i\in[n]}$ such that for any $S\in \cE$, the vertical concatenation of $L_i$ for $i\in S$ has full rank, i.e.,
    \[
    L_S:\mathbb F_q^r\to \prod_{i\in S}\mathbb F_q^{c_i},
    \qquad
    L_S(z)=(L_i z)_{i\in S},
    \]
    has rank
    \[
        \rank(L_S)=\sum_{i\in S}c_i=\min\left(\sum_{i\in S}c_i,r\right).
    \]
\end{lemma}
\begin{proof}
    Consider each $L_i$ as a $c_i\times r$ matrix in $\F_q$. Let all entries of $L_i$ be independently and uniformly distributed over $\F_q$. Given $S\in \cE$, let $t=\sum_{i\in S}c_i$. By definition $t\leq r$. Then 
    \begin{align*}
        \Pr \left[\rank(L_S)<\sum_{i\in S}c_i\right]
        &=1-\prod_{j=0}^{t-1}\left(1-\frac{1}{q^{r-j}}\right)\\
        &\leq 1- \left(1-\frac{1}{q^{r-t+1}}\right)^{t}\\
        &\leq \frac{t}{q^{r-t+1}}\\
        &\leq \frac{r}{q}.
    \end{align*}
    By the union bound, 
    \begin{align*}
        \Pr \left[\exists S\in \cE\,\rank(L_S)<\sum_{i\in S}c_i\right]
        &\leq \frac{|\cE|r}{q}.
    \end{align*}
    Take $q>|\cE|r$. Then there exist $\{L_i:\F_q^r\to \F_q^{c_i}\}_{i\in[n]}$ that satisfy the requirements. 
\end{proof}

\begin{proof}[Proof of Theorem \ref{thm:tlogt-arbitrary-independence-converse}]
    Since $\vec\lambda\notin \Conv(J_{\cE})$ and $\Conv(J_{\cE})$ is a compact
    convex set, by the separating hyperplane theorem, there exists 
    $\vec a\in\mathbb R^n$ such that 
    \begin{align}
        \sum_{i=1}^n a_i\lambda_i
        >
        \max_{S\in\cE}\left[\sum_{i\in S}a_i\right] .
    \end{align}
    We first show that we can replace $a_i$ with nonnegative integers. 
    Let $a_i^+=\max\{a_i,0\}$. Since $\vec \lambda$ is nonnegative, 
    \begin{align*}
        \sum_{i=1}^n a_i^+\lambda_i
        \geq
        \sum_{i=1}^n a_i\lambda_i.
    \end{align*}
We claim that
 \begin{align*}
        \max_{S\in\cE}\left[\sum_{i\in S}a_i^+\right]
        =
        \max_{S\in\cE}\left[\sum_{i\in S}a_i\right].
    \end{align*}
    
    Since $a_i^+\geq a_i$ for all $i\in [n]$, 
    \begin{align*}
        \max_{S\in\cE}\left[\sum_{i\in S}a_i^+\right]
        \geq
        \max_{S\in\cE}\left[\sum_{i\in S}a_i\right].
    \end{align*}
    On the other hand, let $T\in \cE$ be a maximizing subset for
    $\max_{S\in\cE}\left[\sum_{i\in S}a_i^+\right]$. Take
    $T'=\{i\in T:a_i\geq 0\}$. By the downward closure of $\cE$, $T'\in \cE$. Then
    \begin{align*}
        \max_{S\in\cE}\left[\sum_{i\in S}a_i^+\right]
        &=\sum_{i\in T}a_i^+\\
        &= \sum_{i\in T'}a_i \\
        &\leq \max_{S\in\cE}\left[\sum_{i\in S}a_i\right].
    \end{align*}
    Thus
    \[
        \max_{S\in\cE}\left[\sum_{i\in S}a_i^+\right]
        =
        \max_{S\in\cE}\left[\sum_{i\in S}a_i\right] .
    \]
    Combining this with the separating inequality, we obtain
    \begin{align*}
        \sum_{i=1}^n a_i^+\lambda_i
        > 
        \max_{S\in\cE}\left[\sum_{i\in S}a_i^+\right].
    \end{align*}
    Since the inequality is strict, there exist nonnegative rational numbers $\{b_i\}_{i\in[n]}$ such that 
    \begin{align*}
        \sum_{i=1}^n b_i\lambda_i
        > 
        \max_{S\in\cE}\left[\sum_{i\in S}b_i\right].
    \end{align*}
    By the proper scaling, there exist nonnegative integers $\{c_i\}_{i\in[n]}$ such that 
    \begin{align*}
        \sum_{i=1}^n c_i\lambda_i
        > 
        \max_{S\in\cE}\left[\sum_{i\in S}c_i\right].
    \end{align*}
    Let $r\defeq \max_{S\in\cE}\left[\sum_{i\in S}c_i\right]$. Then
    \[
        \sum_{i=1}^n c_i\lambda_i>r.
    \]

    By Lemma \ref{lem:ext-lnr-map}, there exist a prime $q$ and linear maps 
    \[
        L_i:\mathbb F_q^r\to \mathbb F_q^{c_i},
        \qquad i\in[n],
    \]
    such that for every $S\in\cE$, the combined map
    \[
        L_S:\mathbb F_q^r\to \prod_{i\in S}\mathbb F_q^{c_i},
        \qquad
        L_S(z)=(L_i (z))_{i\in S},
    \]
    has rank
    \[
        \rank(L_S)=\sum_{i\in S}c_i.
    \]

    Let $U$ be uniform on $\mathbb F_q^r$, and define $Y_i=L_i(U)$ for $i\in [n]$. Note that $L_i$ is a surjective map and hence each $Y_i$ is uniform on $\F_q^{c_i}$. Since $L_S$ has rank $\sum_{i\in S}c_i$, it guarantees that for any $S\in \cE$, $Y_S$ are mutually independent. We further construct random variables $X^n$ based on $Y^n$ such that $X_S$ are not mutually independent for $S\notin \cE$. Let 
    \begin{align*}
        \cM=\{T\notin \cE: \forall S\subsetneq T, S\in \cE\}. 
    \end{align*}
    Then $\cM$ consists of the minimal forbidden sets with respect to  $\cE$.  
    For each \(T\in\mathcal M\), let \((B_i^T)_{i\in T}\) be uniformly distributed on \(\{b\in\mathbb F_q^T:\sum_{i\in T}b_i=0\}\). Choose these vectors independently across \(T\), with their entire collection independent of \(U\). Then for any $T'\subsetneq T$, $(B_i^T)_{i\in T'}$ are mutually independent. Define $C_i=(B_i^T:T\in \cM,i\in T)$. 
    For $S\in \cE$, $T\cap S\subsetneq T$ for every $T\in \cM$ and thus $C_S$ are mutually independent. For $S\notin \cE$, there exists $T\in \cM$ such that $T\subseteq S$ and thus $C_S$ are not mutually independent. 
    
    Define $X_i=(Y_i,C_i)$. Then for every $S\in \cE$, $X_S$ are mutually independent.  For every $S\notin \cE$, $X_S$ are not mutually independent since $C_S$ are not. 

    For each $i\in [n]$, by the independence of $C_i$ and $U$,  
    \begin{align*}
        I(U;X_i)&=I(U;Y_iC_i)\\
        &=I(U;Y_i)\\
        &=H(Y_i)\\
        &=c_i\log q.
    \end{align*}
    Moreover, 
    \[
        I(U;X^n)\leq H(U)=r\log q.
    \]
    Therefore
    \[
    \begin{aligned}
        \sum_{i=1}^n \lambda_i I(U;X_i)-I(U;X^n)
        &\ge
        \left(\sum_{i=1}^n c_i\lambda_i-r\right)\log q\\
        &>0.
    \end{aligned}
    \]
    This completes the proof for the scalar HC ribbon. 
    
    Next, consider the matrix version. We first write the mutual information in the $\Phi$-entropy form. Let $f_u(x^n)=\frac{p(u,x^n)}{p(u)p(x^n)}$ for each $u\in \cU$. Then for $\Phi(t)=t\log t$, one can verify that 
    \begin{align*}
        \sum_{u}p(u)\HPhi(f_u)=I(U;X^n),
    \end{align*}
    and 
    \begin{align*}
        \sum_{u}p(u)\HPhi(\E[f_u|X_i])=I(U;X_i). 
    \end{align*}
    Then 
    \begin{align*}
        \sum_{i=1}^n \lambda_i \left( \sum_{u} p(u)\HPhi(\E[f_u|X_i])\right) 
        >
        \sum_{u}p(u)\HPhi(f_u). 
    \end{align*}
    Then there exists $u\in \cU$ such that 
    \begin{align*}
        \sum_{i=1}^n \lambda_i \HPhi(\E[f_u|X_i])
        >
        \HPhi(f_u). 
    \end{align*}
    Now let $\mt F(x^n)=f_u(x^n)I_d$ where $I_d$ is the identity matrix of dimension $d$. Then we have 
    \begin{align*}
        \HPhi(\mt F)=\HPhi(f_u),
    \end{align*}
    and 
    \begin{align*}
        \HPhi(\E[\mt F|X_i])=\HPhi(\E[f_u|X_i]). 
    \end{align*}
    Then 
    \begin{align*}
        \sum_{i=1}^n \lambda_i \HPhi(\E[\mt F|X_i])
        >
        \HPhi(\mt F). 
    \end{align*}
    Thus $\vec{\lambda}\notin \fR_{\Phi}^{M,d}(X_1;\dots;X_n)$. 
\end{proof}

\end{document}